\documentclass[aps,pra,twocolumn,nofootinbib, superscriptaddress,10pt]{revtex4-1}
\usepackage{graphicx}
\usepackage{epstopdf}
\usepackage{amsmath}
\usepackage{amssymb}
\usepackage{mathrsfs}
\usepackage{amsthm}
\usepackage{bm}
\usepackage{url}
\usepackage[T1]{fontenc}
\usepackage{csquotes}
\MakeOuterQuote{"}

\usepackage{color}

\newtheoremstyle{note}
  {\topsep/2}               % ABOVE SPACE
  {\topsep/2}               % BELOW SPACE
  {}                      % BODY FONT
  {\parindent}            % INDENT (empty value is the same as 0pt)
  {\itshape}              % HEAD FONT
  {.}                     % HEAD PUNCTUATION
  {5pt plus 1pt minus 1pt}% HEAD SPACE
  {}

\theoremstyle{note}
\newtheorem{theorem}{Theorem}
\newtheorem{lemma}{Lemma}

\newtheorem{corollary}{Corollary}
\newtheorem{proposition}{Proposition}

\theoremstyle{definition}

\theoremstyle{remark}

\newcommand{\tr}{\operatorname{tr}}

 \newcommand{\rmc}{\mathrm{c}}
 
 \newcommand{\rme}{\mathrm{e}}
 \newcommand{\rmi}{\mathrm{i}}

 \newcommand{\na}{\mathrm{NA}}

 \newcommand{\caF}{\mathcal{F}}
 \newcommand{\caH}{\mathcal{H}}

 \newcommand{\caN}{\mathcal{N}}

 \newcommand{\bbZ}{\mathbb{Z}}
 \newcommand{\nni}{\mathbb{Z}^{\geq0}}

 \newcommand{\id}{1}

 \newcommand{\bfk}{\mathbf{k}}
 \newcommand{\bfq}{\mathbf{q}}

 \newcommand{\scrS}{\mathscr{S}}

\newcommand{\be}{\begin{equation}}
\newcommand{\ee}{\end{equation}}
\newcommand{\ba}{\begin{align}}
\newcommand{\ea}{\end{align}}

\def\<{\langle}  %% overriding the original command \<
\def\>{\rangle}  %% overriding the original command \>
\newcommand{\ket}[1]{| #1\>}

\def\eqref#1{\textup{(\ref{#1})}}  %% overriding the original command \eqref
\newcommand{\eref}[1]{Eq.~\textup{(\ref{#1})}}
\newcommand{\Eref}[1]{Equation~\textup{(\ref{#1})}}
\newcommand{\esref}[1]{Eqs.~\textup{(\ref{#1})}}
\newcommand{\Esref}[1]{Equations~\textup{(\ref{#1})}}

\newcommand{\fref}[1]{Fig.~\ref{#1}}

\newcommand{\fsref}[1]{Figs.~\ref{#1}}

\newcommand{\tref}[1]{Table~\ref{#1}}

\newcommand{\sref}[1]{Sec.~\ref{#1}}
\newcommand{\Sref}[1]{Section~\ref{#1}}
\newcommand{\ssref}[1]{Secs.~\ref{#1}}

\newcommand{\thref}[1]{Theorem~\ref{#1}}
\newcommand{\Thref}[1]{Theorem~\ref{#1}}
\newcommand{\thsref}[1]{Theorems~\ref{#1}}
\newcommand{\Thsref}[1]{Theorems~\ref{#1}}

\newcommand{\lref}[1]{Lemma~\ref{#1}}
\newcommand{\Lref}[1]{Lemma~\ref{#1}}
\newcommand{\lsref}[1]{Lemmas~\ref{#1}}
\newcommand{\Lsref}[1]{Lemmas~\ref{#1}}

\newcommand{\pref}[1]{Proposition~\ref{#1}}
\newcommand{\Pref}[1]{Proposition~\ref{#1}}

\newcommand{\crref}[1]{Corollary~\ref{#1}}
\newcommand{\Crref}[1]{Corollary~\ref{#1}}

\newcommand{\cref}[1]{Conjecture~\ref{#1}}
\newcommand{\Cref}[1]{Conjecture~\ref{#1}}

\newcommand{\aref}[1]{Appendix~\ref{#1}}

\newcommand{\rcite}[1]{Ref.~\cite{#1}}
\newcommand{\rscite}[1]{Refs.~\cite{#1}}

\begin{document}
\title{General framework for verifying pure quantum states in the adversarial scenario}

\author{Huangjun Zhu}
\email{zhuhuangjun@fudan.edu.cn}

\affiliation{Department of Physics and Center for Field Theory and Particle Physics, Fudan University, Shanghai 200433, China}

\affiliation{State Key Laboratory of Surface Physics, Fudan University, Shanghai 200433, China}

\affiliation{Institute for Nanoelectronic Devices and Quantum Computing, Fudan University, Shanghai 200433, China}

\affiliation{Collaborative Innovation Center of Advanced Microstructures, Nanjing 210093, China}

\author{Masahito Hayashi}
\affiliation{Graduate School of Mathematics, Nagoya University, Nagoya, 464-8602, Japan}

\affiliation{Shenzhen Institute for Quantum Science and Engineering,
Southern University of Science and Technology,
%No.1088 Xueyuan Avenue,Nanshan District,
Shenzhen,
518055, China}
\affiliation{Center for Quantum Computing, Peng Cheng Laboratory, Shenzhen 518000, China}
\affiliation{Centre for Quantum Technologies, National University of Singapore, 3 Science Drive 2, 117542, Singapore}

\begin{abstract}
Bipartite and multipartite entangled states are of central interest in quantum information processing and foundational studies.
Efficient verification of these states, especially in the adversarial scenario,  is a key to various applications,  including quantum computation, quantum simulation, and quantum networks. However, little is known about this topic in the adversarial scenario.  Here we initiate a systematic study of pure-state verification in the adversarial scenario. In particular, we  introduce a general method for determining the minimal number of tests required by a given strategy to achieve a given precision. In the case of homogeneous strategies, we can even derive an analytical formula. Furthermore, we propose a general recipe to verifying pure quantum states in  the adversarial scenario by virtue of protocols for the nonadversarial scenario. Thanks to this recipe, the resource cost for verifying an arbitrary pure state in the adversarial scenario is comparable to the counterpart for the nonadversarial scenario, and the overhead is at most three times for high-precision verification.  Our recipe can readily be applied to efficiently verify bipartite pure states, stabilizer states, hypergraph states, weighted graph states, and Dicke states in the adversarial scenario, even if  only  local projective measurements are accessible. This paper is an extended version of the companion paper Zhu and Hayashi, Phys. Rev. Lett. \textbf{123}, 260504 (2019).
\end{abstract}

\date{\today}
\maketitle

\section{Introduction}
Quantum states encode all the information about a quantum system and play a central role  in quantum information processing. For example, bipartite entangled states, especially maximally entangled states, are crucial to quantum teleportation, dense coding, and quantum cryptography \cite{HoroHHH09,GuhnT09}.  Multipartite  entangled states, such as graph states \cite{HeinEB04} and hypergraph states \cite{KrusK09,QuWLB13,RossHBM13,SteiRMG17,XionZCC18}
 are especially useful in  (blind) measurement-based quantum computation (MBQC) \cite{RausB01,RausBB03,BroaFK09,MoriF13,HayaM15,FujiH17,HayaH18,TakeMH19,
 	MillM16,MoriTH17,GachGM19}, quantum error correction \cite{Gott97the,SchlW01},  quantum networks \cite{PersLCL13,MccuPBM16,MarkK18},
 and foundational studies \cite{GreeHSZ90,ScarASA05,GuhnTHB05,GachBG16}. Another important class of 
 multipartite   states, including Dicke states \cite{Dick54,HaffHRB05},  are useful in  quantum metrology  \cite{PezzSOS18}. Furthermore, multipartite  states, such as tensor-network states,  also have extensive applications in  research areas beyond quantum information science, including condensed matter physics \cite{VersMC08,Orus14}.

To unleash the potential of multipartite quantum states in quantum information processing, it is paramount to prepare and verify these states with high precision using limited resources. To verify quantum states with traditional tomography \cite{PariR04Book}, however, the resource required  increases exponentially with the number of qubits.  Although  compressed sensing \cite{GrosLFB10} and direct fidelity estimation (DFE) \cite{FlamL11} can improve the efficiency, the exponential scaling behavior cannot be changed in general. As another alternative, self-testing \cite{HayaH18,MayeY04,SupiB19} is also quite resource consuming although it is conceptually appealing from the perspective of device independence.

Recently, a powerful approach known as quantum state verification (QSV) has attracted increasing attention \cite{HayaMT06,Haya09G,AoliGKE15,TakeM18, PallLM18}. 
It is particularly effective  in  extracting the key information---the fidelity with the target state. So far
efficient or even optimal verification protocols based on local projective  measurements have been constructed for bipartite pure states \cite{HayaMT06,Haya09G,PallLM18,ZhuH19O,LiHZ19,WangH19,YuSG19}, Greenberger-Horne-Zeilinger (GHZ) states \cite{LiHZ19O}, stabilizer states (including graph states) \cite{HayaM15,FujiH17,HayaH18, PallLM18,MarkK18,ZhuH19E}, hypergraph states \cite{ZhuH19E}, weighted graph states \cite{HayaT19}, and Dicke states \cite{LiuYSZ19}. Moreover, the efficiency of this approach has been demonstrated in experiments \cite{ZhanCPX19}.

However, the situation is much more troublesome when we turn to the adversarial scenario, in which the quantum states of interest  are controlled by an untrusted party, Eve. 
Efficient QSV in such adversarial scenario is crucial
to many applications in  quantum information processing  that require high security conditions, including blind MBQC \cite{MoriF13,HayaM15,FujiH17,HayaH18,TakeMH19} and quantum networks \cite{PersLCL13,MccuPBM16,MarkK18}. 
Unfortunately, no efficient approach is known for addressing such adversarial scenario in general.
For example, to verify the simplest nontrivial hypergraph states (say of three qubits) already requires an astronomical number of measurements \cite{MoriTH17,TakeM18}. What is worse,
 little is known about the 
resource cost of a given verification strategy to achieve a given precision \cite{TakeM18,TakeMMM19}. As a consequence, no general guideline is known for constructing an efficient verification strategy or for comparing the efficiencies of different strategies.

In this paper we  initiate a systematic study of pure-state verification in the adversarial scenario. In particular, we introduce a general method for determining the minimal number of tests required by a given verification strategy to achieve a given precision. We also introduce the concept of homogeneous strategies, which play a key role in  QSV. Thanks to their high symmetry, we can derive analytical formulas for most figures of merit of practical interest. The conditions for single-copy verification are also clarified. Furthermore,  we provide a general recipe to constructing efficient verification protocols for the adversarial scenario from verification protocols for the nonadversarial scenario.
By virtue of this recipe, we can verify
  pure quantum states  in the adversarial scenario with nearly the same efficiency as
 in the nonadversarial scenario. For high-precision verification, the overhead in the number of tests is at most three times. In this way,  pure-state verification in the adversarial scenario can be greatly simplified since it suffices to focus on the nonadversarial scenario and then apply our recipe. In addition, our study  reveals that entangling measurements are less helpful and often unnecessary in  improving the verification  efficiency  in the adversarial scenario, which is counterintuitive at first sight.

Our work is especially helpful to the verification of bipartite pure states \cite{HayaMT06,Haya09G,PallLM18,ZhuH19O,LiHZ19,WangH19,YuSG19}, GHZ states \cite{LiHZ19O}, stabilizer states (including graph states) \cite{PallLM18,ZhuH19E}, hypergraph states \cite{ZhuH19E}, weighted graph states \cite{HayaT19}, and Dicke states \cite{LiuYSZ19}, for which efficient verification protocols for the nonadversarial scenario have been constructed recently. By virtue of our recipe, all these states can  be verified in the adversarial scenario with much higher efficiencies than was possible previously; moreover, only local projective measurements are required to achieve high efficiencies. For bipartite pure states, GHZ states, and qudit stabilizer states,  even optimal protocols can be constructed using local projective measurements \cite{LiHZ19,LiHZ19O}; see \sref{sec:app}.

This paper is an extended version of the companion paper \cite{ZhuH19AdS}\footnote{This work was originally motivated by the verification of qubit and qudit hypergraph states and is contained as a part of the preprint arXiv:1806.05565 (cf.~\rcite{ZhuH19E}). However, the general framework of QSV in the adversarial scenario we developed applies to all pure states, not only to hypergraph states.
To discuss this topic comprehensively, we finally decided to present these results independently.
}.

The rest of this paper is organized as follows. In \sref{sec:SetStage}, we review the basic framework of QSV in the nonadversarial scenario. In \sref{sec:QSVre}, we clarify the limitation of previous approaches to QSV and motivate the current study. In \sref{sec:QSVadvStart}, we formulate the general ideal of QSV in the adversarial scenario and introduce the main figures of merit. In \sref{sec:CompMFM}, we introduce a general method for computing the main figures of merit in  the adversarial scenario. 
In \sref{sec:Homo}, we discuss in detail
 QSV with homogeneous strategies. In \sref{eq:SingleCopy}, we clarify the power of a single test in  QSV.
In \sref{sec:NumTestGen}, we 
determine the minimal number of tests required by a general verification strategy to achieve a given precision.
 In \sref{sec:recipe}, we 
 propose a general recipe to constructing efficient verification protocols for the adversarial scenario from  protocols devised for the nonadversarial scenario. In \sref{sec:app}, we demonstrate the power of our recipe via its applications  to many important bipartite and multipartite quantum states. In \sref{sec:comparison}, we compare QSV with a number of other approaches for estimating or verifying quantum states.
\Sref{sec:Sum} summarizes this paper. 
To streamline the presentation, most technical proofs are relegated to the appendix.

\section{\label{sec:SetStage}Setting the stage}

In this section we first review the basic framework of QSV in the nonadversarial scenario. The main results presented here were established by Pallister, Linden, and Montanaro (PLM) \cite{PallLM18}, but we have simplified the derivation. These results will serve as a benchmark for understanding pure-state verification in the adversarial scenario, which is the main focus of this paper. Then we 
discuss the connection between QSV and fidelity estimation.

\subsection{\label{sec:QSVna}Verification of pure states: Nonadversarial scenario}

Consider a device that is supposed to  produce the target state $|\Psi\>$ in the (generally multipartite) Hilbert space ${\cal H}$.
In practice, the device may actually produce $\sigma_1, \sigma_2, \ldots, \sigma_{N}$ in $N$ runs. Following \rcite{PallLM18}, here 
we assume that
 the fidelity $\<\Psi|\sigma_j|\Psi\>$ either equals 1  for all $j$ or satisfies $\<\Psi|\sigma_j|\Psi\>\leq 1- \epsilon$
 for all $j$ (the limitation of  this assumption will be analyzed in \sref{sec:QSVre}).
Now the task is to
determine which is the case.

To achieve this task we can perform $N$ tests and accept the states produced if and only if (iff) all tests are passed. Each test is specified by a two-outcome  measurement $\{E_l, 1-E_l\}$
chosen randomly from a set of  accessible measurements.  The test operator $E_l$ corresponds to passing the test and satisfies the condition $0\leq E_l\leq1$. We assume that the target state $|\Psi\>$ can always pass the test, that is, $E_l|\Psi\>=|\Psi\>$ for each $E_l$. A verification strategy is characterized 
by all the tests $E_l$ and the probabilities $\mu_l$ for performing these tests.

To determine the maximal probability of failing to reject the bad case, it is convenient to introduce  the verification operator $\Omega:=\sum_l \mu_l E_l$.
As we shall see later, most key properties of a verification strategy are determined by the verification operator $\Omega$, irrespective of how the test operators are constructed. Therefore, $\Omega$ is also referred to as a strategy when there is no danger of confusion.
By construction, the target state $|\Psi\>$ is an eigenstate of $\Omega$ with the largest  eigenvalue 1. 
Denote by $\beta(\Omega)$  the second largest eigenvalue of $\Omega$, then  $\beta(\Omega)$ is equal to the operator norm of $\Omega-|\Psi\>\<\Psi|$, that is, $\beta(\Omega)=\|\Omega-|\Psi\>\<\Psi|\|$.
Let  $\nu(\Omega):=1-\beta(\Omega)$ be the spectral gap from the largest eigenvalue. 
When $\<\Psi|\sigma_j|\Psi\>\leq 1-\epsilon$, 
the maximum probability that $\sigma_j$ can pass a test on average
is given by
\begin{equation}\label{eq:PassingProb}
\max_{\<\Psi|\sigma|\Psi\>\leq 1-\epsilon }\tr(\Omega \sigma)=1- [1-\beta(\Omega)]\epsilon=1- \nu(\Omega)\epsilon,
\end{equation}
where the maximization in the left-hand side runs over all quantum states $\sigma$ that satisfy the fidelity constraint $\<\Psi|\sigma|\Psi\>\leq 1-\epsilon$. \Eref{eq:PassingProb} was originally derived by PLM \cite{PallLM18} for strategies composed of projective tests, but their proof also applies to general strategies with nonprojective tests; see  \aref{sec:PassProb} for a simpler proof.

After $N$ runs, $\sigma_j$ in the bad case can pass all tests with probability at most $[1-\nu(\Omega)\epsilon]^N$. This is also the maximum probability that the verification strategy fails to detect the bad case. To achieve significance level $\delta$ (confidence level $1-\delta$), that is, $[1-\nu(\Omega)\epsilon]^N\leq \delta$,
the minimum number of tests is given by
\cite{PallLM18}
\begin{equation}\label{eq:NumTest}
N_\na(\epsilon,\delta,\Omega)=\biggl\lceil
\frac{\ln\delta}{\ln[1-\nu(\Omega)\epsilon]}\biggr\rceil
\leq 
\biggl\lceil
\frac{\ln\delta^{-1}}{\nu(\Omega)\epsilon}\biggr\rceil,
\end{equation}
where NA in the subscript means nonadversarial. 
This number is the main figure of merit of concern in QSV because to a large extent it determines the resource costs of implementing the verification strategy $\Omega$.
Note that a single test is sufficient if 
\begin{equation}\label{eq:OneTestCon1}
\nu(\Omega)\epsilon+\delta\geq 1. 
\end{equation}

According to \eref{eq:NumTest}, the efficiency of the strategy $\Omega$ is determined by the spectral gap $\nu(\Omega)$. 
The optimal protocol is obtained by maximizing the spectral gap $\nu(\Omega)$. If there is no restriction on the accessible measurements, then the optimal protocol is composed of the projective measurement $\{|\Psi\>\<\Psi|, \id-|\Psi\>\<\Psi|\}$, in which case we have $\Omega=|\Psi\>\<\Psi|$ and $\nu(\Omega)=1$, so that 
\begin{equation}
N_\na(\epsilon,\delta,\Omega)=\biggl\lceil
\frac{\ln\delta}{\ln(1-\epsilon)}\biggr\rceil
\leq 
\biggl\lceil
\frac{\ln\delta^{-1}}{\epsilon}\biggr\rceil.
\end{equation}
In addition,  the requirement in  \eref{eq:OneTestCon1} reduces to 
\begin{equation}\label{eq:OneTestCon2}
\epsilon+\delta\geq 1. 
\end{equation}
This efficiency cannot be improved further even if we can perform collective measurements. In particular, the scaling behaviors of $\epsilon^{-1}\ln\delta^{-1}$ with $\epsilon$ and $\delta$ are the best we can expect.

In practice, quite often the target state $|\Psi\>$ is  entangled, but it is not easy to perform   entangling measurements. It is therefore crucial to devise efficient verification  protocols based on local operations and classical communication (LOCC). Here by "efficient" we mean that the protocols can be applied in practice with reasonable resource costs, which is a much stronger requirement than what is usually understood in computer science.
 Ideally, the inverse spectral gap $1/\nu(\Omega)$ should be independent of the system size (the number of qubits say) or grow no faster than a low-order polynomial. In addition, the coefficients should be reasonably small. It turns out  many important quantum states in quantum information processing can be verified efficiently with respect to these stringent criteria.
Besides the total number $N$ of tests determined by $\epsilon, \delta$, and $\nu(\Omega)$, the number  of potential measurement settings is also of concern if it is difficult to switch measurement settings. Nevertheless, most of our results in \ssref{sec:SetStage}-\ref{sec:recipe} are independent of the specific details (including the number  of potential measurement settings) of a verification protocol   once the verification operator is fixed.

Here, we compare the approach presented above with previous works \rscite{HayaMT06,Haya09G}.
In mathematical statistics, we often discuss hypothesis testing in the framework of uniformly most powerful test
among a certain class of tests.
In this case, we fix a certain set of states ${\cal S}_0$,
and impose to our test  the condition that the probability of erroneously rejecting states in  
${\cal S}_0$ is upper bounded by a certain value $\delta'\ge 0$.
Under this condition, we maximize the probability of detecting a state $\sigma$ in ${\cal S}^\rmc$, where ${\cal S}^\rmc$ is the complement of ${\cal S}_0$ in the state space.
When a test maximizes the probability uniformly for every state $\sigma$ in ${\cal S}^\rmc$,
it is called a uniformly most powerful (UMP) test.
However,
since the detecting probability depends on the state $\sigma$,
such a test does not exist in general.
In this paper, ${\cal S}_0$ and $\delta'$  are chosen to be $\{ |\Psi\>\<\Psi|\} $ and $0$, respectively.
We consider the case in which the same strategy $\Omega$ is applied $N$ times. 
Since we support the state $|\Psi\>$ only when all our outcomes correspond to the pass eigenspace of 
$\Omega$,
our test is UMP in  this case.

When the set ${\cal S}_0$ is chosen to be $\{ \sigma\, |\, \< \Psi |\sigma |\Psi\> \ge 1-\epsilon' \}$,
and $\delta'$ is a non-zero value, the problem is more complicated.
Such a setting arises when we allow a certain amount of error. 
To resolve this problem, 
imposing a certain symmetric condition to our tests,
\rscite{HayaMT06,Haya09G} discussed several optimization problems
and investigated their asymptotic behaviors
when $|\Psi\>$ is a maximally entangled state.

\subsection{Connection with fidelity estimation}
When all states $\sigma_j$ produced by the device are identical to $\sigma$,  let $F=\<\Psi|\sigma|\Psi\>$ be the fidelity between $\sigma$ and the target state $|\Psi\>$;
then we have 
\begin{equation}\label{eq:FidelityPassingProb1}
[1-\tau(\Omega)]F+\tau(\Omega) \leq \tr(\Omega \sigma)\leq \nu(\Omega)F+\beta(\Omega),
\end{equation}
where $\tau(\Omega)$ is the smallest eigenvalue of $\Omega$. Therefore,
\begin{equation}\label{eq:FidelityPassingProb2}
 \frac{1-\tr(\Omega\sigma)}{1-\tau(\Omega)}\leq 1-F\leq\frac{1-\tr(\Omega\sigma)}{\nu(\Omega)} .
\end{equation}
So the passing probability $\tr(\Omega \sigma)$  provides  upper and lower bounds for the infidelity (and fidelity).  In general, \esref{eq:FidelityPassingProb1} and \eqref{eq:FidelityPassingProb2}  still hold if $F$ and $\tr(\Omega \sigma)$ are replaced by their averages over all $\sigma_j$. Note that the inequalities in  \esref{eq:FidelityPassingProb1} and \eqref{eq:FidelityPassingProb2} are saturated when $\tau(\Omega)=\beta(\Omega)$; 
such a strategy $\Omega$ is called \emph{homogeneous} and is discussed in more detail in \sref{sec:Homo}. 
In this case we have
\begin{equation}\label{eq:FEhom}
1-F=\frac{1-\tr(\Omega\sigma)}{\nu(\Omega)},\quad F=\frac{\tr(\Omega\sigma)-\beta(\Omega)}{\nu(\Omega)}.
\end{equation}
So the fidelity with the target state can be estimated from the passing probability.
The standard deviation of this estimation reads
\begin{equation}\label{eq:Fstd}
\Delta F=\frac{\sqrt{p(1-p)}}{\nu\sqrt{N}}=\frac{\sqrt{(1-F)(F+\nu^{-1}-1)}}{\sqrt{N}}
\leq \frac{1}{2\nu \sqrt{N}}, 
\end{equation}
where $p=\tr(\Omega\sigma)=\nu F+\beta\geq F$ and $N$ is the number of tests performed.  Note that this standard deviation decreases monotonically with $\nu$ and $N$. This conclusion is related to the testing of binomial distributions discussed in 
\rcite{Haya09G}. When $F\geq 1/2$, which is the case of most interest, we also have 
\begin{equation}
\Delta F=\frac{\sqrt{p(1-p)}}{\nu\sqrt{N}}\leq \frac{\sqrt{F(1-F)}}{\nu \sqrt{N}}
\end{equation}
given that $p\geq F$.

\section{\label{sec:QSVre}Verification of pure states: A critical reexamination}
In this section we reexamine the framework
of QSV proposed by PLM  \cite{PallLM18} as summarized in \sref{sec:QSVna} above and clarify the limitation of this framework.  In addition, we show that the limitation can be eliminated when states prepared in different runs are independent. The situation is much more complicated when these states are correlated, which motivates the study of QSV in the adversarial scenario presented in the rest of this paper. 

\subsection{What is verified in QSV?}
Consider a device that is supposed to  produce the target state $|\Psi\>$ in the Hilbert space ${\cal H}$. In practice, the device may actually produce $\sigma_1, \sigma_2, \ldots, \sigma_{N}$ in $N$ runs. In the framework of PLM, it is assumed that
the fidelity $\<\Psi|\sigma_j|\Psi\>$  either equals 1  for all $j$ or satisfies $\<\Psi|\sigma_j|\Psi\>\leq 1- \epsilon$
for all $j$ \cite{PallLM18}. In the independent and identically distributed (i.i.d.) case, all $\sigma_j$ are identical, so the PLM assumption is actually not necessary (or automatically guaranteed) to derive the conclusions presented in \sref{sec:QSVna}. If we drop the i.i.d. assumption, then the assumption of PLM is quite unnatural and difficult to guarantee. Moreover, the conclusion on QSV drawn based on this assumption is much weaker than what the word "verify" usually conveys. Suppose the test $E_l$ is performed with probability $\mu_l$ and $\Omega=\sum_l\mu_l E_l$ as in \sref{sec:QSVna}. 
After  $N$ tests are passed, we can only conclude that the probability of passing $N$ tests is at most $[1-\nu(\Omega)\epsilon]^N$ if $\<\Psi|\sigma_j|\Psi\>\leq 1- \epsilon$
for all $j$. In other words, passing these tests only confirms that $\<\Psi|\sigma_j|\Psi\>> 1- \epsilon$
for at least one run $j$ with significance level $[1-\nu(\Omega)\epsilon]^N$. Such a weak conclusion is  usually far from enough in practice. Note that the property of each run on average is more relevant if we want to make sure that the device works as expected most of the time rather than occasionally.

\subsection{\label{sec:QSVind}Independent state preparation}
Fortunately we can drop the PLM assumption
and draw a stronger conclusion as long as all states $\sigma_j$ are prepared independently of each other. Note that we do not need the i.i.d. assumption. The variation in $\sigma_j$ 
over different runs may be caused by inevitable imperfections of the device  or fluctuations in various relevant parameters for example.

\begin{proposition}\label{pro:PassProbNT}
Suppose the $N$ states $\sigma_1, \sigma_2, \ldots, \sigma_{N}$ are  independent of each other. Then the probability that they can pass all $N$ tests associated with the strategy $\Omega$ satisfies
\begin{equation}
\prod_{j=1}^N \tr(\Omega\sigma_j)\leq [1-\nu(\Omega)\bar{\epsilon}]^N, 
\end{equation}
where $\bar{\epsilon}=\sum_j\epsilon_j/N$ with $\epsilon_j=1-\<\Psi|\sigma_j |\Psi\>$ is the average infidelity. 
\end{proposition}
This proposition guarantees that the average fidelity satisfies the inequality  $\sum_j \<\Psi|\sigma_j |\Psi\>/N> 1-\epsilon$ with significance level $\delta=[1-\nu(\Omega)\epsilon]^N$ if $N$ tests are passed. In addition, to verify $|\Psi\>$ within infidelity $\epsilon$ and significance level $\delta$, which means $[1-\nu(\Omega)\epsilon]^N\leq \delta$,  the minimum number of tests reads
\begin{equation}\label{eq:NumTestAF}
N_\na(\epsilon,\delta,\Omega)=\biggl\lceil
\frac{\ln\delta}{\ln[1-\nu(\Omega)\epsilon]}\biggr\rceil
\leq 
\biggl\lceil
\frac{\ln\delta^{-1}}{\nu(\Omega)\epsilon}\biggr\rceil.
\end{equation}
This formula is identical to the one in \eref{eq:NumTest}, but it does not rely on the unnatural assumption imposed by PLM \cite{PallLM18}. Accordingly,
 the meaning of "verification" is different. Here we can verify the average fidelity of the states $\sigma_1, \sigma_2,\ldots \sigma_N$ prepared by the device rather than the maximal fidelity. 
Nevertheless, our conclusion relies on the implicit assumption that the average fidelity of states produced by the device 
after the verification procedure is the same as the average during the verification procedure. This assumption is reasonable in the nonadversarial scenario and is often 
taken for granted in practice. In case this assumption does not hold, then we have to consider QSV in the adversarial scenario, which is a main  focus of this paper.

\begin{proof}[Proof of \pref{pro:PassProbNT}]
\begin{align}
\prod_{j=1}^N \tr(\Omega\sigma_j)\leq \prod_{j=1}^N [1-\nu(\Omega)\epsilon_j]\leq [1-\nu(\Omega)\bar{\epsilon}]^N.
\end{align}
Here the first inequality follows from \eref{eq:PassingProb} and is saturated iff each $\sigma_j$ is supported in the subspace associated with the largest and second largest eigenvalues of $\Omega$. The second inequality follows from the familiar inequality between the geometric mean and arithmetic mean and is saturated iff all $\epsilon_j$ are equal to $\bar{\epsilon}$; that is, all $\sigma_j$ have the same fidelity (and infidelity) with the target state.  Note that variation in $\sigma_j$ cannot increase the passing probability once the average infidelity $\bar{\epsilon}$ is fixed. 
\end{proof}

\subsection{\label{sec:Correlated}Correlated state preparation}
Here we show that the conclusion in \ssref{sec:QSVna} and \ref{sec:QSVind} will fail if  the states $\sigma_1, \sigma_2, \ldots, \sigma_N$ are correlated. As a special example, suppose the device produces the ideal target state $(|\Psi\>\<\Psi|)^{\otimes N}$ in $N$ runs with probability $0<a<1$ and the alternative $\sigma^{\otimes N}$ with probability $1-a$, where $\<\Psi|\sigma|\Psi\>=1-\epsilon'<1$.  The reduced state of each party reads $a(|\Psi\>\<\Psi|)+(1-a)\sigma$ and its infidelity with the target state is $\epsilon=(1-a)\epsilon'$. Note that
the device can pass $N$ tests  with probability  at least $a$ no matter how large $N$ is. So it is impossible to verify the target state within infidelity $\epsilon=(1-a)\epsilon'$ and significance level $\delta<a$ using the approach presented in \sref{sec:QSVna}  or that in \sref{sec:QSVind}. This observation further reveals the limitation of the PLM framework of QSV. To overcome this difficulty, we need to consider a different framework of QSV as formulated in the next section.

\section{\label{sec:QSVadvStart}Quantum state verification  in the adversarial  scenario}	
Now we turn to the  adversarial scenario in which the device for generating quantum states is controlled by a potentially malicious adversary. In this case the device may produce  arbitrary correlated or even entangled states. Efficient verification of quantum states in such adversarial scenario is crucial to many  tasks in quantum information processing that entail  high security requirements, such as blind quantum computation \cite{MoriF13,HayaM15,FujiH17,HayaH18,TakeMH19} and quantum networks \cite{PersLCL13,MccuPBM16,MarkK18}. However, little is known about this topic in the literature. The approach of PLM does not apply as illustrated by the example of correlated state preparation in \sref{sec:Correlated}.  
Most other studies in the literature only focus on specific families of states, such as graph states \cite{HayaM15,FujiH17,HayaH18,MarkK18} and hypergraph states
\cite{MoriTH17,TakeM18}.
In addition, known protocols are too resource consuming to be applied in practice,  especially for hypergraph states, in  which case the best protocol known in the literature  requires an astronomical number of tests already for
three-qubit hypergraph states. 
The difficulty in constructing efficient verification protocols in the adversarial scenario is tied to the fact that even for a given protocol,  no efficient method is available for determining the minimal resource cost necessary to reach the target precision.

In this section we introduce a general framework of pure state verification in the adversarial scenario together with the main figures of merit. The basic ideas presented here will serve as a stepping stone for the following study.

\subsection{Formulation}
To establish a reliable and efficient framework for verifying pure states in the adversarial scenario, first note that the verification and application of a quantum state cannot be completely separated  in the adversarial scenario. Otherwise, the device may produce ideal target states in the verification stage and so can always pass the tests, but produce a garbage state in the application stage. 
To resolve this problem, suppose the device  produces an arbitrary correlated or entangled state $\rho$ on the whole system  ${\cal H}^{\otimes (N+1)}$. Our goal is to ensure that the reduced state on one  system (for application) has infidelity less than $\epsilon$ by performing $N$ tests on other systems.
We can randomly choose $N$ systems and apply a verification  strategy $\Omega$  to each   system chosen  and accept the state on the remaining system iff all $N$ tests are passed.
Since $N$ systems are chosen randomly, we  may assume that $\rho$ is permutation invariant without loss of generality.

Suppose the strategy  $\Omega$ is applied to the first $N$ systems, then  the probability that $\rho$ can   pass $N$ tests reads
\begin{equation}
p_\rho=\tr[(\Omega^{\otimes N}\otimes \id) \rho].
\end{equation}
If $N$ tests are passed, then 
the reduced state  on system $N+1$ (assuming $p_\rho>0$) is given by
\begin{equation}
\sigma'_{N+1}=p_\rho^{-1}\tr_{1,2,\ldots, N}[(\Omega^{\otimes N}\otimes \id) \rho],
\end{equation}
where $\tr_{1,2,\ldots, N}$ means the partial trace over the systems $1,2,\ldots, N$. The fidelity between $\sigma'_{N+1}$ and the target state $|\Psi\>$  reads
\begin{equation}
F_\rho=\<\Psi|\sigma'_{N+1}|\Psi\>=p_\rho^{-1} f_\rho,
\end{equation}
where
\begin{equation}
f_\rho=\tr[(\Omega^{\otimes N}\otimes |\Psi\>\<\Psi|) \rho].
\end{equation}

When $\rho=\sigma^{\otimes (N+1)}$ is a tensor power of the state $\sigma$  with  $0<\epsilon'=1-\<\Psi|\sigma|\Psi\><1$, we have $p_\rho\leq [1-\nu(\Omega)\epsilon']^N$, $\sigma'_{N+1}=\sigma$, and $F_\rho=1-\epsilon'$. These conclusions coincide with the counterpart for the nonadversarial scenario  as expected. The situation is different if $\rho$ does not have this form.
Suppose $\rho=a (|\Psi\>\<\Psi|)^{\otimes (N+1)} +(1-a)\sigma^{\otimes (N+1)}$ with $0<a<1$ for example; cf.~\sref{sec:Correlated}.
If $N$ tests are passed, then the reduced state of party $N+1$ reads
\begin{equation}
\sigma'_{N+1}=\frac{a|\Psi\>\<\Psi|+b\sigma}{a+b},
\end{equation}
where $b:=(1-a)[\tr(\Omega\sigma)]^N$ satisfies
\begin{equation}
b\leq (1-a)[1-\nu(\Omega) \epsilon']^N
\end{equation}
and decreases exponentially with $N$ unless $\tr(\Omega\sigma)=0$. Therefore, the infidelity $1-\<\Psi|\sigma'_{N+1}|\Psi\>$ approaches zero exponentially with $N$ even if $a$ is arbitrarily small. If the infidelity is bounded from below  $1-\<\Psi|\sigma'_{N+1}|\Psi\>\geq \epsilon$
for $0<\epsilon<1$, then $a$ should approach zero as $N$ increases; accordingly, the passing probability will approach zero. This observation indicates that we can verify the target state within any given infidelity $0<\epsilon<1$ and significance level $0<\delta<1$ even when the states prepared are correlated, which demonstrates the advantage of the  alternative approach presented above over the PLM approach. In the rest of this paper we will show that indeed 
it is possible to verify pure states efficiently even if the device is controlled by the  adversary and can produce arbitrary correlated or even entangled states allowed by quantum mechanics.

\subsection{Main figures of merit}
To  characterize the performance of the strategy $\Omega$ adapted to the adversarial scenario, here we introduce four figures of merit. Define
\begin{subequations}\label{eq:Ffdelta}
\begin{align}
\zeta(N,\delta,\Omega)
&:=\min_{\rho}
\big\{f_\rho
\,|\,p_\rho \ge \delta
\big\}, \quad  0\leq \delta\leq 1, \label{eq:Minfdelta}\\
\eta(N,f,\Omega)
&:=\max_{\rho}
\big\{
p_\rho
\,|\,f_\rho \le f
\big\}, \quad  0\leq f\leq 1, 
\label{eq:Maxprhof}\\
F(N,\delta,\Omega)
&:=\min_{\rho}
\big\{
p_\rho^{-1}f_\rho
\,|\,p_\rho \ge \delta
\big\}, \quad  0<\delta\leq 1,  \label{eq:MinFdelta}\\
\caF(N,f,\Omega)
&:=\min_{\rho}
\big\{
p_\rho^{-1}f_\rho
\,|\,f_\rho \ge f
\big\},\quad  0<f\leq 1, 
\label{eq:MinFf}
\end{align}
\end{subequations}
where $N\geq 1$ is the number of tests performed and the   minimization or maximization is taken over permutation-invariant quantum states $\rho$ on $\caH^{\otimes (N+1)}$. The four figures of merit are closely related to each other, as we shall see later.  
In practice $F(N,\delta,\Omega)$  is a main figure of merit of interest; it denotes the minimum fidelity of the reduced state on the remaining party (with the target state), assuming that  $\rho$ can pass $N$ tests with significance level at least $\delta$. 
By definition $F(N,\delta,\Omega)$ and $\zeta(N,\delta,\Omega)$ are nondecreasing in $\delta$, while $\caF(N,f,\Omega)$ and  $\eta(N,f,\Omega)$ are nondecreasing in $f$.
A simple upper bound for $F(N,\delta,\Omega)$ can be derived by considering quantum states $\rho$ on ${\cal H}^{\otimes (N+1)}$ that can be expressed as tensor powers in \eref{eq:MinFdelta}, which yields
\begin{equation}\label{eq:FidnaBound}
F(N,\delta,\Omega)\leq \max\biggl\{0, 1-\frac{1-\delta^{1/N}}{\nu(\Omega)} \biggr\}. 
\end{equation}

The four figures of merit defined in \eref{eq:Ffdelta} are tied to  the two-dimensional region $R_{N,\Omega}$ composed of all the points  $(p_\rho, f_\rho)$ for permutation-invariant  density matrices $\rho$ on ${\cal H}^{\otimes (N+1)}$, that is,
\begin{equation}\label{eq:ConvexHull}
\{(p_\rho, f_\rho)|   \rho \mbox{ on ${\cal H}^{\otimes (N+1)}$  are permutation invariant}\}. 
\end{equation}
This geometric picture will be very helpful to understanding QSV in the adversarial scenario. By definition the region $R_{N,\Omega}$ is convex since the state space is convex, and $p_\rho, f_\rho$ are both linear in $\rho$. What is not so obvious at the moment is that the region  $R_{N,\Omega}$ is actually  a convex polygon.

In addition to characterizing
the verification precision that is achievable for a given number $N$ of tests, it is equally  important to determine the minimum number of tests required to reach a given precision. To this end, for $0<\epsilon,\delta<1$, 
we define $N(\epsilon,\delta,\Omega)$ as the minimum value of the positive integer $N$ that satisfies the condition $F(N,\delta,\Omega)\geq 1-\epsilon$, namely
\begin{equation}\label{eq:MinNumTestDef}
N(\epsilon,\delta,\Omega):=\min\{N\geq 1 \,| \, F(N,\delta,\Omega)\geq 1-\epsilon \} . 
\end{equation}
Then  \eref{eq:FidnaBound} implies that 
\begin{equation}
N(\epsilon,\delta,\Omega)\geq \biggl\lceil
\frac{\ln\delta}{\ln[1-\nu(\Omega)\epsilon]}\biggr\rceil=N_\na(\epsilon,\delta,\Omega)
\end{equation}
as expected since it is much more difficult to verify a quantum state in the adversarial scenario than nonadversarial scenario.
How much overhead is required  in the adversarial scenario? Can we achieve the same scaling behaviors in $\epsilon$ and $\delta$?

 In general it is very difficult to derive an analytical formula for $N(\epsilon,\delta,\Omega)$ if not impossible. Therefore, it is nontrivial to determine the efficiency limit of QSV in the adversarial scenario even if there is no restriction on the accessible measurements, or even if the target state belongs to a single party, which is in sharp contrast with QSV in the nonadversarial scenario. Indeed, it took a long time and a lot of efforts to settle  this issue.

\section{\label{sec:CompMFM}Computation of the main figures of merit} 
In this section we develop a general method for computing the figures of merit defined in \eref{eq:Ffdelta}, which characterize the verification precision in the adversarial scenario.
We also clarify the properties of these figures of merit in preparation for later study. Both algebraic derivation and geometric pictures will be helpful in our analysis.

\subsection{Key observations}
Suppose the verification operator $\Omega$ for the target state $|\Psi\>\in \caH$ has spectral decomposition $\Omega=\sum_{j=1}^D \lambda_j \Pi_j $, where $\lambda_j$ are the eigenvalues of $\Omega$ arranged in decreasing order $1=\lambda_1> \lambda_2\geq\cdots \geq \lambda_D\geq 0$,  and  $\Pi_j$ are mutually orthogonal rank-1 projectors with $\Pi_1=|\Psi\>\<\Psi|$. Here the
 second largest eigenvalue $\beta:=\lambda_2$ and the smallest eigenvalue $\tau:=\lambda_D$
 deserve special attention because they determine the performance of $\Omega$ to a large extent, as we shall see later.  Suppose the adversary produces the 
state $\rho$ on the whole system  ${\cal H}^{\otimes (N+1)}$, which is  permutation invariant (cf.~\sref{sec:QSVadvStart}). Without loss of generality, we may assume that $\rho$ is diagonal in the product basis constructed from the eigenbasis of $\Omega$ (as determined by the projectors $\Pi_j$), since  $p_\rho$, $f_\rho$, and $F_\rho$  only depend on  the diagonal elements of $\rho$.

Let $\bfk=(k_1,k_2, \ldots,k_D)$ be a sequence of $D$ nonnegative integers that sum up to $N+1$, that is, $\sum_j k_j =N+1$. Let $\scrS_N$ be the set of all such sequences. For each $\bfk\in \scrS_N$, we can define a permutation-invariant diagonal density matrix $\rho_\bfk$ on $\caH^{\otimes(N+1)}$ as the uniform mixture of all permutations of $\Pi_1^{\otimes k_1}\otimes \Pi_2^{\otimes k_2}\otimes\cdots\otimes \Pi_D^{k_D}$. Then any permutation-invariant diagonal density matrix $\rho$ on $\caH^{\otimes(N+1)}$ can be
expressed  as $\rho=\sum_{\bfk\in \scrS_N} c_{\bfk} \rho_{\bfk}$,  where $c_\bfk$ form a probability distribution on $\scrS_N$. Accordingly,
\begin{align}
p_\rho&=\sum_{\bfk\in \scrS_N} c_\bfk \eta_\bfk(\bm{\lambda}),\quad
f_\rho=\sum_{\bfk\in \scrS_N} c_\bfk \zeta_\bfk(\bm{\lambda}),\label{eq:prhofrho}\\
F_\rho&=\frac{f_\rho}{p_\rho}=\frac{\sum_{\bfk\in \scrS_N} c_\bfk \zeta_\bfk(\bm{\lambda})}{\sum_{\bfk\in \scrS_N} c_\bfk \eta_\bfk(\bm{\lambda})},
\end{align}
where $\bm{\lambda}:=(\lambda_1,\lambda_2,\ldots,\lambda_D)$ and
\begin{equation}\label{eq:etazeta}
\begin{aligned}
\eta_\bfk(\bm{\lambda})&:=p_{\rho_\bfk}=
\sum_{i|k_i>0}\frac{k_i}{(N+1)}\lambda_i^{k_i-1}\prod_{j\neq i | k_j>0} \lambda_j^{k_j},\\
\zeta_\bfk(\bm{\lambda})&:=f_{\rho_\bfk}=\frac{k_1}{N+1}\prod_{i|k_i>0}\lambda_i^{k_i}.
\end{aligned}
\end{equation}
Here we set $\lambda_i^0=1$  even if $\lambda_i=0$. 

The assumption $1=\lambda_1> \lambda_2\geq\cdots \geq \lambda_D=\tau\geq 0$ implies that $\zeta_\bfk(\bm{\lambda})\leq \eta_\bfk(\bm{\lambda})\leq1 $;  the second inequality is saturated iff $\bfk=\bfk_0:=(N+1,0,\ldots, 0)$, in which case both inequalities are saturated, that is, $\zeta_{\bfk_0}(\bm{\lambda})= \eta_{\bfk_0}(\bm{\lambda})=1$. As an implication, we have $f_\rho\leq p_\rho\leq 1$, and the second inequality is saturated iff $\rho=\rho_{\bfk_0}=(|\Psi\>\<\Psi|)^{\otimes (N+1)}$, in which case $f_\rho=p_\rho=1$. This observation implies that 
\begin{align}
F(N,\delta=1,\Omega)=
\zeta(N,\delta=1,\Omega)=1,\label{eq:Fidsig1} \\
\caF(N,f=1,\Omega)=
\eta(N,f=1,\Omega)=1.
\end{align}
By contrast, $\eta_\bfk(\bm{\lambda})\geq \tau^N$, and the lower bound is saturated when $\bfk=(0,\ldots,0, N+1)$.
Accordingly, $p_\rho\geq \tau^N$, and the lower bound is saturated when $\rho=\Pi_D^{\otimes(N+1)}$.

\begin{figure}
	\includegraphics[width=8.6cm]{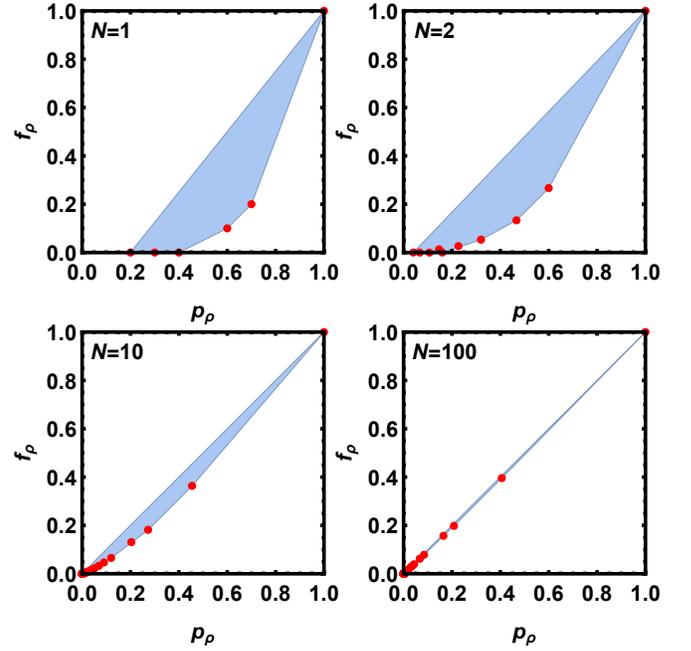}
	\caption{\label{fig:ConvexHull}(color online) The region $R_{N,\Omega}$ composed of $(p_\rho, f_\rho)$ as defined in \eref{eq:ConvexHull}. This region is the convex hull of points
	 $(\eta_\bfk(\bm{\lambda}),\zeta_\bfk(\bm{\lambda}))$ for $\bfk\in \scrS_N$, which are highlighted as red dots. Here $\Omega$ has three distinct eigenvalues, namely, $1$, $0.4$, and $0.2$. }
\end{figure}

In view of the above discussion, the region $R_{N,\Omega}$ defined in \eref{eq:ConvexHull} is the convex hull of $(\eta_\bfk(\bm{\lambda}), \zeta_\bfk(\bm{\lambda}))$ for all $\bfk\in \scrS_N$,  which is a polygon, as illustrated in \fref{fig:ConvexHull}. It should be emphasized  that $R_{N,\Omega}$ only depends on the distinct eigenvalues of $\Omega$, but not on their degeneracies (though $\lambda_1$ is not degenerate by assumption). The same conclusion also applies to the figures of merit
 $F(N,\delta,\Omega)$,
$\caF(N,f,\Omega)$, $
\zeta(N,\delta,\Omega)$, and $
\eta(N,f,\Omega)$ defined in \eref{eq:Ffdelta} given that they are completely determined by the region $R_{N,\Omega}$. 
For example, $\zeta(N,\delta,\Omega)$ corresponds to the lower boundary of the intersection of $R_{N,\Omega}$ and the vertical line $p_\rho=\delta$ as long as $\delta\geq \tau^{N}$ (cf. \lref{lem:equiTwoDef} below). 
 This geometric picture is very helpful to understanding the properties of $F(N,\delta,\Omega)$, although in general it is not easy to find an explicit analytical formula.   As $N$ increases, the region $R_{N,\Omega}$ concentrates more and more around the diagonal defined by the equation $f=p$ as illustrated in \fref{fig:ConvexHull}, which means  $F(N,\delta,\Omega)$ approaches 1 as $N$ increases.
 
Denote by $\sigma(\Omega)$ the set of distinct eigenvalues of $\Omega$. If $\Omega'$ is another verification operator for $|\Psi\>$ with $\beta(\Omega')<1$ and  $\sigma(\Omega')\subset \sigma(\Omega)$, then $R_{N,\Omega'}\subset R_{N,\Omega}$ and  $\Omega'$ is  equally  efficient or more efficient  than $\Omega$ in the sense that
\begin{equation}\label{eq:StrategyOrder}
F(N,\delta,\Omega')\geq F(N,\delta,\Omega),\; N(\epsilon,\delta,\Omega')\leq N(\epsilon,\delta,\Omega). 
\end{equation}
This observation is instructive to constructing efficient verification protocols, as we shall see in \sref{sec:Homo}.

\subsection{Computation of the verification precision}
 
Here we show that the four figures of merit $\zeta(N,\delta,\Omega)$, $
\eta(N,f,\Omega)$, $F(N,\delta,\Omega)$,
and $\caF(N,f,\Omega)$ can be computed by linear programming.  \Lsref{lem:etaf0} and \ref{lem:equiTwoDef} below are proved in \aref{sec:zetaEtaProofs}. To start with,
we  first determine $\eta(N,0,\Omega)$, the maximum of $p_\rho$ under the condition $f_\rho=0$.
 \begin{lemma}\label{lem:etaf0}$\eta(N,0,\Omega)=\delta_\rmc$, where
\begin{equation}\label{eq:deltarmc}
\delta_\rmc:=\begin{cases}
\beta^N, & \tau>0,\\
\max\{\beta^N,1/(N+1)\}, & \tau=0.
\end{cases}
\end{equation}
 \end{lemma}
\Lref{lem:etaf0} has  implications for the figures of merit $F(N,\delta,\Omega)$ and $\zeta(N,\delta,\Omega)$ as well,
\begin{align}
F(N,\delta,\Omega)=\zeta(N,\delta,\Omega)=0,\quad 0< \delta\leq \delta_\rmc, \label{eq:zetazero}\\
F(N,\delta,\Omega)>0, \quad \zeta(N,\delta,\Omega)>0, \quad \delta_\rmc<\delta\leq 1. \label{eq:zetanonzero}
\end{align}
The equality $\zeta(N,\delta,\Omega)=0$ also holds when $\delta=0$.

Next, we introduce alternative definitions of the figures of merit defined in \eref{eq:Ffdelta}, which are easier to analyze and compute.
Define
\begin{subequations}\label{eq:FfdeltaA}
	\begin{align}
	\tilde{\zeta}(N,\delta,\Omega)
	&:=\!\begin{cases}\min_{\rho}
	\big\{f_\rho
	\,|\,p_\rho= 
	\delta\bigr\}, &  \delta_\rmc\leq  \delta\leq 1, \\
	0,& 0\leq \delta\leq \delta_\rmc, 
	\end{cases}   \label{eq:MinfdeltaA}\\
	\tilde{\eta}(N,f,\Omega)
	&:=\!\max_{\rho}
	\big\{
	p_\rho
	\,|\,f_\rho = f
	\big\}, \quad 0\leq f\leq 1,
	\label{eq:MaxprhofA}\\
	\tilde{F}(N,\delta,\Omega)
	&:=\delta^{-1}\tilde{\zeta}(N,\delta,\Omega), \quad  0<  \delta\leq 1, 
	\label{eq:MinFdeltaA}\\	
	\tilde{\caF}(N,f,\Omega)
	&:=[\tilde{\eta}(N,f,\Omega)]^{-1}f, \quad 0<f\leq 1.
	\label{eq:MinFfA}
	\end{align}
\end{subequations} 
Here $\delta_\rmc$ in \eref{eq:MinfdeltaA} can be replaced by $\tau^N$ given that $\min_{\rho}
\big\{f_\rho\,|\,p_\rho= \delta\bigr\}=0$ for $\tau^N\leq \delta \leq \delta_\rmc$.

\begin{lemma}\label{lem:equiTwoDef}
Suppose $N$ is a positive integer and $\Omega$ is a verification operator. Then
\begin{subequations}
 \begin{align}
 \zeta(N,\delta,\Omega)&=\tilde{\zeta}(N,\delta,\Omega), \quad 0\leq   \delta\leq 1,\\
 \eta(N,f,\Omega)&=\tilde{\eta}(N,f,\Omega),\quad  0\leq   f\leq 1,\\
F(N,\delta,\Omega)&=\tilde{F}(N,\delta,\Omega), \quad 0<  \delta\leq 1,\\
\caF(N,f,\Omega)&=\tilde{\caF}(N,f,\Omega), \quad 0<  f\leq 1.
\end{align}
\end{subequations}
\end{lemma}
For $0<\delta,f\leq 1$,  \lref{lem:equiTwoDef} implies that
\begin{subequations}
\begin{align}
F(N,\delta,\Omega)&=\delta^{-1}\tilde{\zeta}(N,\delta,\Omega)=\delta^{-1}\zeta(N,\delta,\Omega),\label{eq:Fzeta} \\  
\caF(N,f,\Omega)&=[\tilde{\eta}(N,f,\Omega)]^{-1}f=[\eta(N,f,\Omega)]^{-1}f. 
\end{align}
\end{subequations}
To compute  $F(N,\delta,\Omega)$ and $\caF(N,f,\Omega)$, it suffices to compute $\zeta(N,\delta,\Omega)$ and $\eta(N,f,\Omega)$. By virtue of \eref{eq:prhofrho}  and \lref{lem:equiTwoDef},
  $\zeta(N,\delta,\Omega)$ with $ \delta_\rmc\leq \delta\leq 1$ and $\eta(N,f,\Omega)$ with $0\leq f\leq 1$ can be computed via linear programming, 
\begin{subequations}\label{eq:zetaetaLP}
\begin{align}
\zeta(N,\delta,\Omega)&=
\min_{\{c_\bfk\}}
\Biggl\{\sum_{\bfk\in \scrS_N}
c_\bfk \zeta_\bfk(\bm{\lambda})
\Bigg|\sum_{\bfk\in \scrS_N} c_\bfk \eta_\bfk(\bm{\lambda})=\delta
\Biggr\}, \label{eq:zetaLP}  \\
\eta(N,f,\Omega)&=
\max_{\{c_\bfk\}}
\Biggl\{
\sum_{\bfk\in \scrS_N}c_\bfk \eta_\bfk(\bm{\lambda})
\Bigg|\sum_{\bfk\in \scrS_N} c_\bfk \zeta_\bfk(\bm{\lambda})=f
\Biggr\}, \label{eq:etaLP} 
\end{align}
\end{subequations}
where $c_\bfk$ form a probability distribution on $\scrS_N$.
Here the minimum in \eref{eq:zetaLP}
 can be attained at a distribution $\{c_\bfk\}$ that is supported on 
at most two points in $\scrS_N$; a similar conclusion holds for the maximum in \eref{eq:etaLP}. These conclusions are tied to the geometric fact that any boundary point of  $R_{N,\Omega}$ lies on a line segment that connects two extremal points. 
This observation can greatly simplify the computation of $F(N,\delta,\Omega)$ and $\caF(N,f,\Omega)$ as well as $\zeta(N,\delta,\Omega)$ and $\eta(N,f,\Omega)$. 
In addition to the computational value, \eref{eq:zetaetaLP}  implies that $\zeta(N,\delta,\Omega)$ and $\eta(N,f,\Omega)$ are piecewise linear functions, whose turning points correspond to the extremal points of the region $R_{N,\Omega}$ and have the form $(\eta_\bfk(\bm{\lambda}),\zeta_\bfk(\bm{\lambda}))$ for some $\bfk\in \scrS_N$; cf.~\lref{lem:ExtremalPoints} in \aref{sec:zetaEtaProofs}.

\subsection{Properties of the main figures of merit}
Next, we summarize the main  properties of the five figures of merit $\zeta(N,\delta,\Omega)$, $
\eta(N,f,\Omega)$, $F(N,\delta,\Omega)$,
 $\caF(N,f,\Omega)$, and $N(\epsilon,\delta,\Omega)$; the proofs are relegated to \aref{sec:zetaEtaProofs}.
These properties are tied to the fact that the region $R_{N,\Omega}$ is a convex polygon.

\begin{lemma}\label{lem:zetaEtaMonoCon}
The following statements hold.
	\begin{enumerate}
\item 	$\zeta(N,\delta,\Omega)$ is convex and nondecreasing in $\delta$ for $0\leq \delta\leq 1$ and  is strictly  increasing  for  $\delta_\rmc\leq \delta\leq 1$.

\item  $\eta(N,f,\Omega)$ is concave and  strictly increasing in $f$ for $0\leq f\leq 1$. 

\item 	$F(N,\delta,\Omega)$ is nondecreasing in $\delta$ for $0< \delta\leq 1$ and is strictly  increasing  for $\delta_\rmc\leq \delta\leq 1$.

\item 	$\caF(N,f,\Omega)$  is strictly  increasing  in $f$ for  $0<f\leq 1$.

\end{enumerate}
\end{lemma}

\begin{lemma}\label{lem:etazetaMU}
Suppose $0\leq \delta,f\leq 1$. Then 
\begin{subequations}
\begin{align}
\eta(N,\zeta(N,\delta,\Omega),\Omega)&=\max\{\delta,\delta_\rmc\},\label{eq:etazetaMU}\\
\zeta(N,\eta(N,f,\Omega),\Omega)&=f. \label{eq:zetaetaMU}
\end{align}
\end{subequations}
\end{lemma}

\begin{lemma}\label{lem:ZetaEtaMonN}
Suppose $N\geq2 $ and $0<\delta,f\leq1$. Then
\begin{subequations}
\begin{align}
\zeta(N,\delta,\Omega)&\geq \zeta(N-1,\delta,\Omega), \label{eq:ZetaMonN}\\
F(N,\delta,\Omega)&\geq F(N-1,\delta,\Omega), \label{eq:FMonN}\\
\eta(N,f,\Omega)&\leq \eta(N-1,f,\Omega), \label{eq:EtaMonN}\\
\caF(N,f,\Omega)&\geq \caF(N-1,f,\Omega). \label{eq:caFMonN}
\end{align}
\end{subequations}
The first two inequalities are saturated iff $\delta\leq \delta_\rmc$ or $\delta=1$, where $\delta_\rmc$ is given in \eref{eq:deltarmc}. 
The last two inequalities are saturated  iff $f=1$.
\end{lemma}

Next, we turn to the figure of merit
$N(\epsilon,\delta,\Omega)$ defined in \eref{eq:MinNumTestDef}. As an implication of \lref{lem:zetaEtaMonoCon}, $N(\epsilon,\delta,\Omega)$ increases monotonically with $1/\epsilon$ and $1/\delta$ as expected. The following lemma provides
several equivalent ways for computing $N(\epsilon,\delta,\Omega)$.
\begin{lemma}\label{lem:MinNumTestDef}
Suppose $0<\epsilon,\delta<1$. Then 
\begin{align}
\! N(\epsilon,\delta,\Omega)&=\min\{N\,| \, \zeta(N,\delta,\Omega)\geq \delta(1-\epsilon) \}  \label{eq:MinNumTestDefzeta}\\
&=\min\{N\,| \, \eta(N,\delta(1-\epsilon),\Omega)\leq \delta \}\label{eq:MinNumTestDefeta}\\
&=\min\{N\,| \, \caF(N,\delta(1-\epsilon),\Omega)\geq (1-\epsilon) \}. \label{eq:MinNumTestDefcaF} 
\end{align}
\end{lemma}

Finally, we present a lemma which is useful for comparing the efficiencies of two verification operators. Let  $\tilde{\Omega}$ be another verification operator for the same target state as $\Omega$. 
\begin{lemma}\label{lem:VScompare}
Suppose $\zeta_\bfk(\bm{\lambda})\geq \zeta(N,\delta=\eta_\bfk(\bm{\lambda}),\tilde{\Omega})$ for all $\bfk\in \scrS_N$. Then
\begin{subequations}
\begin{align}
\zeta(N,\delta,\Omega)&\geq \zeta(N,\delta,\tilde{\Omega}), \quad 0\leq \delta\leq 1,\label{eq:VScompare1} \\
F(N,\delta,\Omega)&\geq F(N,\delta,\tilde{\Omega}), \quad 0< \delta\leq 1,\label{eq:VScompare2} \\
N(\epsilon,\delta,\Omega)&\leq N(\epsilon,\delta,\tilde{\Omega}),\quad 0<\epsilon,\delta<1. \label{eq:VScompare3}
\end{align}
\end{subequations}
\end{lemma}
\Lref{lem:VScompare} is applicable in particular when the set of distinct eigenvalues of $\Omega$ is contained in that of $\tilde{\Omega}$, that is, $\sigma(\Omega)\subset \sigma(\tilde{\Omega})$, assuming $\beta(\Omega)<1$; cf.~\eref{eq:StrategyOrder}.

\section{\label{sec:Homo}Homogeneous strategies}
A strategy (or verification operator) $\Omega$ for $|\Psi\>$ is \emph{homogeneous} if it has the form 
\begin{equation}\label{eq:HomoStrategy}
\Omega=|\Psi\>\<\Psi|+\lambda(\id-|\Psi\>\<\Psi|),
\end{equation}
where $0\leq \lambda<1$. In this case, all eigenvalues of $\Omega$ are equal to $\lambda$ except for the largest one, so  we have $\beta=\tau=\lambda$ and $\nu=1-\lambda$.  Incidentally, the homogeneous strategy $\Omega$ can always be realized by performing the test $P=|\Psi\>\<\Psi|$ with probability $1-\lambda$ and the trivial test with probability $\lambda$. By "trivial test" we mean the test  operator is equal to the identity operator. For bipartite pure states \cite{HayaMT06,Haya09G,PallLM18,ZhuH19O,LiHZ19}
and stabilizer states \cite{PallLM18}, the homogeneous strategy can also be realized by virtue of local projective measurements when $\lambda$ is sufficiently large; see \sref{sec:app}.

In the nonadversarial scenario, a smaller $\lambda$ achieves a better performance among homogeneous strategies. 
Here, we clarify what $\lambda$
is optimal  in the adversarial scenario, which turns out to be very different from the nonadversarial scenario.

Given that the homogeneous strategy $\Omega$ in \eref{eq:HomoStrategy} is determined by the parameter $\lambda$, it is more informative to express the figures of merit defined in \esref{eq:Ffdelta} and  \eqref{eq:MinNumTestDef}  as follows,
\begin{subequations}
	\begin{align}
	F(N,\delta,\lambda)&:=F(N,\delta,\Omega),\\
	\caF(N,f,\lambda)&:=\caF(N,f,\Omega),\\
	\zeta(N,\delta,\lambda)&:=\zeta(N,\delta,\Omega),\\
	\eta(N,f,\lambda)&:=\eta(N,f,\Omega),\\
	N(\epsilon,\delta,\lambda)&:=N(\epsilon,\delta,\Omega).
	\end{align}
\end{subequations}
Then \lref{lem:etaf0} implies that
\begin{equation}
\eta(N,0,\lambda)=\delta_\rmc=\begin{cases}
\lambda^N, & \lambda>0,\\
1/(N+1), & \lambda=0.
\end{cases}
\end{equation}
Suppose $\tilde{\Omega}$ is an arbitrary verification operator with  eigenvalues $1=\tilde{\lambda}_1>\tilde{\lambda}_2\geq\cdots \geq \tilde{\lambda}_D\geq 0$. Then we have $F(N,\delta,\tilde{\lambda}_j)\geq F(N,\delta,\tilde{\Omega})$ for $2\leq j\leq D$ according to \eref{eq:StrategyOrder}. Therefore, the  optimal performance  can always be achieved by a homogeneous strategy if there is no restriction on the accessible measurements. This observation reveals the importance of homogeneous strategies to QSV in the adversarial scenario.

In preparation for the following discussions, we need to introduce a few more notations.
 Denote by $\bbZ$ and $\nni$ the set of integers and  the set of nonnegative integers, respectively. 
For $k\in\nni$, define
\begin{equation}\label{eq:etazetaHomo}
\begin{aligned}
\eta_k(\lambda)&:=\frac{(N+1-k)\lambda^k+k\lambda^{k-1} }{N+1},\\
\zeta_k(\lambda)&:=\frac{(N+1-k)\lambda^k}{N+1}.
\end{aligned}
\end{equation}
We take the convention that $\lambda^0=\eta_0(\lambda)=\zeta_0(\lambda)=1$
even if $\lambda=0$.
Note that 
\begin{equation}\label{eq:etazetaHomo2}
\eta_k(\lambda)=\eta_\bfk(\bm{\lambda}),\quad 
\zeta_k(\lambda)=\zeta_\bfk(\bm{\lambda})
\end{equation}
when $k\in \{0,1,\ldots, N+1\}$, where $\bfk=(N+1-k,k)$, $\bm{\lambda}=(1,\lambda)$, and $\eta_\bfk(\bm{\lambda})$, $\zeta_\bfk(\bm{\lambda})$ are defined in \eref{eq:etazeta}. The extension of the definitions of $\eta_k(\lambda)$ and $
\zeta_k(\lambda)$ over $k$ to the set $\nni$ will be useful in proving several important  results on homogeneous strategies.

\subsection{\label{sec:SHomo}Singular homogeneous strategy}
When  $\lambda=0$,  the verification operator $\Omega=|\Psi\>\<\Psi|$ is singular (has a zero eigenvalue), and \eref{eq:etazetaHomo} reduces to 
\begin{equation}
\eta_k(\lambda)=\begin{cases}
1 & k=0,\\
(N+1)^{-1}& k=1,\\
0 &k\geq 2.
\end{cases}\quad 
\zeta_k(\lambda)=\begin{cases}
1 & k=0,\\
0 &k\geq 1.
\end{cases}
\end{equation}
By \lref{lem:equiTwoDef}, we have $F(N,\delta,\lambda=0)=\zeta(N,\delta,\lambda=0)/\delta$ for $0<\delta\leq1$, where 
\begin{align}
&\zeta(N,\delta,\lambda=0)=\max\biggl\{0, \frac{(N+1)\delta-1}{N}\biggr\}\nonumber\\
&=\begin{cases}
0,& 0\leq \delta\leq (N+1)^{-1},\\
\frac{(N+1)\delta-1}{N}, & (N+1)^{-1}\leq \delta\leq 1.
\end{cases}
\label{eq:FidelityHomo0}
\end{align}
Given $0<\epsilon, \delta<1$, the minimum number  of tests required to verify the  pure state $|\Psi\>$ within infidelity $\epsilon$ and significance level $\delta$ reads
\begin{equation}\label{eq:NumTestSingHomo}
N(\epsilon,\delta,\lambda=0)=\biggl \lceil\frac{1-\delta}{\epsilon\delta}\biggr\rceil.
\end{equation}
Here the scaling with $1/\delta$ is not satisfactory although the strategy is optimal in the nonadversarial scenario according to \esref{eq:NumTest} and \eqref{eq:NumTestAF}. Fortunately, nonsingular homogeneous strategies can achieve a better scaling behavior, as we shall see shortly.

\subsection{\label{sec:NSHomo}Nonsingular homogeneous strategies}
\subsubsection{Verification precision}
Here we assume $0<\lambda<1$, so the homogeneous strategy defined in \eref{eq:HomoStrategy} is nonsingular (which means the verification operator is positive definite). In this case, $\eta_k(\lambda)$  decreases strictly monotonically with $k$ and $\eta_k(\lambda)>0$ for $k\in \nni$; by contrast,
$\zeta_k(\lambda)$ decreases strictly monotonically with $k$ and $\zeta_k(\lambda)\geq 0$ for $k\in \{0,1,\ldots, N+1\}$, while $\zeta_k(\lambda)< 0$ for $k>N+1$. 
Define
\begin{align}
c_{k}(\delta,\lambda):=&\frac{\delta-\eta_{k+1}(\lambda)}{\eta_k(\lambda)-\eta_{k+1}(\lambda)},\label{eq:ck}\\
\zeta(N,\delta,\lambda,k):= &c_{k}(\delta,\lambda)\zeta_k(\lambda)+[1-c_{k}(\delta,\lambda)]\zeta_{k+1}(\lambda)\nonumber\\
=&	
\frac{\lambda\{\delta[1+(N-k)\nu]-\lambda^k\}}{\nu(k\nu+N\lambda)},
 \label{eq:zetadel}
\end{align}
where $\nu=1-\lambda$. The main properties of $\zeta(N,\delta,\lambda,k)$ are summarized in \lsref{lem:zetadelN} and \ref{lem:zetadel} in \aref{sec:HomoApp}. The following theorem determines the fidelity that can be achieved by a given number of tests for a given significance level; see \aref{sec:HomoProof} for a proof.
\begin{theorem}\label{thm:FidelityHomo}
	Suppose $0<\lambda <1$ and $0<\delta \leq 1$.
	Then we have $F(N,\delta,\lambda)=\zeta(N,\delta,\lambda)/\delta$ with
	\begin{equation}\label{eq:FidelityHomo}
	\zeta(N,\delta,\lambda)=\begin{cases}
	0, &\delta\leq \lambda^N,\\
\zeta(N,\delta,\lambda,k_*), &\delta> \lambda^N,
	\end{cases}
	\end{equation}
	where $k_*$ is the largest integer $k$ that satisfies $\eta_k(\lambda)\geq \delta$, that is,  $(N+1-k)\lambda^k+k\lambda^{k-1}\geq (N+1)\delta$. 
\end{theorem}
The choice of the parameter $k_*$ in \thref{thm:FidelityHomo} guarantees that $0<c_{k_*}(\delta,\lambda)\leq 1$. 
Define 
\begin{align}
k_+:=&\lceil\log_\lambda\delta\rceil, \quad k_-:=\lfloor\log_\lambda\delta\rfloor.\label{eq:kpm} 
\end{align}
If  $\lambda^N<\delta\leq 1$, then $0\leq k_+\leq N$ and $0\leq k_-\leq N-1$. Meanwhile, we have   $\eta_{k-}(\lambda)\geq \delta$ and $\eta_{k_++1}(\lambda)<\delta$ by \eref{eq:etazetaHomo}, so $k_*$  is equal to either $k_+$ or $k_-$. In addition,
when $k\in \{0,1,\ldots, N\}$,  \thref{thm:FidelityHomo} implies that
 \begin{equation}\label{eq:FidelityHomoInt}
 F(N,\delta=\lambda^k,\lambda)=\frac{(N-k)\lambda}{k+(N-k)\lambda},
 \end{equation}
 which decreases monotonically with $k$. In particular we have $F(N,\delta=1,\lambda)=1$ as expected; cf.~\eref{eq:Fidsig1}. 
 When $\delta=\eta_k(\lambda)$ with $k\in \{0,1,\ldots, N+1\}$, we have
 \begin{equation}\label{eq:FidelityHomoExP}
 F(N,\delta=\eta_k(\lambda),\lambda)=\frac{\zeta_{k}(\lambda)}{\eta_{k}(\lambda)}=\frac{(N+1-k)\lambda}{k+(N+1-k)\lambda},
 \end{equation}
 which also decreases monotonically with $k$. The dependences of $\zeta(N,\delta,\lambda)$ and $F(N,\delta,\lambda)$ on $\delta$ and $\lambda$ are illustrated in  \fref{fig:FidHomo}.

\begin{corollary}\label{cor:FidHomo}
Suppose $0<\lambda <1$ and  $0<\delta \leq 1$. Then 
	\begin{align}
		\zeta(N,\delta,\lambda)&=\max\Bigl\{0,\max_{k\in \nni}\, \zeta(N,\delta,\lambda,k)\Bigr\} \label{eq:FidelityHomo2}\\
		&=\max\bigl\{0, \zeta(N,\delta,\lambda,k_+), \zeta(N,\delta,\lambda,k_-)\} \label{eq:FidelityHomo3}\\
&=\max\Bigl\{0,\max_{k\in\{0,1,\ldots,N\}}\, \zeta(N,\delta,\lambda,k)\Bigr\}. \label{eq:FidelityHomo4}
	\end{align}
\end{corollary}
\Crref{cor:FidHomo} follows from \thref{thm:FidelityHomo} above and \lref{lem:zetadel} in \aref{sec:HomoApp}. \Eref{eq:FidelityHomo2} provides a family of lower bounds for $\zeta(N,\delta,\lambda)$, namely,
\begin{equation}\label{eq:zetaDelLB}
\zeta(N,\delta,\lambda)\geq \zeta(N,\delta,\lambda,k)\quad \forall k\in \nni. 
\end{equation}

\begin{figure}
	\includegraphics[width=8cm]{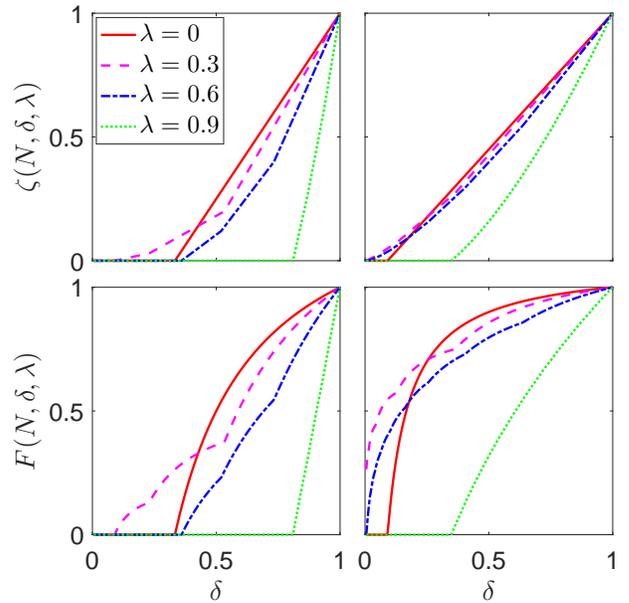}
	\caption{\label{fig:FidHomo}(color online) Variations of $\zeta(N,\delta,\lambda)$ and $F(N,\delta,\lambda)$ with $\delta$ and $\lambda$ for $N=2$ (left plots) and $N=10$ (right plots).}
\end{figure}

\begin{corollary}\label{cor:FidMono}
	Suppose  $0\leq \lambda<1$. Then  $F(N,\delta,\lambda)$ is nondecreasing in $\delta$ for $0<\delta\leq 1$ and in $N$ for $N\geq1$. 
\end{corollary}

\begin{corollary}\label{cor:FidHomoBounds}
	Suppose $0<\lambda <1$ and $\lambda^N\leq \delta \leq 1$. Then 
\begin{equation}\label{eq:FidelityHomoBounds}
\!\frac{(N-k_+)\lambda}{k_++(N-k_+)\lambda}\leq F(N,\delta,\lambda)\leq\frac{(N-k_-)\lambda}{k_-+(N-k_-)\lambda}.
\end{equation}
\end{corollary}
When $\lambda=0$, \crref{cor:FidMono} follows from \eref{eq:FidelityHomo0}.  When $0<\lambda<1$, \crref{cor:FidMono} follows from \thref{thm:FidelityHomo} (cf. 
\crref{cor:FidHomo} above and \lref{lem:zetadelN} in the appendix); alternatively, it is an implication of \lsref{lem:zetaEtaMonoCon} and \ref{lem:ZetaEtaMonN}. \Crref{cor:FidHomoBounds} is an immediate consequence of  \crref{cor:FidMono} and \eref{eq:FidelityHomoInt} given that $\lambda^{k_+}\leq \delta\leq \lambda^{k_-}$. 

\begin{figure}
	\includegraphics[width=7cm]{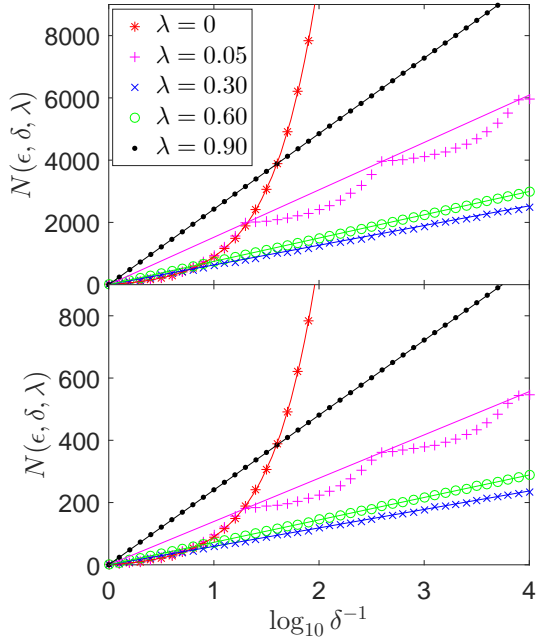}\caption{\label{fig:NumHomodel}(color online)  Minimum  numbers  of tests required to verify a pure state with five different homogeneous strategies. Here $\epsilon=0.01$ in the upper plot and $\epsilon=0.1$ in the lower plot. In each plot, the red curve  represents the approximate formula $(1-\delta)/(\epsilon\delta)$ when $\lambda=0$; cf.~\eref{eq:NumTestSingHomo}. The four lines represent the approximate  formula $(F+\lambda\epsilon)\log_{10}\delta/(\lambda\epsilon\log_{10}\lambda)$; cf.~\eref{eq:limdel}.}
\end{figure}

\subsubsection{Number of required tests}

Now, we are ready to determine the minimum number of tests required to verify the pure state $|\Psi\>$ within infidelity $\epsilon$ and significance level $\delta$ in the adversarial scenario.  \Thsref{thm:NumTestHomo} and \ref{thm:NumTestHomoBounds}  below are proved  in \aref{sec:HomoProof}. The results are illustrated in \fsref{fig:NumHomodel} and \ref{fig:NumHomolam}. 
Define
\begin{align}
\tilde{N}(\epsilon,\delta,\lambda,k)&:=\frac{k\nu^2 \delta F +\lambda^{k+1}+\lambda\delta(k\nu-1)}{\lambda\nu\delta \epsilon}, \label{eq:NumTestHomoC}\\
\tilde{N}_\pm(\epsilon,\delta,\lambda)&:=\tilde{N}(\epsilon,\delta,\lambda,k_\pm), \label{eq:Ntildepm}
\end{align}
where $F=1-\epsilon$,  $\nu=1-\lambda$, and $k_\pm$ are given in \eref{eq:kpm}. The main properties of $\tilde{N}(\epsilon,\delta,\lambda,k)$ are summarized in \lref{lem:NumTestProp} in the appendix; see also \lref{lem:NumTestmBound}. In particular, $\tilde{N}(\epsilon,\delta,\lambda,k)\leq \tilde{N}(\epsilon,\delta,\lambda,k-1)$ iff $\delta \leq \lambda^k/(F+\lambda\epsilon)$, assuming that $0<\epsilon,\delta, \lambda<1$ and $k$ is a positive integer. 
\begin{theorem}\label{thm:NumTestHomo}
	Suppose $0<\epsilon,\delta, \lambda<1$. Then we have
	\begin{align}
	N(\epsilon, \delta,\lambda)&=\Bigl\lceil\min_{k\in \nni}\!\tilde{N}(\epsilon,\delta,\lambda,k)\Bigr\rceil=\bigl\lceil\tilde{N}(\epsilon,\delta,\lambda,k^*)\bigr\rceil \label{eq:NumTestHomo1}\\
&=
\bigl\lceil\min\bigl\{\tilde{N}_+(\epsilon,\delta,\lambda),\tilde{N}_-(\epsilon,\delta,\lambda)\bigr\}\bigr\rceil   \label{eq:NumTestHomo2} \\
&=\begin{cases}
\lceil\tilde{N}_-(\epsilon,\delta,\lambda)\rceil, & \delta \geq \frac{\lambda^{k_+}}{F+\lambda \epsilon},\\
\lceil \tilde{N}_+(\epsilon,\delta,\lambda)\rceil, & \delta \leq  \frac{\lambda^{k_+}}{F+\lambda \epsilon},
\end{cases}\label{eq:NumTestHomo3}	
	\end{align}
where $k^*$ is the largest integer $k$ that satisfies the inequality $\delta\leq \lambda^{k}/(F\nu+\lambda )= \lambda^{k}/(F+\lambda \epsilon)$ and it  is equal to either $k_+$ or $k_-$. 
\end{theorem}

\begin{corollary}\label{cor:NumTestHomoUB}
Suppose $0<\epsilon,\delta, \lambda<1$. Then
\begin{equation}\label{eq:NumTestHomoUB}
N(\epsilon, \delta,\lambda)\leq \lceil\tilde{N}(\epsilon,\delta,\lambda,k)\rceil \quad \forall k\in \nni,
\end{equation}
where the upper bound for a given $k$ is  saturated when $\lambda^{k+1}/(F+\lambda \epsilon)\leq \delta\leq \lambda^{k}/(F+\lambda \epsilon)$.
\end{corollary}
\Crref{cor:NumTestHomoUB} is an easy consequence of \thref{thm:NumTestHomo}. The two cases $k=0,1$  are of special interest,
\begin{align}
N(\epsilon,\delta,\lambda)&\leq \lceil\tilde{N}(\epsilon,\delta,\lambda,0)\rceil=\biggl \lceil\frac{1-\delta}{\nu\epsilon\delta}\biggr\rceil, \label{eq:NumTestHomoLSUB1}\\
N(\epsilon,\delta,\lambda)&\leq \lceil\tilde{N}(\epsilon,\delta,\lambda,1)\rceil=\biggl \lceil\frac{\nu^2 \delta F +\lambda^{2}-\lambda^2\delta}{\lambda\nu\delta \epsilon}\biggr\rceil. \label{eq:NumTestHomoLSUB2}
\end{align}
If $\lambda/(F+\lambda\epsilon)\leq \delta <1$, then \eref{eq:NumTestHomoLSUB1} is saturated, so we have
\begin{equation}\label{eq:NumTestHomoLS}
N(\epsilon,\delta,\lambda)=\biggl \lceil\frac{1-\delta}{\nu\epsilon\delta}\biggr\rceil.
\end{equation}
This result also holds when  $\lambda=0$ (as long as  $0<\epsilon,\delta< 1$) according to \eref{eq:NumTestSingHomo}. 
If $\lambda^2/(F+\lambda\epsilon)\leq \delta\leq \lambda/(F+\lambda\epsilon)$, then \eref{eq:NumTestHomoLSUB2} is saturated, so we have
\begin{equation}\label{eq:NumTestHomoLS2}
N(\epsilon,\delta,\lambda)=\biggl \lceil\frac{\nu^2 \delta F +\lambda^{2}-\lambda^2\delta}{\lambda\nu\delta \epsilon}\biggr\rceil\geq\frac{2\sqrt{(1-\delta)F}}{\epsilon\sqrt{\delta}},
\end{equation}
where the lower bound is proved in \aref{sec:HomoProof}. 
\Esref{eq:NumTestHomoLS} and \eqref{eq:NumTestHomoLS2} indicate that homogeneous strategies with small $\lambda$, say $\lambda\leq 0.1$, are not efficient for high-precision QSV
(say $\epsilon, \delta\leq 0.1$), as reflected in  \fref{fig:NumHomolam}.

\begin{figure}
	\includegraphics[width=7cm]{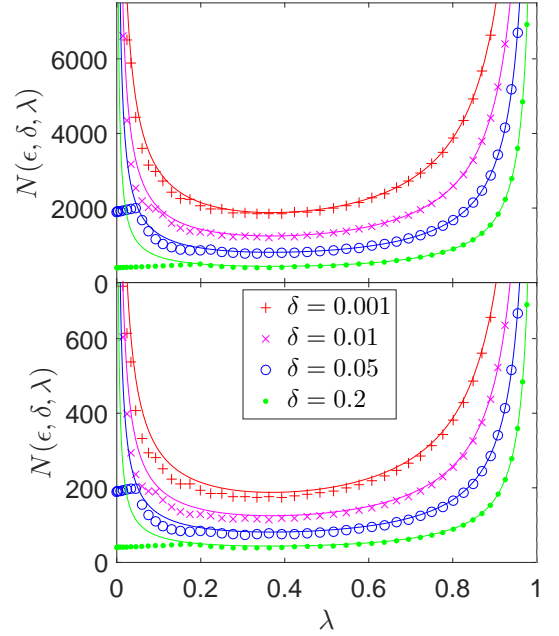}
	\caption{\label{fig:NumHomolam}(color online) Variation of $N(\epsilon,\delta,\lambda)$ with $\lambda$ and $\delta$. Here $\epsilon=0.01$ in the upper plot and $\epsilon=0.1$ in the lower plot. The four curves in each plot represent the  approximate formula $\ln\delta/(\lambda\epsilon\ln\lambda)$; cf.~\esref{eq:limepdel} and \eqref{eq:NumTestHomoApp}.} 	
\end{figure}

The following theorem provides informative bounds for $N(\epsilon,\delta,\lambda)$, which complement the analytical formulas  in \thref{thm:NumTestHomo}. 
\begin{theorem}\label{thm:NumTestHomoBounds} Suppose $0<\epsilon,\delta,\lambda <1$. Then we have
	\begin{gather}
	k_-+\biggl\lceil\frac{k_-F}{\lambda\epsilon}\biggr\rceil\leq N(\epsilon,\delta,\lambda)\leq	k_++\biggl\lceil\frac{k_+F}{\lambda\epsilon}\biggr\rceil,\label{eq:NumberBoundAdvHomo} \\
	N(\epsilon,\delta,\lambda)\leq
	\biggl\lceil \frac{\log_\lambda\delta}{\lambda\epsilon}-\frac{\nu k_-}{\lambda}\biggr\rceil =\biggl\lceil\frac{\ln\delta}{\lambda\epsilon\ln \lambda}-\frac{\nu k_-}{\lambda}\biggr\rceil 
 \label{eq:NumberBoundAdvHomo2}.
	\end{gather}
All three bounds in \esref{eq:NumberBoundAdvHomo} and \eqref{eq:NumberBoundAdvHomo2} are saturated when $\log_\lambda\delta$ is an integer. 
\end{theorem}
When $\delta\leq \lambda\leq 1/2$, we have $k_-\geq 1$ and $\nu k_-/\lambda\geq1$, so \eref{eq:NumberBoundAdvHomo2} implies that
\begin{equation}\label{eq:NumberBoundAdvHomo3}
N(\epsilon,\delta,\lambda)<
 \frac{\ln\delta}{\lambda\epsilon\ln \lambda}.  
\end{equation}
On the other hand, by virtue of \eref{eq:NumberBoundAdvHomo}, we can derive 
\begin{align}
\lim_{\delta\rightarrow 0} \frac{N(\epsilon,\delta,\lambda)}{\ln\delta^{-1}}&=\frac{F+\lambda\epsilon}{\lambda\epsilon \ln \lambda^{-1}},\label{eq:limdel}\\
 \frac{k_-}{\lambda}\leq \lim_{\epsilon\rightarrow 0} \epsilon N(\epsilon,\delta,\lambda)&\leq \frac{k_+}{\lambda},\\
 \lim_{\epsilon,\delta\rightarrow 0} \frac{\epsilon N(\epsilon,\delta,\lambda)}{\ln\delta^{-1}}&=\frac{1}{\lambda \ln \lambda^{-1}}.\label{eq:limepdel}
\end{align}
The exact value of $\lim_{\epsilon\rightarrow 0} \epsilon N(\epsilon,\delta,\lambda)$ can be derived by virtue of  \eref{eq:NumTestHomo3}, with the result
\begin{align}
\!\!\lim_{\epsilon\rightarrow 0} \epsilon N(\epsilon,\delta,\lambda)=\lim_{\epsilon\rightarrow 0} \epsilon \tilde{N}_-(\epsilon,\delta,\lambda)&= \frac{k_-}{\lambda}+\frac{\lambda^{k_-}-\delta}{\nu\delta}.\label{eq:limep}
\end{align}
Note that the inequality $\delta \geq \lambda^{k_+}/(F+\lambda \epsilon)$ is always satisfied in the limit $\epsilon \rightarrow 0$ if $\log_\lambda\delta $ is not an integer, while $k_+=k_-$ and  $\tilde{N}_+(\epsilon,\delta,\lambda)=\tilde{N}_-(\epsilon,\delta,\lambda)$ if $\log_\lambda\delta $ is  an integer. 

\subsection{\label{sec:HomoOpt}Optimal homogeneous strategies}

\subsubsection{Optimal  strategies in the high-precision limit $\epsilon,\delta\rightarrow 0$}
In the adversarial scenario, the optimal performance can always be achieved by a homogeneous strategy if there is no restriction on the measurements. However,  the value of $\lambda$ that minimizes $N(\epsilon,\delta,\lambda)$ depends on the target precision, as characterized by $\epsilon$ and $\delta$. We cannot find a  homogeneous strategy that is optimal for all $\epsilon$ and $\delta$, unlike the nonadversarial scenario. Here we are mostly interested in the high-precision limit, which means $\epsilon,\delta\rightarrow 0$.

According to \eref{eq:limepdel}, in the  high-precision limit, the minimum number of tests can be approximated as follows,
\begin{equation}\label{eq:NumTestHomoApp}
N(\epsilon,\delta,\lambda)\approx
(\lambda\epsilon)^{-1}
\log_\lambda\delta=
(\lambda\epsilon\ln\lambda)^{-1}
\ln\delta.
\end{equation}
To understand the condition of this approximation, 
note that $k_\pm \approx  \log_\lambda \delta$ if $\delta\ll \lambda$, which  is usually the case in high-precision verification. If in addition $\epsilon\ll 1$, then the ratio of the lower bound over   the upper bound in \eref{eq:NumberBoundAdvHomo} is close to 1, so that the  two bounds
 are nearly tight with respect to the  relative deviation. In this case, \eref{eq:NumTestHomoApp} is a good approximation. Furthermore, numerical calculation shows that \eref{eq:NumTestHomoApp} is quite accurate for most parameter range of interest,  
 as illustrated in  \fsref{fig:NumHomodel} and \ref{fig:NumHomolam}. When $\lambda$ is very small, the approximation in \eref{eq:NumTestHomoApp} is not so good. Such homogeneous strategies are not efficient when $\epsilon, \delta\leq 0.1$ as illustrated in  \fref{fig:NumHomolam} [see also \esref{eq:NumTestHomoLS} and \eqref{eq:NumTestHomoLS2}]; in addition,  they  are not so important due to the reasons  explained in \sref{sec:recipe} later.

Thanks to \thsref{thm:NumTestHomo} and \ref{thm:NumTestHomoBounds}, the number of  tests required by  any nonsingular homogeneous strategy can achieve the same scaling behaviors with $\epsilon$ and $\delta$ as the counterpart  in the nonadversarial scenario for high-precision QSV. In the  limit $\epsilon,\delta\rightarrow 0$, the  efficiency is characterized by the function $(\lambda\ln\lambda^{-1})^{-1}$. Analysis shows that the function $(\lambda\ln\lambda^{-1})^{-1}$ is convex for $0<\lambda <1$ and attains the minimum $\rme$ when $\lambda=1/\rme$, with $\rme$ being the base of the natural logarithm. It is strictly decreasing in $\lambda$ when $0< \lambda\leq 1/\rme$ and  strictly increasing when $1/\rme\leq  \lambda< 1$; cf.~\fref{fig:NumHomolam}. Therefore, 
the homogeneous strategy with  $\lambda=1/\rme$, that is,  $\nu=1-(1/\rme)$, 
is  optimal in the high-precision limit $\epsilon,\delta\rightarrow 0$ if there is no restriction on the accessible measurements. In  this case we have
\begin{equation}
N(\epsilon,\delta,\lambda=\rme^{-1})
\approx\rme\epsilon^{-1}
\ln\delta^{-1}.
\end{equation}
Compared with the counterpart $\epsilon^{-1}
\ln\delta^{-1}$ for the nonadversarial scenario,
 the overhead is only $\rme$ times.

Although we cannot find a value of $\lambda$ that is optimal for all $\epsilon$ and $\delta$,  the optimal value  usually lies in a neighborhood, say $[0.32, 0.38]$,  of $1/\rme$  for the values of $\epsilon$ and $\delta$ that are of practical interest, say $\epsilon,\delta\leq 0.1$. In addition, $N(\epsilon,\delta,\lambda)$ varies quite slowly with $\lambda$ in this neighborhood, as illustrated in \fref{fig:NumHomolam}. So the choice  $\lambda=1/\rme$ is usually nearly optimal even if it is not optimal. 

The above analysis shows that the optimal strategies for the adversarial scenario are very different from the counterpart for the nonadversarial scenario. As a consequence, entangling measurements are less helpful and often unnecessary for constructing the optimal strategies for bipartite and multipartite systems. 
In the case of bipartite pure states and GHZ states  for example, the optimal strategies for high-precision verification can  be realized using only local projective measurements \cite{HayaMT06,Haya09G,ZhuH19O,LiHZ19,LiHZ19O} (cf.~\sref{sec:app}).

\subsubsection{Optimal  strategies in the  limit  $\delta\rightarrow 0$}
Here we discuss briefly the scenario in which $\delta\rightarrow 0$, but $\epsilon$ is not necessarily so small, which is relevant to entanglement detection \cite{ZhuH19O}. According to \eref{eq:limdel}, in this case, the performance of  the homogeneous strategy $\Omega$ is characterized by 
\begin{align}\label{eq:Neplam}
\caN(\epsilon,\lambda):=\lim_{\delta\rightarrow 0} \frac{N(\epsilon,\delta,\lambda)}{\ln\delta^{-1}}&=\frac{F+\lambda\epsilon}{\lambda\epsilon \ln \lambda^{-1}},
\end{align}
where $F=1-\epsilon$. 
The partial derivative of $\caN(\epsilon,\lambda)$ over $\lambda$ reads
\begin{align}
  \frac{\partial\caN(\epsilon,\lambda)}{\partial \lambda}=\frac{F+\lambda\epsilon+F\ln \lambda}{\lambda^2\epsilon (\ln \lambda)^2}.
 \end{align}
For a given $\epsilon$, denote by $\caN_*(\epsilon)$ the minimum of $\caN(\epsilon,\lambda)$ over $\lambda$.
This minimum is attained when $\lambda=\lambda_*(\epsilon)$, where $\lambda_*(\epsilon)$ is the unique solution of the equation
\begin{equation}\label{eq:lambdaF}
F+\lambda\epsilon+F\ln \lambda=0,
\end{equation}
which amounts to the equality
\begin{equation}\label{eq:Flambda}
F=\frac{\lambda}{\ln \lambda^{-1}+\lambda-1}.
\end{equation}
It is not difficult to verify that $\lambda_*(\epsilon)=0$ when $\epsilon=1$ ($F=0$) and $\lambda_*(\epsilon)=1/\rme$ when $\epsilon=0$ ($F=1$); in addition, $\lambda_*(\epsilon)$ decreases monotonically with $\epsilon$ and is concave in $\epsilon$, as illustrated in \fref{fig:OptHomo}. Therefore, $\lambda_*(\epsilon)$ satisfies the following equation,
\begin{equation}
\rme^{-1}F\leq \lambda_*(\epsilon)\leq \rme^{-1}.
\end{equation}

\begin{figure}
	\includegraphics[width=7.5cm]{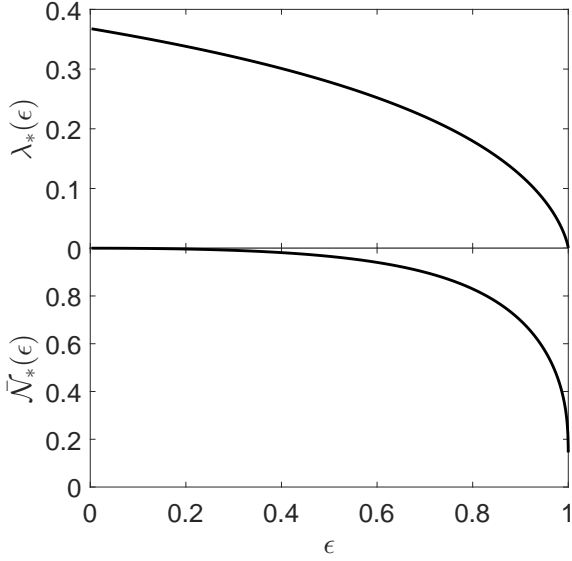}
	\caption{\label{fig:OptHomo}Optimal homogeneous strategy in the limit $\delta\rightarrow 0$. 
		Here $\lambda_*(\epsilon)$ denotes the value of $\lambda$ that minimizes  $\caN(\epsilon,\lambda)$ defined in \eref{eq:Neplam}, which determines the number of required tests.  $\bar{\caN}_*(\epsilon)$ denotes the number of required tests normalized with respect to the benchmark, as defined in \eref{eq:caNnorm}. 	
	} 	
\end{figure}

Next, we study the dependence of the efficiency on the parameter $\lambda$.  As a benchmark, we choose the homogeneous strategy with $\lambda=1/\rme$ in which case we have $\caN(\epsilon,\lambda=\rme^{-1})=(\rme F+\epsilon)/\epsilon$. 
Define
\begin{align}
\bar{\caN}(\epsilon,\lambda)&:=\frac{\caN(\epsilon,\lambda)}{\caN(\epsilon,\rme^{-1})}=\frac{F+\lambda\epsilon}{(\rme F+\epsilon)\lambda\ln\lambda^{-1}},\\ \bar{\caN}_*(\epsilon)&:=\frac{\caN_*(\epsilon)}{\caN(\epsilon,\rme^{-1})}=\frac{F+\lambda_*(\epsilon)\epsilon}{(\rme F+\epsilon)\lambda_*(\epsilon)\ln\lambda_*(\epsilon)^{-1}}. \label{eq:caNnorm}
\end{align}
When $\lambda<1/\rme$,
$\bar{\caN}(\epsilon,\lambda)$ decreases  monotonically with $\epsilon$, so  we have
\begin{align}
\frac{1}{\ln\lambda^{-1}}\leq \bar{\caN}(\epsilon,\lambda)\leq \frac{1}{\rme \lambda\ln\lambda^{-1}}.
\end{align}
The lower bound approaches zero in the limit $\lambda\rightarrow0$. 
Accordingly, a homogeneous strategy $\Omega$ with a small value of $\lambda$ could be significantly more efficient than the benchmark when $\epsilon$ is large.
When $\lambda>1/\rme$, by contrast,
$\bar{\caN}(\epsilon,\lambda)$ increases monotonically with $\epsilon$, so  we have
\begin{align}
1<\frac{1}{\rme \lambda\ln\lambda^{-1}}\leq \bar{\caN}(\epsilon,\lambda)\leq \frac{1}{\ln\lambda^{-1}}.
\end{align}
Such a homogeneous strategy is less efficient than the benchmark.

Finally, by virtue of \esref{eq:Flambda} and \eqref{eq:caNnorm} we can derive the following equality, 
\begin{align}
 \bar{\caN}_*(\epsilon):=\frac{1}{\rme\lambda_*(\epsilon)-\ln\lambda_*(\epsilon)-1}.
\end{align}
Given that  $\lambda_*(\epsilon)\leq \rme^{-1}$ and  $\lambda_*(\epsilon)$ decreases monotonically with $\epsilon$, we can deduce that $ \bar{\caN}_*(\epsilon)$ decreases monotonically with $\epsilon$; it 
approaches 1 in the limit $\epsilon\rightarrow 0$, while it approaches 0 (quite slowly) in the limit $\epsilon\rightarrow 1$, as illustrated in \fref{fig:OptHomo}. Although $ \bar{\caN}_*(\epsilon)$ could be arbitrarily small when $\epsilon$ is large, it is close to 1 when $\epsilon$ is not too large. For example, $\bar{\caN}_*(\epsilon)\geq 0.965$ when $\epsilon\leq 0.5$ and $\bar{\caN}_*(\epsilon)\geq 0.999$ when $\epsilon\leq 0.1$. Therefore, the homogeneous strategy $\Omega$ with $\beta(\Omega)=1/\rme$ is nearly optimal for most parameter range of practical interest, as pointed out earlier.

\section{\label{eq:SingleCopy}Single-copy verification}
In this section we analyze the possibility of QSV in the adversarial scenario using a single test. This problem is of intrinsic interest to single-copy entanglement detection \cite{DimiD18,ZhuH19O}.  Given a verification strategy $\Omega$, the state $|\Psi\>$ can be verified within infidelity $0<\epsilon<1$ and significance level $0<\delta<1$ using  a single test iff
\begin{equation}
F(N=1,\delta,\Omega)\geq 1-\epsilon.
\end{equation}
 Since $F(N,\delta,\Omega)=\zeta(N,\delta,\Omega)/\delta $ according to \eref{eq:Fzeta}, the above equation is equivalent to 
\begin{equation}
\zeta(N=1,\delta,\Omega)\geq \delta(1-\epsilon). 
\end{equation}
So our main task here is to determine the expression of $\zeta(N,\delta,\lambda)$ in the case $N=1$. In the rest of this section we assume $N=1$ except when stated otherwise.  Note that $\zeta(N,\delta=0,\Omega)=0$ and that the range of $\delta$ of practical interest usually satisfies $0<\delta\leq 1/2$.

\subsection{Single-copy verification with  homogeneous strategies}

First, let us consider the homogeneous strategy $\Omega$ defined in \eref{eq:HomoStrategy}.
\begin{proposition}\label{pro:SingleCopyHomo}
	Suppose $N=1$ and $0\leq \lambda<1$; then  
	\begin{align}
\zeta(N,\delta,\lambda)&=\max\left\{ 0, 
\frac{\lambda(\delta-\lambda)}{1-\lambda}, 
\frac{\delta(2-\lambda)-1}{1-\lambda}\right\}\nonumber\\	
&=\begin{cases}
	0, & 0\leq\delta\leq \lambda,\\[0.2ex]
	\frac{\lambda(\delta-\lambda)}{1-\lambda}, & \lambda\leq \delta\leq \frac{1+\lambda}{2}, \\[0.6ex]
	\frac{\delta(2-\lambda)-1}{1-\lambda}, &  \frac{1+\lambda}{2}\leq \delta\leq 1.
	\end{cases} \label{eq:SingleCopyAdvHomo}
	\end{align}
\end{proposition}
\Pref{pro:SingleCopyHomo}  follows from \eref{eq:FidelityHomo0} when $\lambda=0$ and follows from \thref{thm:FidelityHomo} and \crref{cor:FidHomo} when  $0<\lambda<1$. As an implication, we can derive
\begin{align}
\max_{\lambda}\zeta(N,\delta,\lambda)&=\max\{2-2\sqrt{1-\delta}-\delta,\; 2\delta-1\}\nonumber \\
&=\begin{cases}
2-2\sqrt{1-\delta}-\delta, & 0\leq\delta\leq \frac{5}{9},\\[0.2ex]
2\delta-1, & \frac{5}{9}\leq \delta \leq 1. \\[0.6ex]
\end{cases}\label{eq:SingleCopyAdvHomoMax}
\end{align}
Here the maximum is attained at
\begin{equation}\label{eq:SingleCopyLamOpt}
\lambda=\begin{cases}
1-\sqrt{1-\delta}, & 0\leq \delta \leq \frac{5}{9},\\
0,& \frac{5}{9}\leq \delta\leq 1. 
\end{cases}
\end{equation}
In addition, the optimal solution $\lambda$ is unique for $0<\delta<1$ except when $\delta=5/9$, in which case there are two optimal solutions, namely, $\lambda=0$ and $\lambda=1/3$. This observation implies the following corollary given that  the optimal strategy can always be chosen to be homogeneous if there is no restriction on the measurements.
\begin{corollary}\label{cor:SingleCopyAdvCon1}
	The target state can be verified within infidelity $0<\epsilon<1$ and significance level $0<\delta <1$ in the adversarial scenario using a single test iff $\delta$ and $\epsilon$  satisfy the condition
\begin{equation}\label{eq:SingleCopyAdvCon1}
\delta(1-\epsilon)\leq 
\max\{2-2\sqrt{1-\delta}-\delta,\; 2\delta-1\}, 
\end{equation}
or, equivalently, the condition
\begin{align}\label{eq:SingleCopyAdvCon2}
	\!\!\delta\geq \min\left\{ \frac{4(1-\epsilon)}{(2-\epsilon)^2}, \frac{1}{1+\epsilon}\right\}
	\!	=\!\begin{cases}
	\frac{1}{1+\epsilon}, & 0< \epsilon\leq  \frac{4}{5}, \\[0.4ex]
	\frac{4(1-\epsilon)}{(2-\epsilon)^2}, & \frac{4}{5}\leq \epsilon< 1. 
	\end{cases}
	\end{align}
\end{corollary} 
\begin{figure}
	\includegraphics[width=7cm]{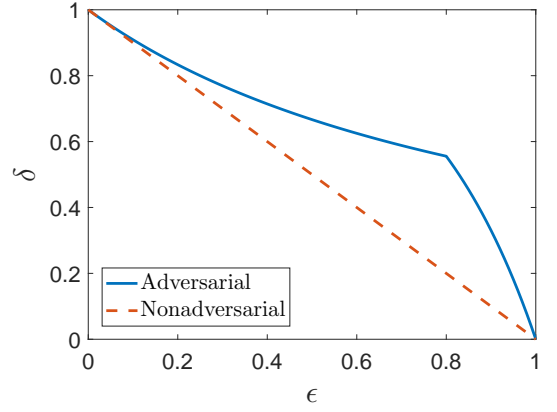}
	\caption{\label{fig:SingleCopy}(color online) Single-copy verification in the adversarial scenario and nonadversarial scenario. The target state can be verified within infidelity $\epsilon$ and significance  level $\delta$ in the adversarial (nonadversarial) scenario using a single test if the value of $\delta$ lies above the blue solid curve (red dashed line);  cf.~\esref{eq:SingleCopyAdvCon2} and \eqref{eq:OneTestCon2}. 
	}
\end{figure}

The parameter range of single-copy verification characterized by \crref{cor:SingleCopyAdvCon1}
is illustrated in \fref{fig:SingleCopy} in contrast with the counterpart   for the nonadversarial scenario in \eref{eq:OneTestCon2}. 
\Eref{eq:SingleCopyAdvCon2} determines the smallest significance level that can be achieved by  a single test to verify the target state within infidelity $\epsilon$. Note that the lower bound is monotonically decreasing in $\epsilon$ for $0<\epsilon<1$ as expected. 
To achieve  significance level $\delta \leq 1/2$, the infidelity must satisfy the condition $\epsilon\geq 2(\sqrt{2}-1)$.
When the  bound in \eref{eq:SingleCopyAdvCon1} or that in \eref{eq:SingleCopyAdvCon2} is saturated, the target state can be verified within infidelity $0<\epsilon<1$ and significance level $0<\delta<1$ 	
by a strategy $\Omega$ iff $\Omega$ is homogeneous and $\beta(\Omega)$ is given by \eref{eq:SingleCopyLamOpt} with $0<\delta<1$ or, equivalently,
\begin{equation}\label{eq:SingleCopyLamOpt2}
\beta(\Omega)=\lambda=\begin{cases}
0, & 0<\epsilon \leq \frac{4}{5},\\
\frac{2-2\epsilon}{2-\epsilon}, & \frac{4}{5}\leq \epsilon <1.
\end{cases}
\end{equation} 
When $\delta\neq 5/9$  (that is, $\epsilon\neq 4/5$), the optimal strategy~$\Omega$  is unique
as shown \esref{eq:SingleCopyLamOpt}  and \eqref{eq:SingleCopyLamOpt2}. 
When $\delta= 5/9$ ($\epsilon=4/5$),  by contrast, there are two optimal strategies, both of which are homogeneous, and $\beta(\Omega)$ can take on two possible values, namely, $\beta(\Omega)=0$ and $\beta(\Omega)=1/3$ (cf.~\thref{thm:SingleCopy} below).

\begin{corollary}\label{cor:SingleCopyAdvHomo}
Given a homogeneous strategy $\Omega$ with $\beta(\Omega)=\lambda$, the target state can be verified  within infidelity $0<\epsilon<1$ and significance level $0<\delta\leq 1/2$ in the adversarial scenario  using a single test
iff 
\begin{equation}\label{eq:OneTestConAdvHomo}
	\frac{\lambda(\delta-\lambda)}{1-\lambda}\geq \delta(1-\epsilon). 
\end{equation}
This requirement is equivalent to the
following conditions,
\begin{align}
\delta \geq \frac{4(1-\epsilon)}{(2-\epsilon)^2},\label{eq:deltaep}\\
\lambda_-\leq \lambda\leq \lambda_+,
\end{align}
where 
\begin{align}\label{eq:lambdapm}
\lambda_{\pm}:=\frac{(2-\epsilon)\delta\pm\sqrt{(2-\epsilon)^2\delta^2-4(1-\epsilon)\delta}}{2}.
\end{align}
\end{corollary}
\Eref{eq:OneTestConAdvHomo}  implies that $0< \lambda<\delta$. So any homogeneous strategy $\Omega$ with $\beta(\Omega)=0$ or $\beta(\Omega)\geq 1/2$ cannot verify the target state within infidelity $0<\epsilon<1$ and significance level $0<\delta\leq 1/2$ using a single test. This conclusion actually applies to an arbitrary strategy, not necessarily homogeneous; see \crref{cor:SingleCopyAdv} below.
Thanks to  the inequality $4(1-\epsilon)\delta>4(1-\epsilon)\delta^2$, $\lambda_{\pm}$ defined in \eref{eq:lambdapm} satisfy the following equation,
\begin{equation}
(1-\epsilon)\delta< \lambda_-\leq \lambda_+< \delta. 
\end{equation}
By computing the derivatives over $\delta$ and $\epsilon$, it is easy to verify that $\lambda_+$ ($\lambda_-$) increases (decreases) monotonically with $\delta$ and $\epsilon$ as expected.
If $\delta \leq 1/2$, then we  have
\begin{equation}
\frac{2-\epsilon-\sqrt{\epsilon^2+4\epsilon-4}}{4} \!\leq \!\lambda_-\!\leq \!\lambda_+\!\leq\! \frac{2-\epsilon+\sqrt{\epsilon^2+4\epsilon-4}}{4}. 
\end{equation}

\subsection{Single-copy verification with  general strategies}

Next, we generalize \pref{pro:SingleCopyHomo} to an arbitrary verification operator $\Omega$.  The following theorem shows that the efficiency of $\Omega$ is determined by $\beta$ and $\tau$, where  $\beta$ and $\tau$ denote the second largest and  smallest eigenvalues of $\Omega$, respectively. See \aref{sec:SingleCopyProof} for a proof. 
\begin{theorem}\label{thm:SingleCopy}
	Suppose $N=1$. 
If $\beta\geq 1/2$, then 
\begin{equation}\label{eq:SingleCopyAdvGen1}
\zeta(N,\delta,\Omega)=\begin{cases}
0, &0\leq \delta\leq \beta,\\[0.2ex]
\frac{\beta(\delta-\beta)}{1-\beta}, &  \beta\leq \delta\leq \frac{1+\beta}{2}, \\[0.6ex]
\frac{\delta(2-\beta)-1}{1-\beta}, &  \frac{1+\beta}{2}\leq \delta\leq 1.
\end{cases}
\end{equation}
If $\beta< 1/2$, then 
\begin{equation}\label{eq:SingleCopyAdvGen2}
\zeta(N,\delta,\Omega)=\begin{cases}
0, & 0\leq\delta\leq \beta,\\[0.2ex]
\frac{\tau(\delta-\beta)}{1+\tau-2\beta}, & \beta\leq \delta\leq \frac{1+\tau}{2}, \\[0.6ex]
\delta-\frac{1}{2}, &\frac{1+\tau}{2}\leq \delta\leq \frac{1+\beta}{2}, \\[0.6ex]
\frac{\delta(2-\beta)-1}{1-\beta}, &  \frac{1+\beta}{2}\leq \delta\leq 1.
\end{cases}
\end{equation}
\end{theorem}

\begin{corollary}\label{cor:SingleCopyAdv}
The target state can be verified by the strategy $\Omega$  within infidelity $0<\epsilon<1$ and significance level $0<\delta\leq 1/2$ using a single test
	iff 
	\begin{equation}\label{eq:OneTestConAdv}
0<\beta<\delta,\quad \frac{\tau(\delta-\beta)}{1+\tau-2\beta}\geq \delta(1-\epsilon). 
	\end{equation}
\end{corollary}
Note that the target state cannot be  verified within infidelity $0<\epsilon<1$ and significance level $0<\delta\leq 1/2$ using a single test if $\beta=0$ or $\beta\geq 1/2$. 
When $0<\beta<1/2$ and $\beta\leq \delta\leq (1+\tau)/2$, we have
\begin{equation}
\frac{\tau(\delta-\beta)}{1+\tau-2\beta}\leq \min\left\{\frac{\beta(\delta-\beta)}{1-\beta},\frac{\tau(\delta-\tau)}{1-\tau}\right\}. 
\end{equation}
So \eref{eq:OneTestConAdv} implies \eref{eq:OneTestConAdvHomo} with $\lambda=\beta$ or $\lambda=\tau$, which in turn implies 
\eref{eq:deltaep} and the sequence of  inequalities $\lambda_-\leq \tau\leq \beta\leq\lambda_+$, where $\lambda_{\pm}$ are defined in  \eref{eq:lambdapm}. This conclusion is expected given that $\zeta(N,\delta,\Omega)\leq \zeta(N,\delta,\beta)$ and $\zeta(N,\delta,\Omega)\leq \zeta(N,\delta,\tau)$.

\section{\label{sec:NumTestGen}Efficiencies of general verification strategies }
In this section we present our main results on the efficiencies of general verification strategies.  As we shall see shortly, the efficiency of a general verification operator $\Omega$ of a pure state $|\Psi\>$ is mainly determined by its  second largest eigenvalue $\beta$ (or equivalently $\nu=1-\beta$) and the smallest eigenvalue  $\tau$. 

\subsection{Singular verification strategies}
The efficiency of a singular verification strategy is characterized by \Lref{lem:MinFidelityUB} and \thref{thm:NIID} below, which  are proved in \aref{sec:NIIDproof}. Note that \esref{eq:MinFidelityBound} and \eqref{eq:MinFidelityBound2}  in \thref{thm:NIID} actually apply to all verification strategies, although these bounds could be quite loose for nonsingular strategies. Define
\begin{equation}
\delta^*:=\frac{1+ N\beta}{N+1}=\frac{1+ N(1-\nu)}{N+1}.
\end{equation}
\begin{lemma}\label{lem:MinFidelityUB}
	Suppose $\Omega$ is a singular verification operator and  $1/(N+1)\leq \delta\leq \delta^*$.  Then
	\begin{equation}\label{eq:MinFub}
	F(N,\delta,\Omega)\leq
	1-\frac{1}{ (N+1)\delta}.
	\end{equation}
\end{lemma}

\begin{theorem}\label{thm:NIID}
Suppose $0<\delta\leq 1$ and $0<\nu\leq1 $. Then
	\begin{align}
	F(N,\delta,\Omega) \ge
	1-\frac{1- \delta}{ N\nu\delta},
	\label{eq:MinFidelityBound}
	\end{align}
and	the inequality  is saturated when $\delta^*\leq \delta \leq 1 $. If in addition $\nu\geq 1/2$, then
	\begin{align}\label{eq:MinFidelityBound2}
	F(N,\delta,\Omega)\ge
	1-\frac{1}{ (N+1)\delta},
	\end{align}
and	the inequality  is saturated when $\Omega$ is singular and $\delta$ satisfies $1/(N+1)\leq \delta\leq \delta^*$.
\end{theorem}
The   bound in \eref{eq:MinFidelityBound} is positive and thus nontrivial if $\delta> 1/(N\nu+1)$, while the one in \eref{eq:MinFidelityBound2} is positive  if $\delta> 1/(N+1)$. The first bound is saturated and thus optimal when
$\delta\geq \delta^*$, while the second bound is better when $\delta< \delta^*$. The two bounds coincide when $\delta=\delta^*$. 
The  bound in \eref{eq:MinFidelityBound2}
under the condition
$\nu \ge 1/2 $
was also given in \rcite{HayaM15} under a slightly different situation.
According to \lref{lem:MinFidelityUB} and \thref{thm:NIID}, if  $\Omega$ is singular,  then 
\begin{align}\label{eq:MinFidelityBound3}
F(N,\delta,\Omega)\leq  \max\biggl\{0,
1-\frac{1- \delta}{ N\nu\delta},
1-\frac{1}{ (N+1)\delta}\biggr\}.
\end{align}
If  $\nu\geq 1/2$, by contrast, then the above inequality is reversed,
\begin{align}\label{eq:MinFidelityBound4}
F(N,\delta,\Omega)\geq \max\biggl\{0,
1-\frac{1- \delta}{ N\nu\delta},
1-\frac{1}{ (N+1)\delta}\biggr\}.
\end{align}
If  $\Omega$ is singular and meanwhile $\nu\geq 1/2$, 
then the inequalities in \esref{eq:MinFidelityBound3} and \eqref{eq:MinFidelityBound4} are  saturated. 

\begin{corollary}\label{cor:NumTestAdvGen}Suppose $0<\epsilon,\delta<1$ and $0<\nu\leq 1$. Then 
\begin{equation}\label{eq:NumTestAdvUB}
N(\epsilon,\delta,\Omega)\leq\biggl\lceil
\frac{1- \delta}{\nu\delta \epsilon}\biggr\rceil.
\end{equation}
If $\Omega$ is singular, then 
\begin{equation}\label{eq:NumTestAdvSingularLB}
N(\epsilon,\delta,\Omega)\geq \min\left\{
\biggl\lceil\frac{1- \delta}{\nu \delta \epsilon}\biggr\rceil,\; \biggl\lceil\frac{1}{\delta\epsilon}-1\biggr\rceil\right\}.
\end{equation}
If  $\nu\geq 1/2$, then 
\begin{equation}\label{eq:NumTestAdvLvUB}
N(\epsilon,\delta,\Omega)\leq \min\left\{
\biggl\lceil\frac{1- \delta}{\nu \delta \epsilon}\biggr\rceil,\; \biggl\lceil\frac{1}{\delta\epsilon}-1\biggr\rceil\right\}.
\end{equation}
\end{corollary}
\Crref{cor:NumTestAdvGen} is an easy consequence of \thref{thm:NIID} and \esref{eq:MinFidelityBound3}, \eqref{eq:MinFidelityBound4}. If  $\Omega$ is singular and  $\nu\geq 1/2$, 
then the inequalities in \esref{eq:NumTestAdvSingularLB} and \eqref{eq:NumTestAdvLvUB} are  saturated, so we have
\begin{equation}
N(\epsilon,\delta,\Omega)= \min\left\{
\biggl\lceil\frac{1- \delta}{\nu \delta \epsilon}\biggr\rceil,\; \biggl\lceil\frac{1}{\delta\epsilon}-1\biggr\rceil\right\},
\end{equation}
which generalizes \eref{eq:NumTestSingHomo}.  The number of tests characterized by the  upper bound in \eref{eq:NumTestAdvUB} is much smaller than what can be achieved by previous approaches that are based on the quantum de Finetti theorem \cite{MoriTH17,TakeM18}. Nevertheless, the scaling with $1/\delta$ is still not satisfactory compared with the counterpart for the nonadversarial scenario.

\subsection{Nonsingular verification strategies}

Next, we  provide an even better bound on the number of tests when $\Omega$ is nonsingular.  \Lref{lem:Fboundf} and \thref{thm:NumTestBounds} below are proved in \aref{sec:NumberTesOptProof}. 
\begin{lemma}\label{lem:Fboundf}
	Suppose $0<\delta,f\leq 1$ and  $\Omega$ is a positive-definite verification operator  with $0<\tau\leq \beta<1$. Then
\begin{align}
F(N,\delta,\Omega)&\geq \frac{N+1-(\ln \beta)^{-1}\ln (\tau\delta)}{N+1-(\ln \beta)^{-1}\ln (\tau\delta)-h\ln (\tau\delta)},  \label{eq:Fbounddel}\\
\caF(N,f,\Omega)&\geq \frac{N+1-(\ln \beta)^{-1}\ln f}{N+1-(\ln \beta)^{-1}\ln f-h\ln f}, \label{eq:Fboundf}
\end{align}
where
\begin{align}\label{eq:Overhead}
h&=h(\Omega):=\max_{j\geq 2}\bigl(\lambda_j \ln \lambda_j^{-1}\bigr)^{-1}\nonumber\\
&=\bigl[\min\{\beta \ln \beta^{-1}, \tau \ln \tau^{-1}\}\bigr]^{-1}.
\end{align}
\end{lemma}
Define
\begin{equation}\label{eq:betatilde}
\tilde{\beta}:=\begin{cases}
\beta, \quad  \beta\ln \beta^{-1}\leq \tau\ln \tau^{-1},\\
\tau, \quad \beta\ln \beta^{-1}> \tau\ln \tau^{-1}.
\end{cases}
\end{equation}
Then we have $h=(\tilde{\beta}\ln \tilde{\beta}^{-1})^{-1}$.
Note that $h>1/|\ln \beta|$ and $-h\ln (\tau\delta)>(\ln \beta)^{-1}\ln (\tau\delta)$, so the denominator in \eref{eq:Fbounddel} is positive, and so is the denominator in \eref{eq:Fboundf}. In addition, the lower bounds in \esref{eq:Fbounddel} and \eqref{eq:Fboundf} increase monotonically with $N$, which is expected in view of \lref{lem:ZetaEtaMonN}.

By virtue of \lref{lem:Fboundf} we can derive upper bounds for	$N(\epsilon,\delta,\Omega)$ which are tight in the high-precision limit. Meanwhile, we can derive lower bounds for $N(\epsilon,\delta,\Omega)$ based on the fact that $	N(\epsilon,\delta,\Omega)\geq N(\epsilon,\delta,\lambda_j)$ for $ j=2,3,\ldots, D$,   where $\lambda_j$ are the eigenvalues of $\Omega$ arranged in decreasing order $1=\lambda_1>\lambda_2\geq \lambda_3\geq \cdots \geq \lambda_D>0$. The main results are summarized in the following theorem.
\begin{theorem}\label{thm:NumTestBounds}
Suppose $0<\epsilon,\delta<1$ and  $\Omega$ is a positive-definite verification operator with $0<\tau\leq \beta<1$. Then
	\begin{widetext}
\begin{gather}
		N(\epsilon,\delta,\Omega)\geq N(\epsilon,\delta,\lambda_j)\geq k_-(\lambda_j)+\biggl\lceil\frac{k_-(\lambda_j)F}{\lambda_j\epsilon}\biggr\rceil,\quad j=2,3,\ldots, D,  \label{eq:NumTestLB}\\
		k_-(\tilde{\beta})+\biggl\lceil\frac{k_-(\tilde{\beta})F}{\tilde{\beta}\epsilon}\biggr\rceil\leq 
		N(\epsilon,\delta,\Omega)\leq \left\lceil \frac{hF\ln(F\delta)^{-1}}{\epsilon}  +\frac{\ln(F\delta)}{\ln\beta}-1\right\rceil< \frac{h\ln(F\delta)^{-1}}{\epsilon} , 	\label{eq:NumTestLUB}\\		
		N(\epsilon,\delta,\Omega)\leq \left\lceil \frac{hF\ln(\tau\delta)^{-1}}{\epsilon}  +\frac{\ln(\tau\delta)}{\ln\beta}-1\right\rceil< \frac{h\ln(\tau\delta)^{-1}}{\epsilon}, 	\label{eq:NumTestUB2}	
		\end{gather}
	\end{widetext}
where we have $F=1-\epsilon$,  $k_-(\lambda_j)=\lfloor (\ln\delta)/\ln\lambda_j\rfloor$, and $k_-(\tilde{\beta})=\lfloor (\ln\delta)/\ln\tilde{\beta}\rfloor$.
\end{theorem}

The upper bounds in \eref{eq:NumTestLUB} are worse than those in \eref{eq:NumTestUB2} if $F<\tau=\tau(\Omega)$, while they are better if $F>\tau$, which is usually the case for high-precision verification. 
Suppose $\tau$ is bounded from below by a positive constant. 
Then  the ratio of the lower bound over   the upper bound in \eref{eq:NumTestLUB} approaches 1 in the high-precision limit $\epsilon,\delta\rightarrow 0$, so the two bounds are nearly tight, as in the case of homogeneous strategies. As a consequence, we have
\begin{align}\label{eq:limepdelGen}
\lim_{\epsilon,\delta\rightarrow 0} \frac{\epsilon N(\epsilon,\delta,\Omega)}{\ln\delta^{-1}}=h=\frac{1}{\tilde{\beta}\ln\tilde{\beta}^{-1}}.
\end{align}
When $\epsilon,\delta\ll 1$, accordingly, $N(\epsilon,\delta,\Omega)$ can be approximated as follows,
\begin{equation}\label{eq:NumberMeasOpt2}
N(\epsilon,\delta,\Omega)\approx  \frac{h\ln\delta^{-1}}{\epsilon}=\frac{\ln\delta}{\epsilon\tilde{\beta}\ln\tilde{\beta}}. 
\end{equation}
The number of tests has the same scaling behaviors with $\epsilon^{-1}$ and $\delta^{-1}$ as the counterpart for  the nonadversarial scenario presented in \esref{eq:NumTest} and \eqref{eq:NumTestAF},  except for an overhead characterized by $\nu h$. However, $\Omega$ is not efficient when $\tau$ is too small according to \eref{eq:NumTestLB} as well as \esref{eq:NumTestHomoLS} and \eqref{eq:NumTestHomoLS2}. In addition,
the scaling behavior with $\delta^{-1}$ would be worse if $\Omega$ were singular according to \eref{eq:NumTestAdvSingularLB}.

The above analysis can be extended to the scenario in which we want to
verify whether the support of the resultant state
belongs to a certain subspace ${\cal K}$.
In this case,
we need to replace
the  projector $|\Psi\>\<\Psi|$ by the projector $P$ onto
the subspace ${\cal K}$, impose the condition $ E_l P=P$, and redefine $ f_\rho$ as $\tr[(\Omega^{\otimes N}\otimes P ) \rho]$. Such an extension is useful when we want to verify whether
the resultant state is correctable in a  fault-tolerant way \cite{FujiH17}.

\section{\label{sec:recipe}General recipe to verifying pure states in the adversarial scenario}
According to  \sref{sec:NumTestGen}, the number $N(\epsilon,\delta,\Omega)$ of tests  required to verify a pure state in the adversarial scenario has the same scaling behavior with $\epsilon^{-1}$ and $\delta^{-1}$ as the counterpart for  the nonadversarial scenario  as long as the verification operator  $\Omega$ is nonsingular, and its smallest eigenvalue $\tau$ is bounded from below by a positive constant.  
However, the scaling behavior of $N(\epsilon,\delta,\Omega)$ with $\delta$ is suboptimal when $\Omega$ is singular, that is,  $\tau=0$. Similarly, the efficiency is limited when $\tau$ is nonzero, but very small. To address this problem, here we provide a simple recipe to  reducing  the number of tests significantly, so that pure states can be verified  in the adversarial scenario with high precision and with nearly  the same efficiency as in the nonadversarial scenario. Surprisingly, all we need to do is to perform the trivial test with a suitable probability. By  "trivial test" we mean the test whose test operator $E$  is equal to the identity operator, that is $E=1$, so that  all the states can pass the test with certainty.

\subsection{The recipe}
Suppose $\Omega$ is a  verification operator for the pure state $|\Psi\>$.  Based on $\Omega$, we can construct a new verification operator as follows,
\begin{equation}
\Omega_p=(1-p)\Omega+p,\quad 0\le p<1,
\end{equation}
which means  the trivial test is performed with probability $p$ and  $\Omega$ is performed with probability $1-p$.
Denote by  $\beta_p$ and $\tau_p$ the second largest eigenvalue and  smallest eigenvalue  of $\Omega_p$, respectively. Then 
\begin{equation}
\beta_p=(1-p)\beta+p=1-\nu+p\nu,\quad \tau_p=(1-p)\tau+p,
\end{equation}
where $\beta$ and $\tau$ are the second largest eigenvalue and  smallest eigenvalue  of $\Omega$,  which satisfy the inequality $\tau\leq \beta$. 
Here we view  $\beta_p$ as  a function of $\nu=1-\beta$ and $p$. The spectral gap of $\Omega_p$ reads
\begin{equation}
\nu_p=1-\beta_p=(1-p)\nu.
\end{equation}

 According to \ssref{sec:SetStage} and \ref{sec:QSVre}, the trivial test can only decrease the efficiency in the nonadversarial scenario. In high-precision verification for example, the number of tests required by $\Omega_p$ is about $1/(1-p)$ times the number required by $\Omega$ according to \esref{eq:NumTest} and \eqref{eq:NumTestAF}. 
In sharp contrast, 
the trivial test can increase the efficiency in the adversarial scenario by hedging the influence of small eigenvalues of $\Omega$. Therefore, $\Omega_p$ is called a \emph{hedged verification operator} of $\Omega$.

Thanks to  \eref{eq:NumTestLUB}, to verify the target state $|\Psi\>$ within infidelity $\epsilon$ and significance level $\delta$ in the adversarial scenario, the number of tests required by the strategy $\Omega_p$ (assuming $\tau_p>0$) is upper bounded as follows,
\begin{equation}\label{eq:NumTestAdvRev}
N(\epsilon,\delta,\Omega_p)< \frac{h(p,\nu,\tau) \ln(F\delta)^{-1}}{\epsilon}, 
\end{equation}
where $F=1-\epsilon$ and 
\begin{align} h(p,\nu,\tau)&=h(\Omega_p)=\bigl[\min\bigl\{\beta_p \ln \beta_p^{-1}, \tau_p \ln \tau_p^{-1}\bigr\}\bigr]^{-1}. 
\end{align}	
In comparison with the number  in \eref{eq:NumTest} or \eqref{eq:NumTestAF} for the nonadversarial scenario, the overhead satisfies
\begin{equation}\label{eq:AdvNaRatio1}
\frac{N(\epsilon,\delta,\Omega_{p})}{N_{\na}(\epsilon,\delta,\Omega)}<
\nu h(p,\nu,\tau)	\frac{[\ln (1-\nu\epsilon)^{-1}] \ln(F\delta)}{\nu\epsilon\ln \delta}. 
\end{equation}
It is straightforward to verify that  this bound decreases monotonically with  $1/\epsilon$ and $1/\delta$. It turns out that the bound also decreases monotonically with  $1/\nu$ according to \lsref{lem:hpnmin} and \ref{lem:hpnt} below.
When $\epsilon$ and $\delta$ approach  zero, the bound in \eref{eq:NumTestAdvRev} becomes tight (with respect to the relative deviation) according to  \esref{eq:NumTestLUB} and \eqref{eq:limepdelGen}, so 
we have
\begin{equation}
\lim_{\epsilon,\delta\rightarrow 0}
\frac{N(\epsilon,\delta,\Omega_{p})}{N_{\na}(\epsilon,\delta,\Omega)}= \nu h(p,\nu,\tau).
\end{equation}
This equation corroborates the significance of the  function $\nu h(p,\nu,\tau)$ for  characterizing the overhead of high-precision QSV in the adversarial scenario.

To construct an efficient hedged verification strategy, we need to choose a suitable value of $p$ so as to minimize $h(p,\nu,\tau)$. To this end, it is instructive to recall that the function $x\ln x^{-1}$ is concave in the interval $0\leq x\leq 1$ and is strictly increasing in $x$ when $0\leq x\leq 1/\rme$, while it is strictly decreasing when $1/\rme\leq  x\leq 1$; it attains the maximum $1/\rme$ when $x=1/\rme$. Given the value of $\nu=1-\beta$ and $\tau$ with $\nu+\tau\leq1$, 
the minimum of $h(p,\nu,\tau)$ over $p$ is denoted by $h_*(\nu,\tau)$; the  unique minimizer in $p$ is denoted by $p_*(\nu,\tau)$ or $p_*$  for simplicity; cf.~\fref{fig:p*}. By definition we have
\begin{equation}\label{eq:phstardef}
h_*(\nu,\tau):=\min_{0\leq p<1}h(p,\nu,\tau)=h(p_*,\nu,\tau).
\end{equation}
In addition, it is straightforward to verify that
\begin{equation}
p_*=\min\bigl\{p\geq 0| \beta_p\geq \rme^{-1}\;\& \;\tau_p \ln \tau_p^{-1}\geq \beta_p \ln \beta_p^{-1} \bigr\}. \label{eq:pstardef}
\end{equation}
Here the condition $\beta_p\geq \rme^{-1}$ is
required  when $\tau=\beta$ (so that $\Omega$ is a homogeneous strategy), but is redundant when $\tau<\beta$.  \Eref{eq:pstardef} implies that   $\beta_{p_*}\geq 1/\rme$; by contrast,  $\tau_{p_*}\leq 1/\rme$ if $\tau\leq 1/\rme$.

When the strategy $\Omega$ is homogeneous, that is, when
$\tau=\beta=1-\nu$, we have
\begin{align}
p_*(\nu,1-\nu)&=\begin{cases}
0, &0<\nu\leq 1-\frac{1}{\rme}, \\
\frac{\rme \nu-\rme+1}{\rme \nu},& 1-\frac{1}{\rme}\leq \nu\leq 1;
\end{cases}\label{eq:pstartaumax}\\
h_*(\nu,1-\nu)&=\begin{cases}
(\beta\ln \beta^{-1})^{-1}, &0<\nu\leq 1-\frac{1}{\rme}, \\
\rme, & 1-\frac{1}{\rme}\leq \nu\leq 1.
\end{cases}\label{eq:hstartaumax}
\end{align}
In this case $\Omega_p$ is also homogeneous, so the results presented in \sref{sec:Homo} can be applied directly. In general, it is not easy to derive an analytical formula for $p_*$, but it is very easy to determine $p_*$ numerically.

\begin{figure}
	\includegraphics[width=7cm]{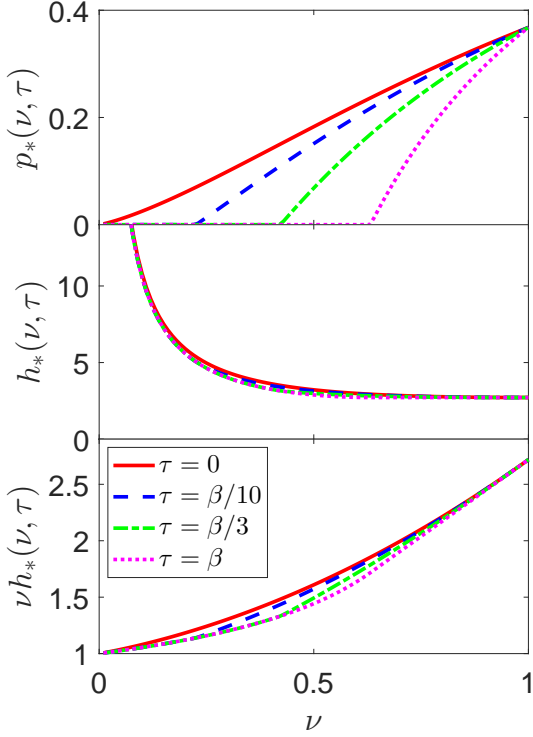}
	\caption{\label{fig:p*}(color online) The optimal probability $p_*(\nu,\tau)$ for performing the trivial test
		(upper plot), 	
		the prefactor $h_*(\nu,\tau)$ (middle plot), and the  overhead $\nu h_*(\nu,\tau)$ (lower plot) for high-precision QSV in the adversarial scenario. In the legend $\beta=1-\nu$.}
\end{figure}

\subsection{Properties of hedged verification strategies}
To determine the overhead of QSV in the adversarial scenario, we need to clarify  the properties of $h(p,\nu,\tau)$, $h_*(\nu,\tau)$,  and $p_*(\nu,\tau)$, which determine the performances of the  hedged verification strategies $\Omega_p$ and $\Omega_{p_*}$. 
By virtue of the properties of the function $x\ln x^{-1}$ we can derive a tight lower bound for $h(p,\nu,\tau)$, namely,
\begin{equation}\label{eq:hpvtLBe}
h(p,\nu,\tau)\geq \rme,
\end{equation}
and the  bound is saturated iff
$\tau_p=\beta_p=1/\rme$, that is, 
$\tau=1-\nu\leq 1/\rme$ and $p=(\rme \nu-\rme+1)/(\rme\nu)$; cf.~\esref{eq:pstartaumax} and \eqref{eq:hstartaumax}.

\begin{lemma}\label{lem:hpnmin}
	Suppose   $0<\nu\leq1$. 
	Then $p_*(\nu,1-\nu)$ is nondecreasing in $\nu$, 
	$h_*(\nu,1-\nu)$ is nonincreasing in  $\nu$, and $\nu h_*(\nu,1-\nu)$ is strictly increasing in $\nu$. Meanwhile,  $\nu h_*(\nu,1-\nu)>1$ and
	$\lim_{\nu\rightarrow 0} \nu h_*(\nu,1-\nu)=1$. If in addition $0\leq p<1$  and $\beta_p=1-\nu+p\nu>0$, 
	then  $\nu h(p,\nu,1-\nu)$  is strictly increasing in $\nu$.  
\end{lemma}

\begin{lemma}\label{lem:hpnt}	
 	Suppose  $\nu$ and $\tau$ satisfy the following conditions $0<\nu\leq1$, $0\leq \tau<1$, and  $\nu+\tau\leq 1$. Then 
	\begin{enumerate}
		\item $p_*(\nu,\tau)$ is nondecreasing in $\nu$ and nonincreasing in~$\tau$.
		
	\item 	$h_*(\nu,\tau)$ is nonincreasing in both $\nu$ and $\tau$.
	
	\item $\nu h_*(\nu,\tau)>1$.
	
	\item 
	$\lim_{\nu\rightarrow 0} \nu h_*(\nu,\tau)=1$.
	
   \item 
  $\nu h_*(\nu,\tau)$ is strictly increasing in $\nu$.  
  
\end{enumerate}  
If in addition 	 $0\leq p<1$ and $\tau_p =(1-p)\tau +p >0$, then 
\begin{enumerate}
\setcounter{enumi}{5}	
	\item $h(p,\nu,\tau)$ is nonincreasing  in both $\nu$ and $\tau$.
	
	\item 
	$\nu h(p,\nu,\tau)$ is  strictly increasing in $\nu$.  
	
	\end{enumerate}
\end{lemma}

\Lsref{lem:hpnmin} and \ref{lem:hpnt} are proved in \aref{sec:PTTproofs}. \Lref{lem:hpnmin}  is tailored to the scenario in which $\Omega$ is homogeneous. In \lref{lem:hpnt} we assume that $\nu$ and $\tau$ can vary independently, which means the Hilbert space $\caH$ on which $\Omega$ acts has dimension at least 3. If $\caH$ has dimension 2, then 
 $\Omega$ is always homogeneous and $\tau=1-\nu$, so \lref{lem:hpnt} is redundant given \lref{lem:hpnmin}. \Lsref{lem:hpnmin} and \ref{lem:hpnt}
 summarize the main properties of $p_*(\nu,\tau)$,  $h(p,\nu,\tau)$, and $h_*(\nu,\tau)$  as illustrated in \fref{fig:p*}, which are very instructive to understanding QSV in the adversarial scenario. In particular \lref{lem:hpnt} 
 reveals that  the overhead   $\nu h_*(\nu,\tau)$ in the number of tests becomes negligible when $\nu$ approaches 0. To be concrete,  calculation shows that 
 $\nu h_*(\nu,\tau)\leq 1.09,1.19,1.31,1.45,1.61$ when $\nu\leq 0.1,0.2,0.3,0.4,0.5$, respectively.

When $p_*(\nu,\tau)\leq p\leq p_*(\nu):=p_*(\nu,0)$, \lref{lem:hpnt} implies that
\begin{equation}\label{eq:hpntBounds}
h_*(\nu,1-\nu)\!\leq\! h_*(\nu,\tau)\!\leq\! h(p,\nu,\tau)
\!\leq\! h(p_*(\nu),\nu,\tau)\!=\! h_*(\nu), 
\end{equation}
where $h_*(\nu):=h_*(\nu,0)$. 
Note that $h(p,\nu,\tau)$ increases monotonically with $p$ when $p\geq p_*(\nu,\tau)$. 
\Lref{lem:hpnt} and \eref{eq:pstartaumax} together yield a lower bound and an upper bound for $ p_*(\nu,\tau)$,
\begin{equation}
 p_*(\nu,1-\nu)\leq  p_*(\nu,\tau)\leq  p_*(1-\tau,\tau)\leq 1/\rme. 
\end{equation}
Here the third inequality is saturated  iff $\tau=0$; in that case, the second inequality is saturated iff $\nu=1$ (cf.~\lref{lem:OverheadPG} below). Therefore, $p_*(\nu,\tau)$ can attain the upper bound $1/\rme$ iff $\nu=1$ and $\tau=0$, in which case the verification operator is homogeneous and singular. As a corollary, 
we have $1/[1-p_*(\nu,\tau)]\leq \rme/(\rme-1)<1.6$, so  the number of tests required by $\Omega_{p_*}$ is at most $60\%$ more than the number required by $\Omega$ for high-precision verification in the nonadversarial scenario although here we are mainly interested in the adversarial scenario. 
By contrast, \lref{lem:hpnt}  and \eref{eq:hstartaumax} yield a lower bound for $h_*(\nu,\tau)$, 
\begin{align}
&h_*(\nu,\tau)\geq
\begin{cases}
(\beta\ln \beta^{-1})^{-1}, &0<\nu\leq 1-\frac{1}{\rme}, \\
\rme, & 1-\frac{1}{\rme}\leq \nu\leq 1,
\end{cases}
\end{align}
where $\beta=1-\nu$.

When $0<\tau<\beta$ and $\tau\ln \tau^{-1}\geq\beta\ln \beta^{-1} $, \eref{eq:pstardef} implies that
\begin{align}
p_*(\nu,\tau)&=0,\quad 
h_*(\nu,\tau)=
(\beta\ln \beta^{-1})^{-1}.
\end{align}
So there is no need to perform the trivial test.
When  $\tau\ln \tau^{-1}<\beta\ln \beta^{-1}$ ( which implies that $\tau<1/\rme$, including the case $\tau=0$), the probability $p_*(\nu,\tau)$ happens to be the unique solution
of the equation
\begin{equation}\label{eq:balance}
\beta_p\ln \beta_p=\tau_p\ln \tau_p,\quad  0<p<1.
\end{equation}
In this case, it is beneficial to perform the trivial test with a suitable probability. The inequality $\tau\ln \tau^{-1}<\beta\ln \beta^{-1}$ is thus an indication that $\tau$ is too small.

\begin{figure}
	\includegraphics[width=7cm]{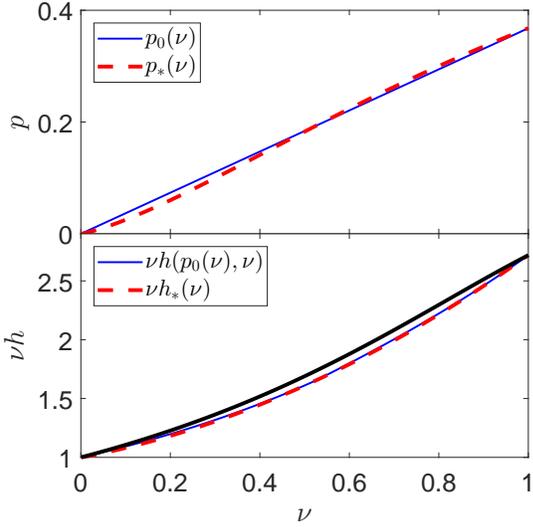}
	\caption{\label{fig:p0}(color online) The optimal probability $p_*(\nu)$ for performing the trivial test in high-precision QSV and a pretty-good approximation $p_0(\nu)=\nu/\rme$ (upper plot).
		Variations  of $\nu h_*(\nu)$ and its upper bound $\nu h(p_0(\nu),\nu)$  with $\nu$ (lower plot). The black solid curve in the lower plot represents the first upper bound for $\nu h(p_0(\nu),\nu)$
		presented in \eref{eq:OverheadPG2}. }
\end{figure}

In view of \lref{lem:hpnt}, a singular verification operator $\Omega$
with $\tau=0$ is of special interest because the overhead $\nu h_*(\nu,\tau)$ for a given $\nu$ is maximized when $\tau=0$.  
In this case, $\tau_p=p$  and the optimal probability in \eref{eq:pstardef} reduces to 
\begin{equation}
p_*=\min\bigl\{p> 0| \beta_p\geq \rme^{-1}\;\& \; p\ln p= \beta_p \ln \beta_p \bigr\}. \label{eq:pstarSingular}
\end{equation}
The requirement $\beta_p\geq \rme^{-1}$ is redundant  when $\beta>0$, in which case $p_*$ is also the unique solution of the equation $ p\ln p= \beta_p \ln \beta_p$ for $0<p<1$. 
In general, we have $h_*(\nu)=(p_*\ln p_*^{-1})^{-1}$. Furthermore, $p_*(\nu)=p_*(\nu,0)$ can be approximated by
\begin{equation}
p_0=p_0(\nu)=\frac{\nu}{\rme}=\frac{1-\beta}{\rme}, 
\end{equation}
which is exact when $\nu=1$, as illustrated in  \fref{fig:p0}. Let  $h(p,\nu):=h(p,\nu,0)$; then  $h(p_0(\nu),\nu)=h(\rme^{-1}\nu,\nu)=h(\rme^{-1}\nu,\nu,0)$. 

\begin{lemma}\label{lem:OverheadPG} Suppose $0<\nu\leq 1$. Then  $p_*(\nu)$, $\nu h_*(\nu)$, and $\nu h(\rme^{-1}\nu,\nu)$ are strictly  increasing in  $\nu$,  while $h_*(\nu)$ and  $h(\rme^{-1}\nu,\nu)$ are strictly decreasing in  $\nu$. In addition,
	\begin{align}
	\nu h_*(\nu)&\leq \nu h(\rme^{-1}\nu,\nu)\leq
	(1-\nu +\rme^{-1}\nu^2)^{-1}\nonumber\\
	&\leq 1+(\rme-1)\nu\leq \rme. \label{eq:OverheadPG2}
	\end{align}
\end{lemma}
\Lref{lem:OverheadPG} is proved in \aref{sec:PTTproofs}. Calculation shows that the difference between $\nu h(\rme^{-1}\nu,\nu)$ and $\nu h_*(\nu)$ is less than $2\%$ (cf.~\fref{fig:p0}); therefore, $p_0=\nu/\rme$ is indeed a good approximation of $p_*(\nu)$. When $p_*(\nu,\tau)\leq p\leq p_*(\nu)$, \lref{lem:OverheadPG} and \eref{eq:hpntBounds} imply that
\begin{align}
& \nu h_*(\nu,\tau)\!\leq\! \nu h(p,\nu,\tau)
\!\leq\! \nu h_*(\nu)\leq  \nu h(\rme^{-1}\nu,\nu)\nonumber\\
&\leq
(1-\nu +\rme^{-1}\nu^2)^{-1}\leq  1+\rme\nu-\nu\leq \rme. \label{eq:OverheadPG3}
\end{align}
In addition, we have $h(\rme^{-1}\nu,\nu,\tau)\leq h(\rme^{-1}\nu,\nu)$ according to \lref{lem:hpnt}. So  \lref{lem:OverheadPG} has implications for all verification operators, not necessarily singular.

%\bigskip 

\subsection{Overhead of QSV in the adversarial scenario}
The overhead of QSV in the adversarial scenario  compared with the nonadversarial scenario is of fundamental interest.
The following theorem is a key to clarifying 
this issue. It follows from  \lref{lem:hpnt} as well as   \esref{eq:NumTestAdvRev} and \eqref{eq:OverheadPG3}. 
\begin{widetext}
\begin{theorem}\label{thm:AdvOHGen}Suppose $\Omega$ is a verification operator for $|\Psi\>$, $\nu=\nu(\Omega)$, and $\tau=\tau(\Omega)$. If $p=\nu/\rme$, then 
\begin{align}
&N(\epsilon,\delta,\Omega_p)< \frac{h(\rme^{-1}\nu,\nu,\tau) \ln(F\delta)^{-1}}{\epsilon}\leq \frac{h(\rme^{-1}\nu,\nu) \ln(F\delta)^{-1}}{\epsilon}\leq  \frac{\ln(F\delta)^{-1}}{(1-\nu+\rme^{-1}\nu^2)\nu\epsilon}\leq \frac{(1+\rme\nu-\nu)\ln(F\delta)^{-1}}{\nu\epsilon},\label{eq:NumTestGenA}
\end{align}	
where  $F=1-\epsilon$.		
 If $p_*(\nu,\tau)\leq p\leq p_*(\nu)$, then 	
\begin{align}
&N(\epsilon,\delta,\Omega_p)< \frac{h(p,\nu,\tau) \ln(F\delta)^{-1}}{\epsilon}\leq \frac{h_*(\nu) \ln(F\delta)^{-1}}{\epsilon}\leq \frac{h(\rme^{-1}\nu,\nu) \ln(F\delta)^{-1}}{\epsilon}\leq  \frac{\ln(F\delta)^{-1}}{(1-\nu+\rme^{-1}\nu^2)\nu\epsilon}.
\end{align}
\end{theorem}	
In conjunction with \eref{eq:NumTestAF} [see also \esref{eq:AdvNaRatio1} and \eqref{eq:OverheadPG3}], \thref{thm:AdvOHGen} sets a general upper bound  on the overhead of QSV in the adversarial scenario. If $p=\nu/\rme$ or $p_*(\nu,\tau)\leq p\leq p_*(\nu)$ for example,  then 	
\begin{equation}\label{eq:AdvNaRatio}
\frac{N(\epsilon,\delta,\Omega_{p})}{N_{\na}(\epsilon,\delta,\Omega)}<
	\nu h(\rme^{-1}\nu,\nu)	\frac{[\ln(1-\nu\epsilon)^{-1}]\ln(F\delta) }{\nu\epsilon\ln \delta}\leq 	
	\frac{[\ln (1-\nu\epsilon)^{-1}] \ln(F\delta)}{(1-\nu+\rme^{-1}\nu^2)\nu\epsilon\ln \delta}\leq 
	\frac{(1+\rme\nu-\nu)[\ln (1-\nu\epsilon)^{-1}] \ln(F\delta)}{\nu\epsilon\ln \delta}. 	
	\end{equation}
\end{widetext}

\begin{figure}
	\includegraphics[width=7cm]{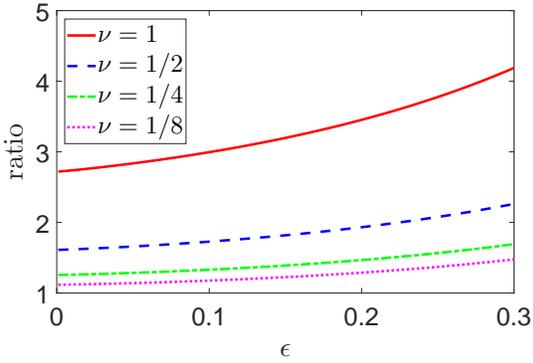}
	\caption{\label{fig:ratio}(color online) Upper bound on the ratio of  $N(\epsilon,\delta,\Omega_p)$ over $N_{\na}(\epsilon,\delta,\Omega)$ according to  the first bound in \eref{eq:AdvNaRatio} with $\delta=\epsilon$, where $p=\nu/\rme$ or $p_*(\nu,\tau)\leq p\leq p_*(\nu)$. This ratio
 characterizes the overhead of QSV in  the adversarial scenario. }
\end{figure}
By virtue of \lsref{lem:hpnmin} and \ref{lem:hpnt},
it is easy to verify that all three bounds in \eref{eq:AdvNaRatio} decrease monotonically with   $1/\epsilon$, $1/\delta$, and $1/\nu$, as illustrated in \fsref{fig:p*}-\ref{fig:ratio}. \Thref{thm:AdvOHGen} has  profound implications for QSV in the adversarial scenario. With the help of the trivial test, the number of required tests
can achieve 
 the same scaling behaviors with $\epsilon^{-1}$ and $\delta^{-1}$ as the   counterpart for  the nonadversarial scenario presented in \eref{eq:NumTest} and \eref{eq:NumTestAF}.  
 The overhead is at most four times  when $\epsilon,\delta\leq 1/4$ and three times when $\epsilon,\delta\leq 1/10$; furthermore, the overhead  becomes negligible when $\nu,\epsilon,\delta$ approach zero. 
It should be emphasized that our recipe for addressing the adversarial scenario is independent of the specific  construction of the verification protocol once the verification operator is fixed. This fact means that our general results can be applied in various  contexts with different constraints on measurements.
Moreover, the protocol for the adversarial scenario requires the same measurement settings (except for the trivial test) as employed for the nonadversarial scenario, which is the best we can hope for. 
 Therefore, pure states can be verified in the adversarial scenario  with nearly the same efficiency as in the nonadversarial scenario with respect to not only the total  number of tests, but also the number of measurement settings.

Although  the performance of $\Omega$ is very  sensitive to  the smallest eigenvalue $\tau$, surprisingly,
the performance of $\Omega_{p_*}$ is not sensitive to $\tau$ at all. According to \lref{lem:hpnt},
the difference between $h_*(\nu,\tau_1)$ and 
$h_*(\nu,\tau_2)$ for a given $\nu$ is maximized when $\tau_1=0$ [in which case $h_*(\nu,\tau_1)=h_*(\nu)$, cf.~\eref{eq:OverheadPG2}] and $\tau_2=1-\nu$ [cf.~\eref{eq:hstartaumax}].  Calculation shows that the difference between $h_*(\nu)$ and 
$h_*(\nu,1-\nu)$ is less than $12\%$, and it is even smaller when $\nu$ is close to zero or close to 1, as illustrated in \fref{fig:p*}. Therefore, the influence of $\tau$ on the performance of $\Omega_{p_*}$  can be neglected to a large extent.
Moreover, the probability $p$ for performing the trivial test can be chosen without even knowing the value of $\tau$, while achieving nearly optimal performance. 
Actually,  both the choices $p=p_*(\nu)$ and $p=p_0(\nu)=\nu/\rme$ are nearly optimal.
These  observations are very helpful to  constructing efficient verification protocols for the adversarial scenario because we can focus on $\nu$ without worrying about the impact of $\tau$ or even knowing the value of $\tau$. Suppose $\Omega$ is a verification operator with the largest possible $\nu$ (under given conditions), then
 $\Omega_{p}$ is guaranteed to be nearly optimal, where $p$ can be chosen to be $p_*(\nu,\tau)$, $p_*(\nu)$, or $p_0(\nu)=\nu/\rme$. 
Without this insight, it would be much more difficult to devise efficient verification protocols.

\section{\label{sec:app}Applications}
Our recipe presented in \sref{sec:recipe}  can be applied to verifying any pure state in the adversarial scenario as long as we can construct a verification strategy for the nonadversarial scenario.  In this section we discuss the applications of this recipe to verifying many important quantum states, some of which have already been published or appeared on arXiv \cite{ZhuH19O,ZhuH19E,LiHZ19,LiHZ19O,HayaT19}.  The main results are summarized in \tref{tab:NumTest}.
All verification strategies considered here are based on (adaptive) local projective measurements together with classical communication, which are most convenient for practical applications, although our general recipe for the adversarial scenario is independent of how the verification strategy is constructed. The results presented here are also very useful  to verifying quantum gates \cite{ZhuZ19,LiuSYZ19}.

\subsection{Minimum measurement settings for verifying multipartite pure states}
Before considering specific quantum states,  it is instructive to clarify the limitation of local measurements in general. 
As a first step towards this goal, we determine  the minimum number of  measurement settings for each party required to verify a general multipartite pure state that is genuinely multipartite  entangled (GME). Recall that a multipartite pure state is GME if it cannot be expressed as a tensor product of two pure states \cite{GuhnT09}. The following proposition sets a fundamental lower bound for the number of measurement settings required by each party; see \aref{sec:MiniMeas} for a proof.

\begin{proposition}\label{pro:MiniMeas}	
	To verify a multipartite pure state  with adaptive local projective measurements, each party needs at least two  measurement settings, unless the party is not entangled with other parties.
\end{proposition}
Here we do not assume that the test operators are projectors. In general many different test operators can be constructed from a given measurement setting using different data-processing methods. If a party is not entangled with other parties, then its reduced state  is a pure state and the party needs to perform only one projective measurement with the pure state as a basis state.

As an implication of \pref{pro:MiniMeas}, each party needs at least two  measurement settings when the state is GME. It turns out two measurement settings for each party are also sufficient for verifying many important quantum states, such as  bipartite maximally entangled states  \cite{ZhuH19O},  stabilizer states (including graph states) \cite{PallLM18,ZhuH19E}, hypergraph states \cite{ZhuH19E}, and Dicke states \cite{LiuYSZ19}. Nevertheless, more measurement settings can often improve the efficiency with respect to the total number of tests.

\begin{table*}
	\caption{\label{tab:NumTest}
Verification of  bipartite and multipartite quantum states using local projective measurements. The second column shows spectral gaps of efficient verification strategies (not necessarily optimal) for the nonadversarial scenario.  The third column indicates whether  homogeneous strategies with  given spectral gaps can be constructed. The last two columns show the 
numbers of tests required to verify these states within infidelity $\epsilon$ and significance level $\delta$ in the nonadversarial scenario ($N_\na$) and adversarial scenario ($N$), respectively.  Strategies for the adversarial scenario can be constructed using the recipe presented in \sref{sec:recipe}. Here $d$ is the local dimension, $n$ is the number of parties, and $\chi(G)$ is the chromatic number of the hypergraph or weighted graph $G$. 
For bipartite pure states and stabilizer states, the table only shows the results in the worst case. 
	}	
	\begin{math}
	\begin{array}{c|cccc}
	\hline\hline
\mbox{Quantum states}	&\nu(\Omega) &\mbox{homogeneous} &N_\na & N\\[0.5ex]
	\hline
	\mbox{Maximally  entangled states}& \frac{d}{d+1} &\mbox{yes} &\lceil\frac{d+1}{d}\epsilon^{-1}\ln\delta^{-1}\rceil & \lceil\rme \epsilon^{-1}\ln\delta^{-1}\rceil\\[0.5ex]
	\mbox{Bipartite pure states} & \frac{2}{3}&\mbox{yes} &\lceil\frac{3}{2}\epsilon^{-1}\ln\delta^{-1}\rceil & \lceil\rme \epsilon^{-1}\ln\delta^{-1}\rceil \\[0.5ex]
	\mbox{GHZ states}& \frac{d}{d+1}&\mbox{yes} &\lceil\frac{d+1}{d}\epsilon^{-1}\ln\delta^{-1}\rceil & \lceil\rme \epsilon^{-1}\ln\delta^{-1}\rceil\\[0.5ex]
	\mbox{Qubit stabilizer states}& \frac{1}{2} &\mbox{yes} &	
	\lceil 2\epsilon^{-1}\ln\delta^{-1} \rceil& \lceil 2(\ln 2)^{-1}\epsilon^{-1}\ln\delta^{-1}\rceil \\[0.5ex]
		\mbox{Qudit stabilizer states ($d$ odd prime)}& \frac{d-1}{d} &\mbox{yes} &	
	\lceil \frac{d}{d-1}\epsilon^{-1}\ln\delta^{-1} \rceil& \lceil \rme\epsilon^{-1}\ln\delta^{-1}\rceil \\[0.5ex]
	\mbox{Hypergraph state } |G\> & \chi(G)^{-1}& \mbox{no}& \lceil\chi(G)\epsilon^{-1}\ln\delta^{-1}\rceil & \lfloor [\chi(G)+\rme-1]\epsilon^{-1}\ln\delta^{-1}\rfloor\\[0.5ex]
	\mbox{Weighted graph state } |G\> & \chi(G)^{-1}& \mbox{no}& \lceil\chi(G)\epsilon^{-1}\ln\delta^{-1}\rceil & \lfloor [\chi(G)+\rme-1]\epsilon^{-1}\ln\delta^{-1}\rfloor\\[0.5ex]
	\mbox{Dicke states } (n=3) & \frac{1}{3} &\mbox{no} & \lceil 3\epsilon^{-1}\ln\delta^{-1}\rceil & \lfloor 4.1\epsilon^{-1}\ln\delta^{-1}\rfloor\\[0.5ex]	
	\mbox{Dicke states } (n\geq 4) &(n-1)^{-1}&\mbox{no} & \lceil (n-1)\epsilon^{-1}\ln\delta^{-1}\rceil & \lfloor(n+\rme-2)\epsilon^{-1}\ln\delta^{-1}\rfloor\\[0.5ex]
	\hline\hline
	\end{array}	
	\end{math}	
\end{table*}

\subsection{Maximally entangled states and GHZ states}
First, consider bipartite maximally entangled states in dimension $d\times d$, which have the form
\begin{equation}
|\Phi\>=\frac{1}{\sqrt{d}}\sum_{j=0}^{d-1} |jj\>
\end{equation}
up to some local unitary transformations. According to \rscite{HayaMT06,ZhuH19O}, the maximum spectral gap  of any verification strategy $\Omega$ based on LOCC or  separable measurements is 
\begin{equation}\label{eq:vMEs}
\nu(\Omega)=\frac{d}{d+1}. 
\end{equation}
Thanks to \eref{eq:NumTestAF}, the minimum number of tests required to verify $|\Phi\>$ within infidelity $\epsilon$ and significance level $\delta$ in the nonadversarial scenario reads 
\begin{equation}
N_\na=\biggl\lceil\frac{ \ln \delta}{\ln[1-d(d+1)^{-1}\epsilon]}\biggr\rceil\leq \left\lceil \frac{d+1}{d\epsilon}\ln\delta^{-1}\right\rceil. 
\end{equation}
Here the upper bound is nearly tight when $\epsilon$ is small, so we will neglect such small difference in favor of a simpler expression in the following discussions. 
In addition, the verification operator $\Omega$ is necessarily homogeneous when $\nu(\Omega)$ attains the upper bound $d/(d+1)$. So the strategy can be employed for fidelity estimation by \eref{eq:FEhom}.  According to \eref{eq:Fstd}, the standard deviation of this estimation reads
\begin{equation}\label{eq:stdMes}
\Delta F=\frac{\sqrt{p(1-p)}}{\nu(\Omega)\sqrt{N}}
= \frac{\sqrt{(1-F)(F+d^{-1})}}{\sqrt{N}}, 
\end{equation}
where $p=\tr(\Omega\sigma)=\nu(\Omega) F+\beta(\Omega)$.

 By adding the trivial test with a suitable probability, any homogeneous strategy $\Omega$ with $\nu(\Omega)\leq d/(d+1)$ [that is, $\beta(\Omega)\geq 1/(d+1)$] can be constructed using LOCC. In particular, we can construct a homogeneous strategy $\Omega$ with $\beta(\Omega)=1/\rme$, which is 
optimal for high-precision verification in the adversarial scenario according to \sref{sec:Homo}.  Then the number of required tests satisfies
\begin{equation}\label{eq:MESadv}
N \leq  \left\lceil\rme \epsilon^{-1}\ln\delta^{-1}\right\rceil  
\end{equation}
by \thref{thm:NumTestHomoBounds}. When $\delta \leq 1/\rme$, the above bound can be strengthened by  \eref{eq:NumberBoundAdvHomo3}, which yields  $N <  \rme \epsilon^{-1}\ln\delta^{-1}$.  This bound is nearly tight in the high-precision limit.

\Esref{eq:vMEs}-\eqref{eq:MESadv} above also apply to the $n$-qudit GHZ state for $n\geq 3$ as shown in 
\rcite{LiHZ19O}.

\subsection{Bipartite pure states}

Next, consider a general bipartite pure state of the form $\ket{\Psi} = \sum_{j=0}^{d-1} s_j \ket{jj}$, where the Schmidt coefficients $s_j$ are arranged in decreasing order and satisfy the condition $\sum_{j=0}^{d-1} s_j^2=1$. When $d=2$, by virtue of adaptive measurements with two-way communication, one can construct a verification operator $\Omega$ with
 spectral gap $(1+s_0s_1)^{-1}$, which attains the maximum over separable measurements \cite{WangH19}. 
For a general bipartite pure state,
the  spectral gap achievable so far is \cite{LiHZ19,YuSG19}
\begin{equation}
\nu(\Omega)=\frac{2}{2+s_0^2+s_1^2}\geq \frac{2}{3}. 
\end{equation}
With this strategy, the number of tests required for the nonadversarial scenario reads
\begin{equation}
N_\na=\left\lceil \frac{2+s_0^2+s_1^2}{2\epsilon}\ln\delta^{-1}\right\rceil\leq \left\lceil \frac{3}{2\epsilon}\ln\delta^{-1}\right\rceil. 
\end{equation}
Moreover, this strategy can be turned into a homogeneous strategy with the same spectral gap \cite{LiHZ19}, which is useful for fidelity estimation
by \esref{eq:FEhom} and  \eqref{eq:Fstd}.
The standard deviation of this estimation satisfies
\begin{equation}
\Delta F=\frac{\sqrt{p(1-p)}}{\nu(\Omega)\sqrt{N}}
\leq  \frac{\sqrt{(1-F)(F+2^{-1})}}{\sqrt{N}}, 
\end{equation}
where $p=\tr(\Omega\sigma)=\nu(\Omega) F+\beta(\Omega)$ and the inequality follows from the inequality $\nu(\Omega)\geq 2/3$, given that the standard deviation decreases monotonically with $\nu(\Omega)$.

By adding the trivial test with a suitable probability, any homogeneous strategy $\Omega$  with  $\nu(\Omega)\leq 2/(2+s_0^2+s_1^2)$ can be constructed using LOCC \cite{LiHZ19}.   In particular, we can construct a homogeneous strategy $\Omega$ with $\beta(\Omega)=1/\rme$ [that is, $\nu(\Omega)=1-(1/\rme)$], which is 
optimal for high-precision verification in the adversarial scenario,  so \eref{eq:MESadv} also applies to general bipartite pure states. Despite the simplicity of bipartite pure states, we are not aware of any other protocol for verifying them in the adversarial scenario that does not rely on our result. Note that self-testing can only verify a pure state up to some local isometry \cite{MayeY04,ColaTS17,SupiB19}, which is different from what we consider here.

\subsection{\label{sec:Stabilizer}Stabilizer states}
For stabilizer  states, which are equivalent  to graph states under local Clifford transformations \cite{Schl02,GrasKR02}, several verification protocols are known in the literature \cite{HayaM15,FujiH17,HayaH18, PallLM18,MarkK18}. If  the total number of tests is the main figure of merit, then the protocol introduced by PLM \cite{PallLM18} is an ideal choice. Recall that each $n$-qubit stabilizer state $|G\>$ is uniquely determined by $n$ commuting stabilizer generators in the Pauli group, which generate the stabilizer group of order $2^n$. The PLM protocol is composed of  $2^n-1$ projective tests associated with  $2^n-1$ nontrivial stabilizer operators of $|G\>$. 
The corresponding verification operator  reads \cite{PallLM18} 
\begin{align}\label{eq:OmegaPLM}
\Omega_{\rm PLM} = |G\>\<G|+ \frac{2^{n-1}-1}{2^n-1}(\id-|G\>\<G|),
\end{align}
which is homogeneous with
\begin{equation}\label{eq:vPLM}
\beta(\Omega_{\rm PLM})=  \frac{2^{n-1}-1}{2^n-1}\leq \frac{1}{2},\quad   \nu(\Omega_{\rm PLM})= \frac{2^{n-1}}{2^n-1}\geq\frac{1}{2}.
\end{equation}
To verify $|G\>$ within infidelity $\epsilon$ and significance level $\delta$, the number of tests required by  this protocol is 
\begin{equation}
\lceil2^{1-n}(2^n-1)\epsilon^{-1}\ln\delta^{-1}\rceil\leq\lceil 2\epsilon^{-1}\ln\delta^{-1}\rceil,
\end{equation}
which is almost independent of the number $n$ of qubits especially when $n$ is large. Since the strategy  in \eref{eq:OmegaPLM} is homogeneous, it can also be applied for fidelity estimation by virtue of \esref{eq:FEhom} and  \eqref{eq:Fstd}. The standard deviation of this estimation satisfies
\begin{equation}
\Delta F=\frac{\sqrt{p(1-p)}}{\nu\sqrt{N}}\leq \frac{\sqrt{1-F^2}}{\sqrt{N}},
\end{equation}
where $p=\tr(\Omega\sigma)=\nu F+\beta$,  $\nu=\nu(\Omega_{\rm PLM})\geq 1/2$ and $\beta=\beta(\Omega_{\rm PLM})\leq 1/2$.

When adapted to the adversarial scenario,  the strategy in \eref{eq:OmegaPLM} is nearly optimal thanks to 
\thref{thm:NumTestHomoBounds} and \eref{eq:limepdel};  the number of required tests satisfies
\begin{align}	N \leq \biggl\lceil\frac{\ln\delta}{(\beta\ln \beta)\epsilon}\biggr\rceil\leq 
 \biggl\lceil\frac{2\ln\delta^{-1}}{(\ln2)\epsilon}\biggr\rceil<\biggl\lceil\frac{2.89\ln\delta^{-1}}{\epsilon}\biggr\rceil.
\end{align}
Here the latter two upper bounds are independent of the number of qubits and the specific stabilizer state (or graph state). Moreover, the scaling behaviors in $\epsilon$ and $\delta$ are both optimal. Such a high efficiency in the adversarial scenario is achieved for the first time. Previously, the best protocol for the adversarial scenario (without using our recipe)  required   $\lceil m^3/(\delta\epsilon)\rceil$ tests ($\lceil n^3/(\delta\epsilon)\rceil$ tests in the worst case) when $|G\>$ is a graph state whose underlying graph $G$  is $m$-colorable \cite{HayaH18,ZhuH19E}.

\subsection{\label{sec:StabilizerQudit}Qudit stabilizer states}
Here we introduce an efficient  protocol for verifying qudit stabilizer states (including qudit graph states), assuming that the local dimension $d$ is a prime. 
Our protocol reduces to the PLM protocol \cite{PallLM18} for qubit stabilizer states ($d=2$). Let $|G\>$ be a stabilizer state of $n$-qudits. The stabilizer group $S$ of $|G\>$ is composed of all qudit Pauli operators that stabilize $|G\>$ and is isomorphic to the group $\bbZ_d^{n}$, where $\bbZ_d$ is the field of integers modulo $d$. Note that $\bbZ_d^{n}$ is also an $n$ dimensional vector space over $\bbZ_d$. 
The stabilizer group  can be generated by $n$ commuting Pauli operators, say,  $K_1, K_2,\ldots, K_n$, which satisfy $K_r^d=1$ for  $r=1,2,\ldots,n$.  Each stabilizer operator in $S$ has the form $\prod_{r=1}^n K_r^{k_r}$ with $\bfk:=(k_1,k_2, \ldots, k_n)\in \bbZ_d^n$.  If $\bfk= (0,0,\ldots,0)$, then this stabilizer operator  is equal to the identity operator; otherwise, it has $d$ distinct eigenvalues $\omega^j$ for $j=0,1,\ldots,d-1$, where $\omega=\rme^{2\pi\rmi/d}$ is a primitive $d$th root of unity.

For each nonzero element $\bfk$ in $\bbZ_d^n$
we can construct a test for $|G\>$ by measuring the stabilizer operator $\prod_{r=1}^n K_r^{k_r}$:  each party performs a Pauli measurement determined by the decomposition of $\prod_{r=1}^n K_r^{k_r}$ in terms of local Pauli operators. The test is passed if the outcome corresponds to the eigenspace of $\prod_{r=1}^n K_r^{k_r}$ with eigenvalue 1. The corresponding test projector reads 
\begin{align} P_\bfk=\frac{1}{d}\sum_{j=0}^{d-1}\bigg(\prod_{r=1}^n K_r^{k_r}\bigg)^j. 
\end{align}
Note that  $j\bfk$ for $j\in \bbZ_d$ will lead to the same measurement and test operator. Moreover, $P_{\bfk'}=P_\bfk$ iff $\bfk'=j\bfk$ for some $j\in \bbZ_d$ with $j\neq0$ (this conclusion may fail if $d$ is not a prime, and that is why we assume that  $d$ is a prime). So each test corresponds to a line in $\bbZ_d^n$  that passes through the origin, and vice versa. In total $(d^n-1)/(d-1)$ distinct tests can be constructed in this way.

A verification protocol for $|G\>$ can be constructed by performing all distinct tests $P_\bfk$ randomly each with  probability $(d-1)/(d^n-1)$. The resulting verification operator reads
\begin{align}\label{eq:OmegaHomoQudit}
\Omega&= \frac{1}{d^n-1}\sum_{\bfk\in \bbZ_d^n, \; \bfk\neq (0,0,\ldots,0)} P_\bfk\nonumber\\
&= |G\>\<G|+ \frac{d^{n-1}-1}{d^n-1}(\id-|G\>\<G|),
\end{align}
which is homogeneous with
\begin{equation}
\beta(\Omega)=\frac{d^{n-1}-1}{d^n-1}\leq \frac{1}{d},\quad
\nu(\Omega)=\frac{d^{n}-d^{n-1}}{d^n-1}\geq \frac{d-1}{d}.    
\end{equation}
The number of tests required by  this protocol is
\begin{equation}
\biggl\lceil\frac{d^n-1}{d^{n}-d^{n-1}}\epsilon^{-1}\ln\delta^{-1}\biggr\rceil\leq\Bigl\lceil \frac{d}{d-1}\epsilon^{-1}\ln\delta^{-1}\Bigr\rceil,
\end{equation}
which decreases monotonically with the local dimension~$d$. Surprisingly, qudit stabilizer states with $d>2$ (assuming $d$ is a prime) can be verified more efficiently than qubit stabilizer states.

Similar to the qubit case, the above protocol can be applied  for fidelity estimation.  According to \eref{eq:Fstd}, the standard deviation of this estimation satisfies
\begin{equation}
\Delta F=\frac{\sqrt{p(1-p)}}{\nu\sqrt{N}}\leq \frac{\sqrt{(1-F)[F+(d-1)^{-1}]}}{\sqrt{N}}
\end{equation}
given that $\nu\geq (d-1)/d$, where $p=\tr(\Omega\sigma)=\nu F+\beta$.

By adding the trivial test with a suitable probability we can construct any homogeneous verification operator $\Omega$ for $|G\>$ with $\frac{d^{n-1}-1}{d^n-1}\leq\beta(\Omega)<1$ using LOCC. When $d$ is an odd prime, 
we can construct a homogeneous verification operator $\Omega$ with $\beta(\Omega)=1/\rme$,
which is optimal 
 for the adversarial scenario in the high-precision limit. Then the number of required tests  satisfies $N \leq  \left\lceil\rme \epsilon^{-1}\ln\delta^{-1}\right\rceil$ as in \eref{eq:MESadv}.

The verification protocol presented above is also highly efficient for certifying GME. Suppose $|G\>$ is a qudit graph state associated with a connected graph, where the local dimension $d$ is a prime. Then $|G\>$ is GME; in addition,
$\rho$ is GME if its  fidelity with $|G\>$ is  larger than $1/d$. 
In general, to certify the GME of the graph state $|G\>$ with significance level $\delta$, we need to guarantee $\<G|\rho|G\> > 1/d$ with significance level $\delta$. 
Given a verification strategy $\Omega$, then 
it suffices to perform
\begin{equation}\label{eq:NumTestGME}
N=\biggl\lceil
\frac{\ln \delta}{\ln[1-(d-1)\nu(\Omega)/d]}\biggr\rceil
\end{equation}
tests according to \eref{eq:NumTestAF} with $\epsilon=(d-1)/d$.
For the strategy in \eref{eq:OmegaHomoQudit}, we have $\nu(\Omega)\geq(d-1)/d$, so the minimum number of tests satisfy
\begin{equation}
N\leq \biggl\lceil
\frac{\ln \delta}{\ln[1-(d-1)^2/d^2]}\biggr\rceil=\biggl\lceil
\frac{\ln \delta}{\ln[(2d-1)/d^2]}\biggr\rceil. 
\end{equation}
Surprisingly, only one test is required to certify the GME of $|G\>$ when $\delta \geq (2d-1)/d^2$, that is, $d\geq (1+\sqrt{1-\delta})/\delta$. 

In the adversarial scenario, we can construct a  homogeneous strategy $\Omega$ with $\beta(\Omega)=2/(d+1)$ using local projective measurements according to the above analysis. Thanks to \crref{cor:SingleCopyAdvHomo} with $\epsilon=(d-1)/d$, then the GME of $|G\>$ can be certified  using only one test as long as the significance level satisfies $\delta\geq 4d/(d+1)^2$, that is, $d\geq(2+2\sqrt{1-\delta}-\delta)/\delta$ (cf.~Theorem~3 in \rcite{ZhuH19O}). According to \crref{cor:SingleCopyAdvCon1}, the lower bound for $\delta$ cannot be decreased if $d\geq 5$ and if we can perform only one test.  Therefore, the GME of a connected graph state can be certified with any given significance level using only one test as long as the local dimension $d$ is large enough, assuming $d$ is a prime. Previously, a similar result was known only for GHZ states~\cite{LiHZ19O}.

\subsection{Hypergraph states}
A hypergraph $G=(V,E)$ is characterized by a set $V$ of vertices and a set $E$ of hyperedges \cite{QuWLB13,RossHBM13}. For each hypergraph $G$, one can construct a hypergraph state by preparing the state $|+\>=(|0\>+|1\>)/\sqrt{2}$ for each vertex of $G$ and then applying the generalized controlled-$Z$ operation  on the vertices of each hyperedge $e\in E$ \cite{QuWLB13,RossHBM13,ZhuH19E}. As a generalization of graph states, hypergraph states are very useful to quantum computation and foundational studies. 

Recently, the authors proposed 
an efficient protocol---the cover protocol---for verifying  general hypergraph states, which requires only Pauli $X$ and $Z$ measurements for each party \cite{ZhuH19E}. As a special case, a coloring protocol can be constructed for each coloring of the hypergraph $G$.  Suppose $G$ has chromatic number  $\chi(G)$; then the optimal coloring protocol requires only $\chi(G)$ distinct measurement settings and can achieve a  spectral gap of
\begin{equation}\label{eq:GapHypergraph}
\nu(\Omega)=\chi(G)^{-1}\geq [\Delta(G)+1]^{-1}\geq n^{-1},
\end{equation}
where $\Delta(G)$ is the degree of $G$ and $n$ is the number of qubits. Accordingly, the number of required tests reads
\begin{equation}
N_\na=\left\lceil \chi(G)\epsilon^{-1}\ln\delta^{-1}\right\rceil\leq \left\lceil n\epsilon^{-1}\ln\delta^{-1}\right\rceil. 
\end{equation}
This performance is nearly optimal if the chromatic number $\chi(G)$ is small. 
For example, Union Jack states \cite{MillM16} can be verified with a very high efficiency since the chromatic number of the underlying Union Jack lattice  is only~3. These states are particularly interesting because they can realize universal quantum computation under  Pauli measurements \cite{MillM16}.

By virtue of the general recipe presented in \sref{sec:recipe}, we can construct a hedged coloring protocol as characterized by the verification operator $\Omega_p$ with $p=\nu/\rme$ \cite{ZhuH19E}. In the adversarial scenario, the number of tests  required by  $\Omega_p$ satisfies
\begin{align}\label{eq:NumThypergraph}
N&\leq  \frac{[\chi(G)+\rme-1]\ln(F\delta)^{-1}}{\epsilon}\leq \frac{(n+\rme-1)\ln(F\delta)^{-1}}{\epsilon}, 
\end{align}
where $F=1-\epsilon$. The bound is comparable to the counterpart for the nonadversarial scenario especially when $n$ is large. 
The hedged coloring protocol is dramatically more efficient than previous protocols for verifying hypergraph states as proposed in  \rscite{MoriTH17,TakeM18}. For example, the protocol of \rcite{TakeM18} (which improves over \rcite{MoriTH17}) requires more than 
$(2\ln 2) n^3\epsilon^{-18}$ tests when  $\delta=\epsilon$ and $4n\epsilon\leq 1$ (the number of required tests was derived only for  a restricted parameter range) \cite{ZhuH19E}. This number is astronomical even when $n=3$ and $\epsilon=\delta=0.05$. In addition,  the protocol of \rcite{TakeM18} requires adaptive stabilizer tests with $n$ measurement settings. By contrast, the hedged coloring protocol requires at most $\Delta(G)+1$ settings without adaption [the number of settings can be reduced to $\chi(G)$ if an optimal coloring can be found; here we do not count the setting corresponding to the trivial test]. The hedged coloring protocol  is instrumental to realizing verifiable blind MBQC and quantum supremacy. Its high efficiency 
demonstrates the power of our general recipe to constructing efficient verification protocols for the adversarial scenario.

Incidentally, 
the above  results also apply to qudit hypergraph states, including qudit graph states in particular \cite{ZhuH19E}.  For graph states, the hedged coloring protocol is less efficient than the PLM protocol \cite{PallLM18} adapted for the adversarial scenario as discussed in \sref{sec:Stabilizer} and its generalization in \sref{sec:StabilizerQudit}, but requires much fewer measurement settings.

\subsection{Weighted graph states}
Next, consider  weighted graph states \cite{HartCDB07}.
Recently, Hayashi and Takeuchi introduced several  efficient  protocols for verifying the weighted graph state $|G\>$ associated with any weighted graph $G$  \cite{HayaT19}. One of their protocols is based on a coloring of $G$ and adaptive local projective measurements. It can achieve the same spectral gap as in  \eref{eq:GapHypergraph}, that is, $\nu(\Omega)=\chi(G)^{-1}\geq n^{-1}$,
where $\chi(G)$ now refers to the chromatic number of the weighted graph $G$. As in the case of hypergraph states, we can construct a hedged coloring protocol characterized by the verification operator $\Omega_p$ with $p=\nu/\rme$. Then  the number of tests  required by  $\Omega_p$ to verify $|G\>$
in the adversarial scenario satisfies
\begin{align}\label{eq:NumTweightedgraph}
N&\leq  \frac{[\chi(G)+\rme-1]\ln(F\delta)^{-1}}{\epsilon}\leq \frac{(n+\rme-1)\ln(F\delta)^{-1}}{\epsilon} 
\end{align}
as in \eref{eq:NumThypergraph}.
So weighted graph states can be verified with the same efficiency as hypergraph states.

It should be pointed out that the original protocol in \rcite{HayaT19} is based on an earlier version of this paper for dealing with the adversarial scenario (arXiv:1806.05565), so  the scaling behavior of $N$ with the significance level
is suboptimal. The latest results developed in our study as presented in  \sref{sec:recipe} are required  to achieve the optimal scaling behavior shown  in \eref{eq:NumTweightedgraph}. We are not aware of any other protocol for verifying weighted graph states in the adversarial scenario.

\subsection{Dicke states}
Dicke states are another important class of multipartite quantum states which are useful for quantum metrology. 
The $n$-qubit Dicke state with $k$ excitations reads
\begin{equation}\label{eq:Dicke}
\ket{D_n^k}=\binom{n}{k}^{-1/2}\sum_{x\in B_{n,k}} |x\>,
\end{equation}
where $B_{n,k}$ denotes the set of strings in $\{0,1\}^n$ with Hamming weight $k$. To avoid trivial cases, here we assume that $n\geq 3$ and $1\leq k\leq n-1$.  The Dicke state reduces to a $W$ state when $k=1$. 
Recently, Liu et al. \cite{LiuYSZ19} proposed an efficient protocol for verifying the Dicke state,  which can achieve a spectral gap of 
\begin{equation}
\nu(\Omega)=\begin{cases}
\frac{1}{3}, &n=3,\\
\frac{1}{n-1}, & n\geq 4.
\end{cases}
\end{equation}
To verify the Dicke state within infidelity $\epsilon$ and significance level $\delta$, the number of required tests reads
\begin{equation}
N_\na=\begin{cases}
\left\lceil 3\epsilon^{-1}\ln\delta^{-1}\right\rceil, &n=3,\\[0.5ex]
\left\lceil (n-1)\epsilon^{-1}\ln\delta^{-1}\right\rceil, & n\geq 4. 
\end{cases} 
\end{equation}

In the adversarial scenario, we can construct a hedged verification strategy $\Omega_p$ with $p=\nu/\rme$ according to the recipe in  \sref{sec:recipe}. 
Thanks to  \thref{thm:AdvOHGen},  the number of  tests required by $\Omega_p$ satisfies
\begin{align}
N\leq \begin{cases}
 4.1\epsilon^{-1}\ln\delta^{-1}, &n=3,\\[0.5ex]
 (n+\rme-2)\epsilon^{-1}\ln\delta^{-1}, & n\geq 4.  
\end{cases} 
\end{align}
This number is comparable to the counterpart for the nonadversarial scenario. To the best of our knowledge, no  protocol  is known previously for verifying general Dicke states in the adversarial scenario, although there are several works on self-testing Dicke states \cite{SupiCAA18, Fade17}.

To summarize the above discussions, by virtue of our recipe presented in \sref{sec:recipe}, optimal verification protocols for the adversarial scenario can be constructed using local projective measurements for all bipartite pure states, GHZ states, and qudit stabilizer states whose local dimension is an odd prime. Nearly optimal protocols can be constructed for qubit stabilizer states and those hypergraph states with small chromatic numbers, including Union Jack states. For general hypergraph states, weighted graph states, and Dicke states, the number of required tests
is only about $n\epsilon^{-1}\ln\delta^{-1}$ as shown in \tref{tab:NumTest}, which is dramatically smaller than what is required by previous verification protocols (whenever such protocols are available).

\section{\label{sec:comparison}Comparison with other approaches}
Before concluding this paper, it is instructive to compare QSV with other approaches for estimating or verifying quantum states, such as  (traditional) quantum state tomography \cite{PariR04Book}, compressed sensing \cite{GrosLFB10}, direct fidelity estimation (DFE) \cite{FlamL11}, and self-testing \cite{MayeY04,SupiB19}.
In this way we hope to put QSV in a wide context, 
but we do not intend to be exhaustive. Here we are mainly interested in the efficiencies of these approaches with respect to the total number of tests, measurements, or copies of the state required to reach a given precision. Before such a comparison, it should be pointed out that different approaches rely on different assumptions and address different problems. So it is impossible to make a completely fair comparison.

In quantum state tomography, compressed sensing, and DFE, we usually assume that  the states prepared in different runs are independent and identical and that the measurement devices are trustworthy.
In addition, many protocols only require local projective measurements or even Pauli measurements, which are usually much easier to implement than other more complicated operations. In QSV,   the measurement devices are still trustworthy, but the states for different runs may be different as long as they are independent (cf.~\sref{sec:QSVre}). In the adversarial scenario, arbitrary correlated or entangled state preparation is allowed. In self-testing, even the measurement devices are not trusted \cite{MayeY04,SupiB19}. The different strengths  of assumptions underlying these approaches are illustrated in \fref{fig:vp}. 

\begin{figure}
	\includegraphics[width=8cm]{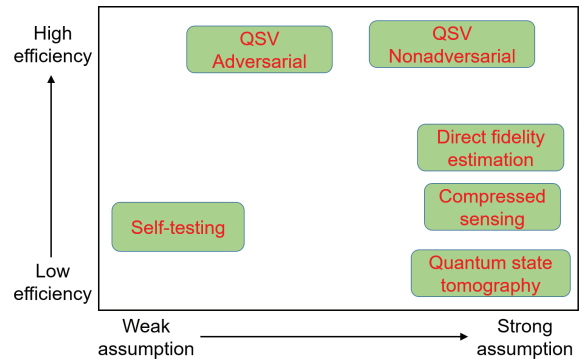}
	\caption{\label{fig:vp}(color online) Qualitative comparison among various approaches for estimating or verifying quantum states with respect to the efficiency and the strength of assumptions. Thanks to the recipe proposed in \sref{sec:recipe},  QSV in the adversarial scenario can achieve nearly the same efficiency as QSV in the nonadversarial scenario, although the underlying assumptions are much weaker.}
\end{figure}

In addition, different approaches address different questions. Quantum state tomography  aims to address the following question: What is the state? To answer this question amounts to reconstructing the density matrix, so the number of parameters to be determined increases exponentially with the number of qubits (here we assume that each subsystem is a qubit for simplicity; the general situation is similar). That is why the resource overhead of tomography increases exponentially with the number of qubits. Compressed sensing addresses a similar question, and so cannot avoid the exponential scaling of resource costs. Nevertheless, it can  reduce the resource overhead significantly by exploiting the structure of quantum states of low ranks \cite{GrosLFB10}.

DFE, QSV, and self-testing address a different type of questions: Is the state identical to the target state, or how close is it? Here the target state is usually a pure state, and the closeness is usually quantified by fidelity or infidelity. Quite often answering these questions is sufficient for many applications in quantum information processing, so it is of fundamental interest to extract such key information efficiently without full tomography. 
 DFE aims to determine the fidelity (infidelity) between the state prepared and the target state \cite{FlamL11}. QSV tries to decide whether the fidelity (infidelity) is larger (smaller) than a given threshold, which is  usually easier than fidelity estimation \cite{HayaMT06,Haya09G, PallLM18}. Self-testing can only provide a lower bound for the fidelity up to some local isometry because the measurement devices are not trustworthy, and the conclusion is solely based on the observed probabilities \cite{MayeY04,SupiB19}.

Suppose we can optimize measurement settings and data-processing procedures, then the efficiency of an approach is mainly determined by the strength of the underlying assumptions and the amount of information it extracts. However, in general it is very difficult to determine the efficiency limit of a given approach because it is very difficult to perform such optimization. In addition, it is highly nontrivial to determine the impacts  of various assumptions.

Although DFE is much more efficient than quantum state tomography,   the resource cost still increases exponentially with the number of  qubits, except for some special families of states, such as stabilizer states.
The DFE protocol originally proposed in \rcite{FlamL11} only requires Pauli measurements; it is not clear whether we can avoid the exponential scaling behavior if more general local measurements are taken into account. 
In the case of self-testing,  there are already numerous research works (see the review paper \rcite{SupiB19}); however,  little is known about the resource cost to reach a given precision, especially in the multipartite setting. 
A few known protocols for self-testing multipartite states are highly resource consuming and hardly practical for systems of more than ten qubits. For example, the resource required to self-test  Dicke states increases exponentially with the number of qubits \cite{SupiCAA18, Fade17}. It is still not clear whether this inefficiency is fundamentally inevitable or is due to our lack of imagination.

 In QSV in the nonadversarial scenario, we have shown in \sref{sec:QSVind} that the variation in states prepared in different runs does not incur any resource overhead as long as these states are independent of each other. In other words, as far as the efficiency is concerned, we can assume that these states are identical and independent as assumed in quantum state tomography, compressed sensing, and DFE.
 Moreover, thanks to our recipe presented in \sref{sec:recipe}, pure states can be verified in the adversarial scenario with nearly the same efficiency as in the nonadversarial scenario.
 In many cases, we can even construct optimal protocols, which are quite rare for other approaches. Therefore, we can expect that QSV even in the adversarial scenario is more efficient than DFE and self-testing, as illustrated in \fref{fig:vp}. This is indeed the case 
 for all states for which verification protocols have been found, such as
 bipartite pure states, GHZ states, stabilizer states (including graph states), hypergraph states, weighted graph states, and Dicke states. For example, Dicke states, hypergraph states, and weighted graph states can be verified efficiently in the adversarial scenario, although no efficient DFE or self-testing protocols are available. In the case of general hypergraph states  and weighted graph states, actually, no self-testing protocols are known at all.

As pointed out earlier, it would be unfair to compare QSV with self-testing directly,  but so far the former is the only practical choice for intermediate and large quantum systems especially in the adversarial scenario. Although self-testing has been studied more intensively in the literature \cite{SupiB19}, it is still very difficult to construct efficient self-testing protocols for multipartite states because the measurement devices are not trustworthy. Insight from QSV may be helpful to studying self-testing, and vice versa.
The relations between QSV and self-testing are worth further exploration in the future. In particular, it would be desirable to combine the merits of the two approaches. We hope that our work can stimulate further progresses along this direction.

\section{\label{sec:Sum}Summary}

We presented a comprehensive study of pure-state verification in the adversarial scenario. Notably, we introduced a general method for computing the main figures of merit pertinent to QSV in the adversarial scenario, such as the fidelity   and the number of required tests. 
In addition, we introduced  homogeneous strategies and derived analytical formulas for the main figures of merit of practical interest. The conditions for single-copy verification are also clarified, which are instructive to understanding single-copy entanglement detection.
Moreover,  we proposed a simple, but 
powerful recipe to constructing efficient verification protocols for the adversarial scenario from the counterpart for the  nonadversarial scenario. Thanks to this recipe, any pure state can be verified in the adversarial scenario with nearly the same efficiency as in the  nonadversarial scenario. Therefore, to   verify a pure quantum state efficiently in the adversarial  scenario, it remains to find an efficient protocol for the nonadversarial scenario, which is usually much easier.

Our recipe can readily be applied to the verification of many important quantum states in quantum information processing, including
bipartite pure states, GHZ states, stabilizer states, hypergraph states, weighted graph states, and Dicke states. Recently, efficient protocols based on local projective measurements have been constructed  for verifying these  states in the nonadversarial scenario. 
By virtue of our recipe, all these states can be verified efficiently in the adversarial scenario using local projective measurements. 
These results are instrumental to many
applications in quantum information processing that demand high security requirements, such as blind MBQC and quantum networks.
The potential of our study is to be  unleashed further in the future.

\section*{Acknowledgments}
This work is  supported by  the National Natural Science Foundation of China (Grant No. 11875110). MH is supported in part by
 Fund for the Promotion of Joint International Research (Fostering Joint International Research) Grant No. 15KK0007,
 Japan Society for the Promotion of Science (JSPS) Grant-in-Aid for Scientific Research (A) No. 17H01280, (B) No. 16KT0017, and Kayamori Foundation of Informational Science Advancement.

\appendix
\section*{Appendix}

In this Appendix, we prove many results presented in the main text, including \thsref{thm:FidelityHomo}-\ref{thm:NumTestBounds}, \lsref{lem:etaf0}-\ref{lem:OverheadPG}, and \pref{pro:MiniMeas}. We also present a simpler proof of \eref{eq:PassingProb}, which was originally proved in  \rcite{PallLM18}.

\section{\label{sec:PassProb}Proof of Eq.~(\ref{eq:PassingProb})}
Here we present a simpler proof of \eref{eq:PassingProb}, which was originally proved in  \rcite{PallLM18}. 

\begin{proof}
Suppose the verification operator $\Omega$ has spectral decomposition $\Omega=\sum_{j=1}^D \lambda_j \Pi_j $, where $D$ is the dimension of the Hilbert space $\caH$, $\lambda_j$ are the eigenvalues of $\Omega$ arranged in decreasing order $1=\lambda_1> \lambda_2\geq\cdots \geq \lambda_D\geq 0$,  and  $\Pi_j$ are mutually orthogonal rank-1 projectors with $\Pi_1=|\Psi\>\<\Psi|$. 
Without loss of generality, we may assume that $\sigma$ is diagonal in the eigenbasis of $\Omega$ because both $\tr(\Omega \sigma)$ and $\<\Psi|\sigma|\Psi\>$ only depend on the diagonal elements of $\sigma$ in this basis. Suppose $\sigma=\sum_{j=1}^D x_j \Pi_j$ with $x_j\geq0$ and $\sum_j x_j=1$. Then
\begin{equation}
\<\Psi|\sigma|\Psi\>=x_1,\quad \tr(\Omega \sigma)=\sum_j \lambda_j x_j.
\end{equation}
Therefore,
\begin{align}
&\max_{\<\Psi|\sigma|\Psi\>\leq 1-\epsilon }\tr(\Omega \sigma)=\max_{x_j\geq0,\, \sum_j x_j=1,\,x_1\leq 1-\epsilon} \sum_j \lambda_j x_j\nonumber\\
&= \max_{0\leq x_1\leq 1-\epsilon} x_1+\lambda_2(1-x_1)= 1- \nu(\Omega)\epsilon, 
\end{align}
where  $\nu(\Omega):=1-\beta(\Omega)=1-\lambda_2$. The maximum can be attained when $\sigma=(1-\epsilon)(|\Psi\>\<\Psi|)+\epsilon \Pi_2$.
\end{proof}

\section{\label{sec:zetaEtaProofs}Proofs of \lsref{lem:etaf0} to \ref{lem:MinNumTestDef}}
In this Appendix we prove \lsref{lem:etaf0} to \ref{lem:MinNumTestDef}  in \sref{sec:CompMFM}.

\subsection{Proofs of \lsref{lem:etaf0} to \ref{lem:zetaEtaMonoCon}}
\begin{proof}[Proof of \lref{lem:etaf0}]
	Let  $\rho$ be an arbitrary  permutation-invariant diagonal density matrix on $\caH^{\otimes (N+1)}$ with decomposition  $\rho=\sum_{\bfk\in \scrS_N} c_{\bfk} \rho_{\bfk}$,  where $c_\bfk$ form a probability distribution on $\scrS_N$. Recall that $\scrS_N$ is  the set of  all sequences  $\bfk=(k_1,k_2, \ldots,k_D)$ of $D$ nonnegative integers that sum up to $N+1$, that is, $\sum_j k_j =N+1$.
 If $f_\rho=0$, then $\zeta_\bfk(\bm{\lambda})=0$ whenever $c_\bfk>0$. Therefore,
	\begin{equation}
\eta(N,0,\Omega)=\max_{\bfk\in \scrS_N}\{\eta_\bfk(\bm{\lambda}) \,|\, \zeta_\bfk(\bm{\lambda})=0\}.
	\end{equation}

To compute $\eta(N,0,\Omega)$,
	we need to determine those $\bfk\in \scrS_N$
	at which $\zeta_\bfk(\bm{\lambda})=0$. By \eref{eq:etazeta}, this condition is satisfied  iff  $k_1=0$, or $\lambda_i =0$ and $k_i\geq 1$ for some $2\leq i\leq D$.  In the first case, we have $\eta_\bfk(\bm{\lambda})\leq \beta^N $, and the inequality is saturated when $\bfk=(0,N+1,0,\ldots,0)$. 
	In the second case, we have
	\begin{equation}
		\eta_\bfk(\bm{\lambda})=\frac{k_i\lambda_i^{k_i-1}}{N+1}\prod_{j\neq i, k_j>0}\lambda_j^{k_j}\leq \frac{1}{N+1},
	\end{equation}
	and the inequality is saturated when $\bfk=(N,0,\ldots,0,1)$.
If $\tau>0$, then only the first case can occur, so we have	$\eta(N,0,\Omega)=\beta^N$. If $\tau=0$, then both cases can occur, so $\eta(N,0,\Omega)=\max\{\beta^N,1/(N+1)\}$. 
In conclusion, we have $\eta(N,0,\Omega)=\delta_\rmc$, which confirms \lref{lem:etaf0}. 
\end{proof}

Next, consider the proofs of \lsref{lem:equiTwoDef} and \ref{lem:zetaEtaMonoCon}. 
From the definitions  in    \esref{eq:Ffdelta} and \eqref{eq:FfdeltaA} together with the results in \esref{eq:zetazero} and \eqref{eq:zetanonzero} we can deduce the following relations.
\begin{subequations}
	\begin{align}
\zeta(N,\delta,\Omega)
&=\min_{\delta'\geq \delta }
\tilde{\zeta}(N,\delta',\Omega)\leq \tilde{\zeta}(N,\delta,\Omega),\\
\eta(N,f,\Omega)
&=\max_{f'\leq f}
\tilde{\eta}(N,f',\Omega)\geq \tilde{\eta}(N,f,\Omega),\\
	F(N,\delta,\Omega)
	&=\min_{\delta'\geq \delta }  \tilde{F}(N,\delta',\Omega)\leq \tilde{F}(N,\delta,\Omega), \\
	\caF(N,f,\Omega)
	&=\min_{f'\geq f}\tilde{\caF}(N,f',\Omega)\leq \tilde{\caF}(N,f,\Omega).
	\end{align}	
\end{subequations}
Therefore, \lsref{lem:equiTwoDef} and \ref{lem:zetaEtaMonoCon} are immediate consequences of  \lref{lem:zetaEtaMonoConA} below.
\begin{lemma}\label{lem:zetaEtaMonoConA}The following statements hold.
\begin{enumerate}
		\item 	$\tilde{\zeta}(N,\delta,\Omega)$ is convex  and nondecreasing in  $\delta$ for $0\leq \delta\leq 1$ and  is strictly  increasing  for $\delta_\rmc\leq \delta\leq 1$.
		
		\item  $\tilde{\eta}(N,f,\Omega)$ is concave and  strictly increasing in $f$ for $0\leq f\leq 1$. 
		
		\item 	$\tilde{F}(N,\delta,\Omega)$ is  nondecreasing in  $\delta$ for $0<\delta\leq 1$ and is strictly  increasing  for  $\delta_\rmc\leq \delta\leq 1$.
		
		\item 	$\tilde{\caF}(N,f,\Omega)$  is strictly  increasing  in $f$ for  $0<f\leq 1$.
\end{enumerate}

\end{lemma}

Here $\delta_\rmc$ is defined in \eref{eq:deltarmc}. The convexity of $\tilde{\zeta}(N,\delta,\Omega)$ means 
\begin{align}\label{eq:zetaConvexity}
\tilde{\zeta}(N,\delta,\Omega)\leq (1-s)\tilde{\zeta}(N,\delta_1,\Omega)+s\tilde{\zeta}(N,\delta_2,\Omega)
\end{align}
for $\delta=(1-s)\delta_1+s\delta_2$ and $0\leq s,\delta_1,\delta_2\leq 1$. Note that this inequality is trivial when $\delta_1=\delta_2$ or $s=0,1$. The concavity  of $\tilde{\eta}(N,f,\Omega)$ means 
\begin{equation}
\tilde{\eta}(N,f,\Omega)\geq (1-s)\tilde{\eta}(N,f_1,\Omega)+s\tilde{\eta}(N,f_2,\Omega)
\end{equation}
for $f=(1-s) f_1+s f_2$ and $0\leq s,f_1,f_2\leq 1$.

\begin{proof}[Proof of \lref{lem:zetaEtaMonoConA}]	The convexity of $\tilde{\zeta}(N,\delta,\Omega)$ in $\delta$ can be proved by virtue of the definition in \eref{eq:FfdeltaA}. Suppose $0\leq \delta_1<\delta_2\leq 1$ and $0<s<1$; let $\delta=(1-s)\delta_1+s\delta_2$.  
If  $\delta_1>\delta_\rmc$,  then there exist two quantum states $\rho_1$ and $\rho_2$ that satisfy
\begin{equation}
	\begin{aligned}
	p_{\rho_1}&=\delta_1,&\quad  f_{\rho_1}&=\tilde{\zeta}(N,\delta_1,\Omega),\\
	p_{\rho_2}&=\delta_2,&\quad  f_{\rho_2}&=\tilde{\zeta}(N,\delta_2,\Omega).
	\end{aligned}
	\end{equation}
	Let $\rho=(1-s)\rho_1+s\rho_2$;  then 
	\begin{equation}
	p_\rho=(1-s)\delta_1+s\delta_2=\delta,
	\end{equation}
	so that
	\begin{align}
\tilde{\zeta}(N,\delta,\Omega)&\leq f_\rho=(1-s)\tilde{\zeta}(N,\delta_1,\Omega)+s\tilde{\zeta}(N,\delta_2,\Omega),\label{eq:zetaConvexityProof}
	\end{align}	
which confirms \eref{eq:zetaConvexity}.

If $\delta_1\leq\delta_\rmc$ and $\delta\leq \delta_\rmc$, then $\tilde{\zeta}(N,\delta,\Omega)=\tilde{\zeta}(N,\delta_1,\Omega)=0$, while $\tilde{\zeta}(N,\delta_2,\Omega)\geq 0$, so  \eref{eq:zetaConvexity} holds.

	If $\delta_1\leq \delta_\rmc$ and $\delta> \delta_\rmc$, then $\tilde{\zeta}(N,\delta_1,\Omega)=0$. Let $\rho_\rmc$ be a quantum state that satisfies $p_{\rho_\rmc}=\delta_\rmc$ and $f_{\rho_\rmc}=0$. 
	Let $s'$ be the solution of the equation $\delta=(1-s')\delta_\rmc+s'\delta_2$, which satisfies $0\leq s'\leq s$. 
	Let $\rho=(1-s')\rho_\rmc+s'\rho_2$. 	Then $p_\rho=\delta$,
	so that
	\begin{align}
	\tilde{\zeta}(N,\delta,\Omega)&\leq
	f_\rho=s'\tilde{\zeta}(N,\delta_2,\Omega)\leq s\tilde{\zeta}(N,\delta_2,\Omega)\nonumber\\
	&=(1-s)\tilde{\zeta}(N,\delta_1,\Omega) + s\tilde{\zeta}(N,\delta_2,\Omega),
	\end{align}		
	which confirms \eref{eq:zetaConvexity} again.	Therefore,
	$\tilde{\zeta}(N,\delta,\Omega)$ is convex in  $\delta$ for $0\leq \delta\leq 1$.

To prove the monotonicity   of $\tilde{\zeta}(N,\delta,\Omega)$ with $\delta$, let $\delta_1, \delta_2$ be  real numbers that satisfy $\delta_\rmc\leq \delta_1<\delta_2\leq 1$. Then 
there  exists a quantum state $\rho_2$ such that $p_{\rho_2}=\delta_2$ and $f_{\rho_2}=\tilde{\zeta}(N,\delta_2,\Omega)>0$. Let $s$ be the solution to the equation $\delta_1=(1-s)\delta_\rmc+s\delta_2$; then  $0\leq  s<1$. Let $\rho=(1-s)\rho_\rmc+s\rho_2$; then $p_\rho=\delta_1$, so that 
	\begin{align}
	\tilde{\zeta}(N,\delta_1,\Omega)&\leq f_\rho=s\tilde{\zeta}(N,\delta_2,\Omega)<\tilde{\zeta}(N,\delta_2,\Omega).
	\end{align}
	It follows that	$\tilde{\zeta}(N,\delta,\Omega)$  is strictly  increasing  in $\delta$ when $\delta_\rmc\leq \delta\leq 1$. As a corollary, $\tilde{\zeta}(N,\delta,\Omega)$  is  nondecreasing  in $\delta$ for  $0\leq \delta\leq 1$ given that 	$\tilde{\zeta}(N,\delta,\Omega)=0$ for $0\leq \delta\leq \delta_\rmc$.

	Next, consider statement~2 in \lref{lem:zetaEtaMonoConA}. The concavity of $\tilde{\eta}(N,f,\Omega)$ follows from a similar reasoning that leads to \eref{eq:zetaConvexityProof}.

	To prove the monotonicity of $\tilde{\eta}(N,f,\Omega)$ over $f$, choose $0\leq f_1< f_2\leq 1$.
	Then there exists a quantum state $\rho_1$ such that $f_{\rho_1}=f_1$ and $p_{\rho_1}=\tilde{\eta}(N,f_1,\Omega)<1$. Choose $\varrho=(|\Psi\>\<\Psi|)^{\otimes (N+1)}$; then $f_\varrho=p_\varrho=1$.
	Let  $s$ be the solution to the equation $f_2=(1-s) f_1+s$; note that $0<s\leq 1$ because of the assumption $f_1<f_2\leq 1$. Let $\rho_2=(1-s)\rho_1+s\varrho$; then $f_{\rho_2}=f_2$,  so that
	\begin{align}
	\tilde{\eta}(N,f_2,\Omega)&\geq p_{\rho_2}=(1-s)\tilde{\eta}(N,f_1,\Omega)+s>\tilde{\eta}(N,f_1,\Omega).
	\end{align}
	Here the second inequality follows from the facts that $0< s\leq 1$ and that $\tilde{\eta}(N,f_1,\Omega)<1$.

	Next, consider statement~3 in \lref{lem:zetaEtaMonoConA}.
	Suppose $\delta_1, \delta_2$ are real numbers that satisfy	
	$\delta_\rmc\leq \delta_1< \delta_2\leq 1$. Then $\tilde{F}(N,\delta_2,\Omega)\geq F(N,\delta_2,\Omega)>0$
	and  there is a quantum state $\rho_2$ such that $p_{\rho_2}=\delta_2$ and $f_{\rho_2}=\delta_2 \tilde{F}(N,\delta_2,\Omega)$. By assumption, $\delta_1$ can be expressed as a convex sum of $\delta_2$ and $\delta_\rmc$, that is, $\delta_1=s \delta_2+(1-s) \delta_\rmc$ with $0\leq s< 1$. Let $\rho_1=s\rho_2+(1-s)\rho_\rmc$, then
	\begin{equation}
	p_{\rho_1}=s\delta_2+(1-s)\delta_\rmc=\delta_1,\quad
	f_{\rho_1}=s f_{\rho_2}=
	s\delta_2 \tilde{F}(N,\delta_2,\Omega),
	\end{equation}
	so that
	\begin{equation}
	\tilde{F}(N,\delta_1,\Omega)\leq\frac{f_{\rho_1}}{p_{\rho_1}}= \frac{s\delta_2 \tilde{F}(N,\delta_2,\Omega)}{s\delta_2+(1-s)\delta_\rmc}< \tilde{F}(N,\delta_2,\Omega).
	\end{equation}
	Therefore,  $\tilde{F}(N,\delta,\Omega)$ is strictly increasing  in $\delta$ whenever $\delta_\rmc\leq \delta\leq 1$. As a corollary, $\tilde{F}(N,\delta,\Omega)$  is  nondecreasing  in $\delta$ for  $0< \delta\leq 1$ given that 	$\tilde{F}(N,\delta,\Omega)=0$ for $0< \delta\leq \delta_\rmc$.

	Finally, consider statement~4 in \lref{lem:zetaEtaMonoConA}.	
	Suppose $f_1$ and $f_2$ are real numbers that satisfy	
	$0< f_1< f_2\leq 1$ and let   $s=f_1/f_2$. Then  $0<s<1$ and
there exists a quantum state $\rho_2$ such that
	$f_{\rho_2}=f_2$ and  $p_{\rho_2}=f_2/\tilde{\caF}(N,f_2,\Omega)$. Let 
	$\rho_1=s\rho_2+(1-s)\rho_\rmc$, where  $\rho_\rmc$ is a quantum state that satisfies $p_{\rho_\rmc}=\delta_\rmc$ and $f_{\rho_\rmc}=0$. Then we have
	\begin{equation}
	f_{\rho_1}=sf_2 =f_1,\quad 
	p_{\rho_1}=sp_{\rho_2}+(1-s)\delta_\rmc,
	\end{equation}
	so that 
	\begin{equation}
	\tilde{\caF}(N,f_1,\Omega)\leq\frac{sf_2}{sp_{\rho_2}+(1-s)\delta_\rmc}< \frac{f_2}{p_{\rho_2}}=\tilde{\caF}(N,f_2,\Omega).
	\end{equation}
	Therefore,  $\tilde{\caF}(N,f,\Omega)$ increases strictly monotonically with   $f$ for $0<f\leq 1$.
\end{proof}

\subsection{Proofs of \lsref{lem:etazetaMU} to \ref{lem:VScompare}}

\begin{proof}[Proof of \lref{lem:etazetaMU}] To prove \eref{eq:etazetaMU} in the lemma, 
	let $f_1=\zeta(N,\delta,\Omega)$ and $\delta_1=\eta(N,f_1,\Omega)$. If $\delta$ satisfies the condition $0\leq \delta\leq \delta_\rmc$, then $f_1=0$ and $\delta_1=\delta_\rmc$ according to \lref{lem:etaf0}, which confirms \eref{eq:etazetaMU}. 
	
	Now suppose $\delta_\rmc<\delta\leq 1$; then  $\max\{\delta,\delta_\rmc\}=\delta$.   In addition, there exists a quantum state $\rho$  on $\caH^{\otimes(N+1)}$ such that $p_\rho=\delta$ and $f_\rho=f_1$, which implies that 
	$\delta_1=\eta(N,f_1,\Omega)\geq \delta$. Meanwhile, there exists a state $\rho'$ such that $f_{\rho'}=f_1$ and $p_{\rho'}=\delta_1$, which implies that $\zeta(N,\delta_1,\Omega)\leq f_1=\zeta(N,\delta,\Omega)$. 
	Since $\zeta(N,\delta,\Omega)$ is strictly  increasing  in $\delta$ for  $\delta_\rmc\leq \delta\leq 1$ according to \lref{lem:zetaEtaMonoCon}, we conclude that  $\delta_1\leq \delta$. This observation  implies that $\delta_1=\delta$ and confirms \eref{eq:etazetaMU} given the opposite inequality derived above.

	Next, consider \eref{eq:zetaetaMU}. 
	Let $\delta_1=\eta(N,f,\Omega)$ and $f_1=\zeta(N,\delta_1,\Omega)$. 
	Then $\delta_1\geq \delta_\rmc$ and	
	there exists a quantum state $\rho$  on $\caH^{\otimes(N+1)}$ such that $f_\rho=f$ and $p_\rho=\delta_1$, which implies that 
	$f_1=\zeta(N,\delta_1,\Omega)\leq f$. Meanwhile, there exists a state $\rho'$ such that $p_{\rho'}=\delta_1$ and $f_{\rho'}=f_1$, which implies that $\eta(N,f_1,\Omega)\geq \delta_1=\eta(N,f,\Omega)$. 
	Since $\eta(N,\delta,\Omega)$ is strictly  increasing  in $f$ for  $0\leq f \leq 1$ according to \lref{lem:zetaEtaMonoCon}, we conclude that  $f_1\geq f$. This observation  implies that $f_1=f$ and confirms \eref{eq:zetaetaMU} given the opposite inequality derived above. 
\end{proof}

\begin{proof}[Proof of \lref{lem:ZetaEtaMonN}] Recall that 	$\zeta(N,\delta,\Omega)$ is convex and nondecreasing in $\delta$  according to \lref{lem:zetaEtaMonoCon}. In addition, $\zeta(N,\delta,\Omega)$ is a piecewise-linear function of $\delta$, and each turning point is equal to $\eta_\bfk$ for some $\bfk\in \scrS_N$ at which $\zeta(N,\delta=\eta_\bfk,\Omega)=\zeta_\bfk$ (cf.~\lref{lem:ExtremalPoints} below). Here  $\eta_\bfk$ and $\zeta_\bfk$ are  shorthands for $ \eta_\bfk(\bm{\lambda})$ and $\zeta_\bfk(\bm{\lambda})$, respectively, which are defined in \eref{eq:etazeta}. 
	To prove \eref{eq:ZetaMonN}, it suffices to prove the inequality $\zeta_\bfk\geq\zeta(N-1,\eta_\bfk,\Omega)$ for each  turning point. 
	
If $k_1=0$, then $\zeta_\bfk=0$, which implies that $\eta_\bfk\leq \delta_\rmc$ according to \lref{lem:etaf0}, so that $\zeta(N-1,\eta_\bfk,\Omega)=0\leq \zeta_\bfk$. 
	
	If $k_1\geq1$, let $\bfk'=(k_1-1, k_2,\ldots, k_D)$. Then 
	\begin{align}
	\eta_{\bfk',N-1}\geq \eta_{\bfk},\quad 
	\zeta_{\bfk',N-1}\leq \zeta_{\bfk},
	\end{align}
	where $\eta_{\bfk',N-1}$ and $\zeta_{\bfk',N-1}$ are given  in \eref{eq:etazeta} with $N$ replaced by $N-1$ and $\bfk$ replaced by $\bfk'$. In conjunction with \lref{lem:zetaEtaMonoCon}  we conclude that
	\begin{align}\label{eq:ZetaMonNProof}
	\zeta(N-1,\eta_\bfk,\Omega)&\leq \zeta(N-1,\eta_{\bfk',N-1},\Omega)\leq \zeta_{\bfk',N-1}\leq \zeta_{\bfk},
	\end{align}
	which implies \eref{eq:ZetaMonN} as desired.

	If  $\delta\leq \delta_\rmc$ then we have $\zeta(N,\delta,\Omega)= \zeta(N-1,\delta,\Omega)=0$. If 
	$\delta=1$ by contrast,  then $\zeta(N,\delta,\Omega)= \zeta(N-1,\delta,\Omega)=1$. So the inequality in \eref{eq:ZetaMonN} is saturated in both cases.

	If the upper bound in \eref{eq:ZetaMonNProof} is saturated, then 
	$\zeta_{\bfk',N-1}= \zeta_{\bfk}$, which implies that $\zeta_{\bfk}=0$ (which means $\eta_{\bfk}\leq \delta_\rmc$) or $\zeta_{\bfk}=1$ (which means $\eta_{\bfk}=1$). So  the upper bound in \eref{eq:ZetaMonNProof} cannot be saturated whenever the turning point satisfies $\delta_\rmc<\eta_\bfk<1$. 
	In conjunction with \esref{eq:zetazero} and \eqref{eq:zetanonzero}, this observation implies that the inequality in \eref{eq:ZetaMonN} 
	is saturated iff $\delta\leq \delta_\rmc$ or $\delta=1$. 
	According to \lref{lem:equiTwoDef},
	\eref{eq:FMonN} and \eref{eq:ZetaMonN} are equivalent, so 
	the same conclusion also applies to \eref{eq:FMonN}.
	
\Eref{eq:EtaMonN} and the equality condition can be derived using a similar reasoning as presented above. \Esref{eq:caFMonN} and \eqref{eq:EtaMonN} are equivalent according to \lref{lem:equiTwoDef}. 
	
Alternatively, \eref{eq:EtaMonN} can be derived from \lsref{lem:etaf0}, \ref{lem:zetaEtaMonoCon}, \ref{lem:etazetaMU}, and \eref{eq:ZetaMonN}. To be specific,
if $f=0$, then $\eta(N,f,\Omega)< \eta(N-1,f,\Omega)$ according to 
 \lref{lem:etaf0}, so \eref{eq:EtaMonN} holds with strict inequality. If $f>0$, then
\begin{gather}
\eta(N,f,\Omega)>\eta(N,0,\Omega)=\delta_\rmc, \\
\eta(N-1,f,\Omega)> \eta(N-1,0,\Omega)>\delta_\rmc,
\end{gather}
according to \lsref{lem:etaf0} and \ref{lem:zetaEtaMonoCon},
where $\delta_\rmc$ is given in \eref{eq:deltarmc}. 
In addition,
\eref{eq:ZetaMonN} and \lref{lem:etazetaMU} imply that 
\begin{align}
&\zeta(N,\eta(N-1,f,\Omega),\Omega)\geq\zeta(N-1,\eta(N-1,f,\Omega),\Omega)\nonumber \\
&=f=
\zeta(N,\eta(N,f,\Omega),\Omega).\label{eq:EtaMonNproof}
\end{align}
In conjunction with \lref{lem:zetaEtaMonoCon}, this equation implies that 
\begin{align}
\eta(N,f,\Omega)\leq \eta(N-1,f,\Omega)
\end{align}
and confirms \eref{eq:EtaMonN}. If the inequality in  \eref{eq:EtaMonN}  is saturated, then the inequality in \eref{eq:EtaMonNproof} is saturated, so that $\eta(N-1,f,\Omega)\leq \delta_\rmc$ or $\eta(N-1,f,\Omega)=1$. The first case cannot happen, while the second case holds iff $f=1$. Therefore,  the inequality in  \eref{eq:EtaMonN} is saturated iff $f=1$. 
\end{proof}

\begin{proof}[Proof of \lref{lem:MinNumTestDef}]
	\Lref{lem:MinNumTestDef}	follows from the definition of $N(\epsilon,\delta,\Omega)$  in \eref{eq:MinNumTestDef} and the fact that the following four inequalities are equivalent,
	\begin{align}
	F(N,\delta,\Omega)&\geq 1-\epsilon,\\
	\zeta(N,\delta,\Omega)&\geq \delta(1-\epsilon),\\
	\eta(N,\delta(1-\epsilon),\Omega)&\leq \delta,\\
	\caF(N,\delta(1-\epsilon),\Omega)&\geq (1-\epsilon).
	\end{align}	
	Here the equivalence of the first two inequalities is a corollary of  \lref{lem:equiTwoDef}; so is the  equivalence of the last two inequalities.  The equivalence of the middle two inequalities follows from \lsref{lem:zetaEtaMonoCon} and \ref{lem:etazetaMU},  note that $\delta> \delta_\rmc$ if either inequality is satisfied.
\end{proof}

\begin{proof}[Proof of \lref{lem:VScompare}] \Eref{eq:VScompare2} is an immediate consequence of \esref{eq:Fzeta} and \eqref{eq:VScompare1}; \eref{eq:VScompare3} is an immediate consequence of \esref{eq:MinNumTestDef} and \eqref{eq:VScompare2}. So to prove \lref{lem:VScompare}, it suffices to prove \eref{eq:VScompare1}.

	By the definition in \eref{eq:Minfdelta} and \eref{eq:prhofrho} we have
	\begin{align}
	&\zeta(N,\delta,\Omega)=\min_{\{c_\bfk\}}\Biggl\{\sum_{\bfk\in \scrS_N} c_\bfk\zeta_\bfk(\bm{\lambda})\Bigg| \sum_{\bfk\in \scrS_N} c_\bfk\eta_\bfk(\bm{\lambda}) \geq \delta\Biggr\}\nonumber\\
	&\geq \min_{\{c_\bfk\}}\Biggl\{\sum_{\bfk\in \scrS_N} c_\bfk\zeta(N,\delta=\eta_\bfk(\bm{\lambda}),\tilde{\Omega})\Bigg| \sum_{\bfk\in \scrS_N} c_\bfk\eta_\bfk(\bm{\lambda}) \geq \delta\Biggr\}\nonumber\\
	&\geq \min_{\{c_\bfk\}}\Biggl\{\zeta\Biggl(N,\sum_{\bfk\in \scrS_N} c_\bfk\eta_\bfk(\bm{\lambda}),\tilde{\Omega}\Biggr)\Bigg| \sum_{\bfk\in \scrS_N} c_\bfk\eta_\bfk(\bm{\lambda}) \geq \delta\Biggr\}\nonumber\\
	&= \min_{\delta' \geq \delta}\zeta(N,\delta',\tilde{\Omega})=\zeta(N,\delta,\tilde{\Omega}), \label{eq:VScompareProof}
	\end{align}	
	which confirms \eref{eq:VScompare1}. 
	Here $\{c_\bfk\}$ is a probability distribution on $\scrS_N$; the first inequality in \eref{eq:VScompareProof} follows from the assumption $\zeta_\bfk(\bm{\lambda})\geq \zeta(N,\delta=\eta_\bfk(\bm{\lambda}),\tilde{\Omega})$ for all $\bfk\in \scrS_N$, and the second inequality follows from the convexity of $\zeta(N,\delta',\tilde{\Omega})$ in $\delta'$ (cf.~\lref{lem:zetaEtaMonoCon}); the last equality follows from the monotonicity of $\zeta(N,\delta',\tilde{\Omega})$ in $\delta'$. 
\end{proof}

By  \eref{eq:zetaetaLP} in the main text,  $\zeta(N,\delta,\Omega)$ and $\eta(N,f,\Omega)$ are piecewise linear functions, whose turning points correspond to  the extremal points of the region $R_{N,\Omega}$, which have the form $(\eta_\bfk(\bm{\lambda}),\zeta_\bfk(\bm{\lambda}))$ for certain $\bfk\in \scrS_N$. In conjunction with the monotonicity and convexity (concavity) of $\zeta(N,\delta,\Omega)$  ($\eta(N,f,\Omega)$) stated in \lref{lem:zetaEtaMonoCon} (see also \lref{lem:etazetaMU}), we can deduce  the following conclusion. Here $\delta_\rmc$ is defined in \eref{eq:deltarmc}. 
\begin{lemma}\label{lem:ExtremalPoints}
	$\zeta(N,\delta,\Omega)$ for $\delta_\rmc\leq \delta\leq 1$ and $\eta(N,f,\Omega)$ for $0\leq f\leq 1$ can be expressed as follows,
	\begin{align}
	\zeta(N,\delta,\Omega)&=\frac{a_{j+1}-\delta}{a_{j+1}-a_j}b_j +\frac{\delta-a_{j}}{a_{j+1}-a_j}b_{j+1},\\
	\eta(N,f,\Omega)&=\frac{b_{l+1}-f}{b_{l+1}-b_l}a_l +\frac{f-b_{l}}{b_{l+1}-b_l}a_{l+1},
	\end{align}
	where $j$ and $l$ are chosen so that $a_j\leq \delta\leq a_{j+1}$ and $b_l\leq f\leq b_{l+1}$. Here $a_j=\eta_{\bfk^{(j)}}(\bm{\lambda})$ and  $b_j=\zeta_{\bfk^{(j)}}(\bm{\lambda})$ with $\bfk^{(j)}\in \scrS_N$ for $j=0,1,\ldots, m$, which satisfy the following conditions
	\begin{gather}
	\delta_\rmc=a_0<a_1<\ldots <a_{m-1}<a_m=1,\\
	0=b_0<b_1<\ldots <b_{m-1}<b_m=1,\\
	0=\frac{b_0}{a_0}<\frac{b_1}{a_1}<\ldots <\frac{b_{m-1}}{a_{m-1}}<\frac{b_m}{a_m}=1.
	\end{gather}	
\end{lemma} 
Note that we can choose strict inequalities $\delta>a_j$ and $f>b_l$ in \lref{lem:ExtremalPoints} if $\delta>\delta_\rmc$ and $f>0$. If  $\Omega$ is a nonsingular  homogeneous strategy  defined in \eref{eq:HomoStrategy} for example, then we have $\delta_\rmc=\lambda^N$, $m=N+1$, $a_j=\eta_{N+1-j}(\lambda)$, and $b_j=\zeta_{N+1-j}(\lambda)$; cf.~\thref{thm:FidelityHomo} in the main text.

\Lref{lem:ExtremalPoints} is very helpful to understanding the properties of $\zeta(N,\delta,\Omega)$ and $\eta(N,f,\Omega)$, although, in general, it is not easy to determine the values of $m$, $\bfk^{(j)}$, $a_j$, and $b_j$. Geometrically, $(a_j, b_j)$ happen to be the extremal points of the region $R_{N,\Omega}$. When $\delta_\rmc=\tau^N$, which can happen iff $\tau=\beta>0$, $R_{N,\Omega}$ has no other extremal point; when $\delta_\rmc>\tau^N$, $R_{N,\Omega}$ has only one additional extremal point, namely $(\tau^N,0)$, as illustrated in \fref{fig:ConvexHull}. This conclusion is tied to \lref{lem:slope} presented in \aref{sec:NIIDproof}.

\section{\label{sec:HomoApp}Homogeneous strategies}

\subsection{Auxiliary results on  homogeneous strategies}
Before proving the results on  homogeneous strategies presented in the main text, we need to introduce a few auxiliary results. For $j,k\in \nni$ and $0<\lambda<1$, define
\begin{align}
g_{jk}(\lambda)=g_{kj}(\lambda)&:=\frac{\zeta_k(\lambda)-\zeta_{j}(\lambda)}{\eta_k(\lambda)-\eta_{j}(\lambda)}, \quad j\neq k, \label{eq:gkj}\\ 
g_{k}(\lambda)&:= g_{k(k+1)}(\lambda)=\frac{\zeta_k(\lambda)-\zeta_{k+1}(\lambda)}{\eta_k(\lambda)-\eta_{k+1}(\lambda)},\label{eq:gk}
\end{align}
where $\eta_k(\lambda)$ and $\zeta_k(\lambda)$ are defined in \eref{eq:etazetaHomo}, assuming that $N$ is a positive integer. 
To simplify the notations, we shall use  $\eta_k$, $\zeta_k$, $g_k$, $g_{kj}$ as shorthands for $\eta_k(\lambda)$, $\zeta_k(\lambda)$, $g_{k}(\lambda)$, $g_{kj}(\lambda)$ if there is no danger of confusions. Geometrically, $g_{jk}$ and $g_{kj}$ denote the slope of the line passing through the two points $(\eta_j, \zeta_j)$ and $(\eta_k,\zeta_k)$. 
\begin{lemma}\label{lem:zetaEtaSlope}
Suppose  $0<\lambda<1$ and  $j,k\in \nni$ with $k<j$. Then
	$g_{k}(\lambda)$ decreases strictly monotonically with $k$, and   $g_{kj}(\lambda)$ decreases strictly monotonically with $j,k$. 
\end{lemma}
\begin{lemma}\label{lem:11slope}
Let $0\leq \lambda<1$ and $k\in\{1,2,\ldots, N+1	\}$. Then 
	\begin{equation}
	\frac{1}{1-\lambda^N}\leq 	\frac{1-\zeta_k(\lambda)}{1-\eta_k(\lambda)}\leq\frac{1+N(1-\lambda)}{N(1-\lambda)}=
	\frac{1+N\nu}{N\nu}.
	\end{equation}
The first inequality is saturated iff $k=N+1$,  or $k\geq 2$ and 	$\lambda=0$; the second inequality is saturated iff $k=1$. 
\end{lemma}
When $0\leq \lambda<1$ (and $N$ is a positive integer), \lref{lem:11slope}  implies that 
\begin{equation}\label{eq:lambdaN}
\frac{1}{1-\lambda^N}<
\frac{1+N\nu}{N\nu},\quad \lambda^N<
\frac{1}{N\nu+1}. 
\end{equation}
The two inequalities actually hold for a wider parameter range according to \lref{lem:alambda} below.
\begin{lemma}\label{lem:alambda}
Suppose $0<\lambda \leq 1$, $\nu=1-\lambda$,  and $a$ is a real number.  Then 
\begin{align}
\lambda^{-a}-a\nu-1&\geq 0 \quad \mbox{if}\; a\geq 0 \mbox{ or } a\leq -1, \label{eq:alambda} \\
\lambda^{-a}-a\nu-1&\leq 0\quad  \mbox{if}\; -1\leq a\leq 0. \label{eq:alambda2}
\end{align}
If $a\neq -1,0$, then the inequality in \eref{eq:alambda} is saturated
iff $\lambda=1$; the same holds for the inequality in \eref{eq:alambda2}. 
\end{lemma}

\begin{lemma}\label{lem:zetadelN}
Let $0<\lambda <1$,   $0\leq \delta \leq 1$, and $k\in \nni$. Then $\zeta(N,\delta,\lambda,k)$ increases strictly monotonically with $\delta$ when $k\leq N+1$. Also, $\zeta(N,\delta,\lambda,k)$ increases strictly monotonically with $N$ except when $\delta=1$ and $k=0$.
\end{lemma} 
Here $\zeta(N,\delta,\lambda,k)$ is defined in \eref{eq:zetadel} in the main text. Note that $\zeta(N,\delta,\lambda,k)=1$ is independent of $N$ and $\lambda$  when $\delta=1$ and $k=0$.

\begin{lemma}\label{lem:zetadel}
Suppose $0<\lambda <1$ and  $0<\delta \leq 1$. Then 
\begin{align}
&\max_{k\in\nni}\zeta(N,\delta, \lambda,k)=\zeta(N,\delta, \lambda,k_*),\label{eq:zetadelMax} \\
&\max \{0,\zeta(N,\delta, \lambda,k_*)\}\!=\!\max\Bigl\{0,\max_{k\in\{0,1,\ldots, N\}}\zeta(N,\delta, \lambda,k)\Bigr\}\nonumber  \\&=\max\{0,\zeta(N,\delta, \lambda,k_+),\zeta(N,\delta, \lambda,k_-)\}, \label{eq:zetadelMax2}
\end{align}
where $k_*$ is the largest integer $k$ that satisfies $\eta_k\geq \delta$, $k_+=\lceil\log_\lambda \delta\rceil$, and $k_-=\lfloor\log_\lambda \delta\rfloor$. In addition,
\begin{align}
\zeta(N,\delta, \lambda,k_*)&\leq 0,\quad  0\leq\delta \leq \lambda^N, \label{eq:zetadelNeg} \\
\zeta(N,\delta, \lambda,k_*)&> 0,\quad  \lambda^N<\delta \leq 1. \label{eq:zetadelpos}
 \end{align}
\end{lemma}

\begin{lemma}\label{lem:NumTestProp}
	Suppose $0<\epsilon,\delta,\lambda< 1$ and $k\in \nni$. Then $\tilde{N}(\epsilon,\delta,\lambda,k)>0$ and it
	decreases strictly monotonically with $\epsilon$ and $\delta$. If  $\delta\leq  \lambda^{k}/(F+\lambda \epsilon)$, then $\tilde{N}(\epsilon,\delta,\lambda,k)>k-1$. 
Given $k\geq 1$,  then  $\tilde{N}(\epsilon,\delta,\lambda,k)\leq \tilde{N}(\epsilon,\delta,\lambda,k-1)$ iff $\delta \leq \lambda^k/(F+\lambda\epsilon)$.  In addition,
\begin{align}
&\min_{k\in \nni}\tilde{N}(\epsilon,\delta,\lambda,k)=\tilde{N}(\epsilon,\delta,\lambda,k^*)\label{eq:NtildeMin}\\
&=\min\bigl\{\tilde{N}_+(\epsilon,\delta,\lambda),\tilde{N}_-(\epsilon,\delta,\lambda)\bigr\}\label{eq:NtildeMin2} \\
&=\begin{cases}
\tilde{N}_-(\epsilon,\delta,\lambda), & \delta \geq \frac{\lambda^{k_+}}{F+\lambda \epsilon},\\
 \tilde{N}_+(\epsilon,\delta,\lambda), & \delta \leq  \frac{\lambda^{k_+}}{F+\lambda \epsilon},
\end{cases}\label{eq:NtildeMin3}
\end{align}
where $k^*$ is the largest integer $k$ that satisfies the inequality $\delta\leq \lambda^{k}/(F+\lambda \epsilon)$ and $\tilde{N}_\pm(\epsilon,\delta,\lambda)=\tilde{N}(\epsilon,\delta,\lambda,k_\pm)$ with $k_+=\lceil \log_\lambda \delta\rceil$ and $k_-=\lfloor \log_\lambda \delta\rfloor$. 		
\end{lemma}
Here $\tilde{N}(\epsilon,\delta,\lambda,k)$ is defined in \eref{eq:NumTestHomoC}.

\begin{lemma}\label{lem:NumTestmBound}
	Suppose $0<\epsilon,\delta,\lambda <1$. Then 
	\begin{equation}\label{eq:NumTestmBound}
	\tilde{N}_-(\epsilon,\delta,\lambda)\leq \frac{F\nu+\lambda}{\lambda\epsilon}k_-+ \frac{\log_\lambda\delta-k_-}{\lambda\epsilon}=\frac{\log_\lambda\delta}{\lambda\epsilon}-\frac{\nu k_-}{\lambda},
	\end{equation}
	where $F=1-\epsilon$, $\nu=1-\lambda$, and $k_-=\lfloor\log_\lambda\delta\rfloor$. The inequality is saturated iff $\log_\lambda\delta$ is an integer.
\end{lemma}

\begin{proof}[Proof of \lref{lem:zetaEtaSlope}]
According to \esref{eq:gkj} and \eqref{eq:gk} as well as 
 the definitions of $\eta_k(\lambda)$ and $ \zeta_k(\lambda)$ in \eref{eq:etazetaHomo},  we have
\begin{align}
	g_k(\lambda)&=\frac{\lambda[1+(N-k)\nu]}{\nu(N\lambda+k\nu)},\\
g_k(\lambda)-g_{k+1}(\lambda)&=\frac{(N+1)\lambda}{[N\lambda+(k+1)\nu](N\lambda+k\nu)}>0,
\end{align}	
where $\nu=1-\lambda$. So 
 $g_{k}(\lambda)$ decreases strictly monotonically with $k$ for $k\in \nni$. 
	
	Simple analysis shows that $g_{kj}(\lambda)$ can be expressed as a weighted average of  $g_{m}(\lambda)$ for $m=k,k+1,\ldots,j-1$, namely, 
	\begin{equation}
	g_{kj}(\lambda)=\sum_{m=k}^{j-1} \frac{\eta_m(\lambda)-\eta_{m+1}(\lambda)}{\eta_k(\lambda)-\eta_{j}(\lambda)}g_m(\lambda).
	\end{equation}
Here the weight for each $g_m(\lambda)$ is strictly positive given that $\eta_m(\lambda)$  decreases strictly monotonically with $m$ for $m\in \nni$. So  $g_{j}(\lambda)< g_{j-1}(\lambda)< g_{kj}(\lambda)<g_{k}(\lambda)$ when $k+1<j$. In addition,  $g_{k(j+1)}(\lambda)$ is a convex sum  of  $g_{kj}(\lambda)$ and  $g_{j}(\lambda)$, that is, 
	\begin{equation}
	g_{k(j+1)}=\frac{(\eta_k-\eta_{j})g_{kj} +(\eta_j-\eta_{j+1})g_{j}}{\eta_k-\eta_{j+1}},
	\end{equation}
	which implies that  $g_{k(j+1)}(\lambda)<g_{kj}(\lambda)$; by the same token we can prove $g_{(k+1)j}(\lambda)<g_{kj}(\lambda)$ when $k+1<j$. Therefore, $g_{kj}(\lambda)$ decreases strictly monotonically with $k$ and $j$. 
\end{proof}

 \begin{proof}[Proof of \lref{lem:11slope}]
 	When $0<\lambda<1$,
 	\lref{lem:11slope} is an immediate consequence of \lref{lem:zetaEtaSlope} given that 
 	\begin{gather}
 \!\!	\eta_0(\lambda)=\zeta_0(\lambda)=1,\;
 	\eta_{N+1}(\lambda)=\lambda^N,\; \zeta_{N+1}(\lambda)=0,\\
 	\eta_1(\lambda)=\frac{1+N\lambda}{N+1},\quad \zeta_1(\lambda)=\frac{N\lambda}{N+1},
 	\end{gather}
 	so that 
 	\begin{align}
 	g_{0k}(\lambda)&=\frac{1-\zeta_{k}(\lambda)}{1-\eta_{k}(\lambda)}=
 	\begin{cases}
 	\frac{1+N(1-\lambda)}{N(1-\lambda)}, & k=1,\\
 	\frac{1}{1-\lambda^N}, & k=N+1.
 	\end{cases}
 	\end{align}

 	When $\lambda=0$, we have $\zeta_0=\eta_0=1$, $\eta_1=1/(N+1)$, $\eta_k=0$ for $k=2,3,\ldots, N+1$, and  $\zeta_k=0$ for  $k=1,2,\ldots, N+1$, in which case \lref{lem:11slope} can be verified explicitly. 
 \end{proof}

\begin{proof}[Proof of \lref{lem:alambda}]
	Note that $\lambda^{-a}-a\nu-1= 0$ when $\lambda=1$, or $a=0$, or $a=-1$. The derivative of $\lambda^{-a}-a\nu-1$  over $\lambda$ reads $a(1-\lambda^{-a-1})$, and it satisfies
\begin{align}
a(1-\lambda^{-a-1})&\leq 0 \quad \mbox{if}\; a\geq 0 \mbox{ or } a\leq -1, \label{eq:alambdaProof}\\
a(1-\lambda^{-a-1})&\geq 0\quad  \mbox{if}\; -1\leq a\leq 0,\label{eq:alambda2Proof}
\end{align}	
which imply the inequalities in \esref{eq:alambda} and \eqref{eq:alambda2}  given that  $\lambda^{-a}-a\nu-1= 0$ when $\lambda=1$.
If $a\neq -1,0$, then the inequality in \eref{eq:alambdaProof} is saturated iff $\lambda=1$, and the same holds for the inequality in \eref{eq:alambda2Proof}. Therefore, both \eref{eq:alambda} and \eref{eq:alambda2} are saturated
iff $\lambda=1$, which completes the proof of \lref{lem:alambda}. 
\end{proof}

\begin{proof}[Proof of \lref{lem:zetadelN}]
	The monotonicity of $\zeta(N,\delta,\lambda,k)$ with $\delta$ follows from the facts that $\zeta(N,\delta,\lambda,k)$ is linear in $\delta$ and that $1+(N-k)\nu>0$ when $k\leq N+1$. 
	
	According to the following equation
	\begin{align}
	&\zeta(N+1,\delta,\lambda,k)-\zeta(N,\delta,\lambda,k)\nonumber\\
	&=\frac{\lambda[\lambda^{k+1}+\delta(k\nu-\lambda)]}{\nu(N\lambda+k\nu)[(N+1)\lambda+k\nu]},
	\end{align}
	to prove the monotonicity of $\zeta(N,\delta,\lambda,k)$ with $N$, it suffices to prove the inequality
	\begin{equation}\label{eq:zetadelNproof}
	\lambda^{k+1}+\delta(k\nu-\lambda)\geq 0, 
	\end{equation}
which  is saturated iff $\delta=1$ and $k=0$. To this end, 	
 it suffices to consider the two special cases $\delta=0$ and $\delta=1$ 	since the left-hand side in \eref{eq:zetadelNproof} is linear in $\delta$. In the first case, the inequality is strict. In the second case, according to \lref{lem:alambda} with $a=-(k+1)$, we have
	\begin{equation}
	\lambda^{k+1}+k\nu-\lambda\geq -(k+1)\nu+1+k\nu-\lambda=0,
	\end{equation}
	and the inequality is saturated iff $k=0$. This observation confirms 
the inequality in \eref{eq:zetadelNproof} and the saturation condition, which in turn confirms \lref{lem:zetadelN}. 
\end{proof}

\begin{proof}[Proof of \lref{lem:zetadel}]
When $k-1\in \nni$, 	by the definition of $\zeta(N,\delta, \lambda,k)$ in \eref{eq:zetadel}, we can derive
	\begin{align}
	&\zeta(N,\delta, \lambda,k)-\zeta(N,\delta, \lambda,k-1)\nonumber\\
	&
	=\frac{\lambda^k[k+(N+1-k)\lambda]-(N+1)\lambda\delta}{(k\nu+N\lambda) [(k-1)\nu+N\lambda]}. 
	\end{align}
	So  $\zeta(N,\delta, \lambda,k)\geq \zeta(N,\delta, \lambda,k-1)$ iff $\delta\leq  \eta_k$ and the inequality is saturated only when $\delta= \eta_k$. Therefore, the maximum of $\zeta(N,\delta, \lambda,k)$ over $k\in\nni$ is attained when $k$ is the largest integer that satisfies $\eta_k\geq \delta$, that is, $k=k_*$, which confirms \eref{eq:zetadelMax}. 
	
	Before proving \eref{eq:zetadelMax2}, we first prove \esref{eq:zetadelNeg} and \eqref{eq:zetadelpos}.  According to \eref{eq:zetadel} in the main text and the definition of $k_*$,  $\zeta(N,\delta, \lambda,k_*)$ is a convex sum of $\zeta_{k_*}(\lambda)$ and $\zeta_{k_*+1}(\lambda)$ in which the weight of $\zeta_{k_*}(\lambda)$ is nonzero. If $0< \delta\leq \lambda^N$, then we have $k_*\geq N+1$, which implies that $\zeta_{k_*}(\lambda)\leq 0$ and $\zeta_{k_*+1}(\lambda)<0$. Therefore, $\zeta(N,\delta, \lambda,k_*)\leq 0$, which confirms \eref{eq:zetadelNeg}. Conversely,
	if $\lambda^N<\delta\leq 1$, then $k_*\leq N$, which implies that $\zeta_{k_*}(\lambda)> 0$ and $\zeta_{k_*+1}(\lambda)\geq 0$. So  $\zeta(N,\delta, \lambda,k_*)> 0$, which confirms \eref{eq:zetadelpos}. 
	
	Alternatively, to prove \eref{eq:zetadelNeg}, we can prove that 
	$\zeta(N,\delta, \lambda,k)\leq 0$ for $k\in \nni$. Given that  $\zeta(N,\delta, \lambda,k)$ is a linear function of $\delta$, it suffices to prove the result when $\delta =0$ and $\delta=\lambda^N$. According to \eref{eq:zetadel}, we have
	\begin{align}
\zeta(N,\delta=0, \lambda,k)&=-\frac{\lambda^{k+1}}{\nu(k\nu+N\lambda)}<0,\\
\zeta(N,\delta=\lambda^N,\lambda,k)
&=	
\frac{\lambda\{\lambda^N[1+(N-k)\nu]-\lambda^k\}}{\nu(k\nu+N\lambda)}\leq 0, \label{eq:zetadelNegProof1}
\end{align}	
which imply  \eref{eq:zetadelNeg}. Here $\nu=1-\lambda$ and the inequality in \eref{eq:zetadelNegProof1} follows from \lref{lem:alambda} with $a=N-k$.

Finally,  we can prove \eref{eq:zetadelMax2}. If $0<\delta \leq \lambda^N$, then \eref{eq:zetadelMax2} follows from \eref{eq:zetadelMax} and the fact that $\zeta(N,\delta, \lambda,k_*)\leq 0$.  If instead $\lambda^N<\delta \leq 1$, then we have $0\leq k_+\leq N$ and $0\leq k_-\leq N-1$; in addition, $\eta_{k_-}(\lambda)\geq \delta $ and $\eta_{1+k_+}(\lambda)< \delta $. Therefore,  $k_*\in \{0,1,\ldots,N\}$ and $k_*$ is equal to either $k_+$ or $k_-$, which implies \eref{eq:zetadelMax2} given \eref{eq:zetadelMax}. 
\end{proof}

\begin{proof}[Proof of \lref{lem:NumTestProp}]
	To prove \lref{lem:NumTestProp}, we first investigate the monotonicity of  
	$\tilde{N}(\epsilon,\delta,\lambda,k)$ defined in \eref{eq:NumTestHomoC} for $0<\epsilon,\delta \leq 1$, $0<\lambda<1$, and $k\in \nni$. 
	The partial derivative of $\tilde{N}(\epsilon,\delta,\lambda,k)$ over $\epsilon$ reads
	\begin{align}\label{eq:NumTestepDerivative}
	\frac{\partial \tilde{N}(\epsilon,\delta,\lambda,k)}{\partial \epsilon}=-\frac{\lambda^{k+1}+\delta(k\nu-\lambda)}{\lambda\nu\delta\epsilon^2}\leq 0, 
	\end{align}
	where the inequality is saturated iff  $k=0$ and $\delta=1$; cf.~\eref{eq:zetadelNproof}.  Therefore, $\tilde{N}(\epsilon,\delta,\lambda,k)$ is strictly decreasing in $\epsilon$ for $0<\epsilon\leq 1$ except when $k=0$ and $\delta=1$, in which case $\tilde{N}(\epsilon,\delta,\lambda,k)=0$.

	Next, the partial derivative of $\tilde{N}(\epsilon,\delta,\lambda,k)$ over $\delta$ reads
	\begin{align}\label{eq:NumTestdelDerivative}
	\frac{\partial \tilde{N}(\epsilon,\delta,\lambda,k)}{\partial \delta}=-\frac{\lambda^k}{\nu\delta^2\epsilon}<0. 
	\end{align}
	So  $\tilde{N}(\epsilon,\delta,\lambda,k)$ is strictly decreasing in $\delta$ for $0<\delta\leq 1$. 
	
	According to the above analysis, 
	\begin{align}\label{eq:NumTestCUB}
	\tilde{N}(\epsilon,\delta,\lambda,k)&\geq \tilde{N}(\epsilon=1,\delta=1,\lambda,k)\nonumber\\
	&=\frac{\lambda^k+k\nu-1}{\nu}\geq0. 
	\end{align}
	Here the first inequality is saturated iff $\epsilon=\delta=1$, or $\delta=1$ and $k=0$; the second inequality is saturated iff $k=0,1$ (cf.~\lref{lem:alambda}). Therefore, $\tilde{N}(\epsilon,\delta,\lambda,k)>0$, except when $\delta=1$ and $k=0$, or $\epsilon=\delta=k=1$. 
Given the assumption $0<\epsilon,\delta<1$, then 
	$\tilde{N}(\epsilon,\delta,\lambda,k)>0$ and $\tilde{N}(\epsilon,\delta,\lambda,k)$ decreases strictly monotonically with $\epsilon$ and~$\delta$.

Next, suppose $0<\epsilon,\delta<1$. If   $\delta\leq  \lambda^{k}/(F+\lambda \epsilon)$ and  $k=0$, then $\tilde{N}(\epsilon,\delta,\lambda,k)>k>k-1$ according to the first statement in \lref{lem:NumTestProp}.  
If instead  $\delta= \lambda^{k}/(F+\lambda \epsilon)$ and $k\geq 1$, then 
\begin{align}
\tilde{N}(\epsilon,\delta,\lambda,k)
=k-1+\frac{kF}{\lambda \epsilon}>k-1.
\end{align}
So $\tilde{N}(\epsilon,\delta,\lambda,k)>k-1$ whenever $\delta\leq  \lambda^{k}/(F+\lambda \epsilon)$	
given that  $\tilde{N}(\epsilon,\delta,\lambda,k)$ is monotonically deceasing in $\delta$.

Next, if $k\geq1$, then
	\begin{align}
	&\tilde{N}(\epsilon,\delta,\lambda,k)- \tilde{N}(\epsilon,\delta,\lambda,k-1)\nonumber\\
	&=\frac{\nu\delta(F\nu+\lambda)+\lambda^{k+1}-\lambda^k}{\lambda\nu\delta\epsilon}=\frac{\delta(F+\lambda\epsilon)-\lambda^k}{\lambda\delta\epsilon},\label{eq::NumTestHomoDif}
	\end{align}
	so $\tilde{N}(\epsilon,\delta,\lambda,k)\leq \tilde{N}(\epsilon,\delta,\lambda,k-1)$ iff $\delta \leq \lambda^k/(F+\lambda\epsilon)$. Consequently,  the minimum of $\tilde{N}(\epsilon,\delta,\lambda,k)$ over $k\in \nni$ is attained when $k$ is the largest integer  that satisfies the inequality $\delta\leq \lambda^{k}/(F+\lambda \epsilon)$, that is, $k=k^*$,
 which  confirms \eref{eq:NtildeMin}. 
	
In addition, we have
\begin{equation}
\frac{\lambda^{k_++1}}{F+\lambda \epsilon}<\lambda^{k_+}\leq \delta\leq \lambda^{k_-}<  \frac{\lambda^{k_-}}{F+\lambda \epsilon},
\end{equation}	
given that $\lambda<F+\lambda\epsilon<1$. So   $k^*$ in \eref{eq:NtildeMin} is equal to either $k_+$ or $k_-$, which implies \eref{eq:NtildeMin2}. Finally, \eref{eq:NtildeMin3} is an easy consequence of \eref{eq::NumTestHomoDif}. 
\end{proof}

\begin{proof}[Proof of \lref{lem:NumTestmBound}]
	The equality in \eref{eq:NumTestmBound} can be verified by straightforward calculation given the equality $F\nu+\lambda=1-\nu\epsilon$.
According to the definitions in \esref{eq:NumTestHomoC} and \eqref{eq:Ntildepm}, we have
	\begin{align}
	\tilde{N}_-(\epsilon,\delta,\lambda)&=\frac{k_-\nu^2 \delta F +\lambda^{k_-+1}+\lambda\delta(k_-\nu-1)}{\lambda\nu\delta \epsilon}\nonumber\\
	&=\frac{F\nu+\lambda}{\lambda\epsilon}k_-+ \frac{\lambda^{k_-+1}-\lambda\delta}{\lambda\nu\delta\epsilon}\nonumber\\
	&=\frac{F\nu+\lambda}{\lambda\epsilon}k_-+ \frac{\lambda^{k_--\log_\lambda\delta+1}-\lambda}{\lambda\nu\epsilon}.
	\end{align}
So the inequality in \eref{eq:NumTestmBound} is equivalent to  the following inequality
	\begin{equation}\label{eq:alambda3}
	\lambda^{1-b}-\lambda-\nu b\leq 0, 
	\end{equation}
	where $ b=\log_\lambda\delta-k_-=\log_\lambda\delta-\lfloor\log_\lambda\delta\rfloor$, which satisfies $0\leq b< 1$. \Eref{eq:alambda3} holds because the function $\lambda^{1-b}-\lambda- \nu b$ is strictly convex in $b$ (given the assumption $0<\lambda<1$) and it is equal to 0 when $b=0$ and $b=1$ (the function is well defined when $b=1$ although this value cannot be attained here). In addition, the inequality in \eref{eq:alambda3} is saturated iff $b=0$, which means  $\log_\lambda\delta$ is an integer. Alternatively, these conclusions  follow from \lref{lem:alambda} with $a=b-1$. 
Therefore, the inequality in \eref{eq:NumTestmBound} is saturated iff $\log_\lambda\delta$ is an integer. 
\end{proof}

\subsection{\label{sec:HomoProof}Proofs of \thsref{thm:FidelityHomo}-\ref{thm:NumTestHomoBounds} and \eref{eq:NumTestHomoLS2}}
\begin{proof}[Proof of \thref{thm:FidelityHomo}]
	According to  \lref{lem:equiTwoDef},  we have $F(N,\delta,\lambda)=\zeta(N,\delta,\lambda)/\delta$. If $\delta\leq \delta_\rmc=\lambda^N$, then we have $\zeta(N,\delta,\lambda)=0$ by \eref{eq:zetazero}. If  $\delta> \lambda^N$, then 
	\begin{equation}\label{eq:FidelityLP2}
	\zeta(N,\delta,\lambda)=
	\min_{0\leq k<j\leq N+1}
(c_j \zeta_j+c_k \zeta_k), 
	\end{equation}
where $\zeta_j, \zeta_k$ are shorthands for $\zeta_j(\lambda), \zeta_k(\lambda)$, and   the parameters $k,j$ are  restricted by the requirements $\eta_k\geq \delta$ and $\eta_j<  \delta$. The coefficients
	 $c_j,c_k$ are determined by the conditions  
	\begin{equation}\label{eq:cjckCondition}
	c_j+c_k=1, \quad c_j \eta_j+c_k \eta_k=\delta,
	\end{equation}
which yield
\begin{equation}
	c_{j}=\frac{\eta_{k}-\delta}{\eta_k-\eta_{j}}, \quad c_k=\frac{\delta-\eta_{j}}{\eta_k-\eta_{j}}.
	\end{equation}
Therefore,
\begin{align}
&c_j \zeta_j+c_k \zeta_k=\frac{\eta_{k}-\delta}{\eta_k-\eta_{j}}\zeta_j+\frac{\delta-\eta_{j}}{\eta_k-\eta_{j}}\zeta_k\nonumber\\
&=\zeta_j+g_{kj}(\delta-\eta_j)= \zeta_k+g_{kj}(\delta-\eta_k),
\end{align}
where $g_{kj}=g_{kj}(\lambda)$ is defined in \eref{eq:gkj}.

If $j>k+1$, then $\eta_{j-1}< \delta$ or $\eta_{k+1}\geq  \delta$, 
so the value of $c_j \zeta_j+c_k \zeta_k$ does not increase if we replace
 $j$ with $j-1$ or $k$ with $k+1$ according to \lref{lem:zetaEtaSlope}. 
Therefore, the minimum in \eref{eq:FidelityLP2} can be attained when $j=k+1$ and $\eta_{k+1}<\delta \leq \eta_k$, in which case $k=k_*$ is the largest integer that satisfies the condition $\eta_k\geq \delta$. In addition we have $c_k=c_k(\delta,\lambda)$ and  $c_j=1-c_k(\delta,\lambda)$, so that
\begin{equation}
\zeta(N,\delta,\lambda)=c_j \zeta_j+c_k \zeta_k=\zeta(N,\delta,\lambda,k_*),
\end{equation}
which confirms \eref{eq:FidelityHomo}. 
\end{proof}

\begin{proof}[Proof of \thref{thm:NumTestHomo}]
By definition $N(\epsilon, \delta,\lambda)$ is the minimum  value  of the positive integer $N$ under the condition $F(N,\delta,\lambda)\geq F$ with $F=1-\epsilon$, that is, 
\begin{equation}\label{eq:zetaCondition}
\zeta(N,\delta,\lambda)\geq F\delta, 
\end{equation}
where $F\delta>0$.
According to \crref{cor:FidHomo} in the main text, \eref{eq:zetaCondition} is equivalent to 
\begin{align}\label{eq:zetaCondition1}
\max_{k\in \nni} \zeta(N,\delta, \lambda,k)\geq F\delta.
\end{align}
Note that the maximum in the left-hand side can be attained at a finite value of $k$. 

From  the definition of $\zeta(N,\delta, \lambda,k)$ in \eref{eq:zetadel} we can deduce that the inequality  $\zeta(N,\delta, \lambda,k)\geq F\delta$ is satisfied iff
\begin{align}
\! N\geq \tilde{N}(\epsilon,\delta,\lambda,k)=\frac{k\nu^2 \delta F +\lambda^{k+1}+\lambda\delta(k\nu-1)}{\lambda\nu\delta \epsilon}.
\end{align} 
So \eref{eq:zetaCondition} is satisfied iff 
\begin{align}
N&\geq\min_{k\in \nni}\tilde{N}(\epsilon,\delta,\lambda,k), 
\end{align}
which implies  \thref{thm:NumTestHomo} given  \lref{lem:NumTestProp}. 
\end{proof}

\begin{proof}[Proof of \eref{eq:NumTestHomoLS2}]
The equality in \eref{eq:NumTestHomoLS2} follows from \thref{thm:NumTestHomo} and \crref{cor:NumTestHomoUB}, note that 
\begin{equation}
\tilde{N}(\epsilon,\delta,\lambda,1)=\frac{\nu^2 \delta F +\lambda^{2}-\lambda^2\delta}{\lambda\nu\delta \epsilon}. 
\end{equation}
To prove the lower bound in \eref{eq:NumTestHomoLS2}, we first compute the derivative of  $\tilde{N}(\epsilon,\delta,\lambda,1)$ over $\lambda$, with the result
\begin{equation}
\frac{\partial\tilde{N}(\epsilon,\delta,\lambda,1)}{\partial \lambda}=\frac{(1-\delta)\lambda^2-\delta F\nu^2}{\lambda^2\nu^2\epsilon \delta}. 
\end{equation}
The minimum of $\tilde{N}(\epsilon,\delta,\lambda,1)$ over the interval $0<\lambda<1$ is attained when  $\lambda/(1-\lambda)=\sqrt{\delta F/(1-\delta)}$, that is,
\begin{equation}
\lambda=\lambda_*:=\frac{\sqrt{\delta F}}{\sqrt{1-\delta}+\sqrt{\delta F}}. 
\end{equation}
Therefore, 
\begin{align}
N(\epsilon,\delta,\lambda)\geq \tilde{N}(\epsilon,\delta,\lambda,1)\geq \tilde{N}(\epsilon,\delta,\lambda_*,1)
=\frac{2\sqrt{(1-\delta)F}}{\epsilon\sqrt{\delta}},
\end{align}
which confirms the lower bound in \eref{eq:NumTestHomoLS2}.
\end{proof}

\begin{proof}[Proof of \thref{thm:NumTestHomoBounds}]
Let $N=k_++\bigl\lceil\frac{k_+F}{\lambda\epsilon}\bigr\rceil$. According to \crref{cor:FidHomoBounds}, we have
\begin{align}
F(N,\delta,\lambda)&\geq 	\frac{(N-k_+)\lambda}{k_++(N-k_+)\lambda}=\frac{\bigl\lceil\frac{k_+F}{\lambda\epsilon}\bigr\rceil\lambda}{k_++\bigl\lceil\frac{k_+F}{\lambda\epsilon}\bigr\rceil\lambda}\nonumber \\
&\geq
 \frac{\frac{k_+F}{\epsilon}}{k_++\frac{k_+F}{\epsilon}}=F=1-\epsilon,
\end{align}
which implies that $N(\epsilon,\delta,\lambda)\leq N$ and confirms the upper bound in \eref{eq:NumberBoundAdvHomo}.

Next,  let $N=k_-+\bigl\lceil\frac{k_-F}{\lambda\epsilon}\bigr\rceil$. If $k_-=0$, then we have $N=0<N(\epsilon,\delta,\lambda)$. If $k_-\geq 1$, then $N-1\geq k_-\geq 1 $. By virtue of  \crref{cor:FidHomoBounds} we can deduce that
\begin{align}
&F(N-1,\delta,\lambda)\leq 	\frac{(N-1-k_-)\lambda}{k_-+(N-1-k_-)\lambda}\nonumber \\
&=\frac{\bigl(\bigl\lceil\frac{k_-F}{\lambda\epsilon}\bigr\rceil-1\bigr)\lambda}{k_-+\bigl(\bigl\lceil\frac{k_-F}{\lambda\epsilon}\bigr\rceil-1\bigr)\lambda}<
\frac{\frac{k_-F}{\epsilon}}{k_-+\frac{k_-F}{\epsilon}}=1-\epsilon,
\end{align}
which implies that $N(\epsilon,\delta,\lambda)\geq N$ and confirms the lower bound in \eref{eq:NumberBoundAdvHomo}.

 If $\log_\lambda\delta$ is an integer, then $k_+=k_-$, so the lower bound and  upper bound in \eref{eq:NumberBoundAdvHomo} coincide, which means both of them are saturated. Alternatively, this fact can be verified by virtue of \thref{thm:NumTestHomo}.

Finally, let us prove \eref{eq:NumberBoundAdvHomo2}.
 \Thref{thm:NumTestHomo} in the main text and \lref{lem:NumTestmBound} imply that
\begin{align}
&N(\epsilon, \delta,\lambda)=\lceil\min\{\tilde{N}_+(\epsilon,\delta,\lambda),\tilde{N}_-(\epsilon,\delta,\lambda)\}\rceil\nonumber\\
&\leq \lceil\tilde{N}_-(\epsilon,\delta,\lambda)\rceil\leq\biggl\lceil \frac{\log_\lambda\delta}{\lambda\epsilon}-\frac{\nu k_-}{\lambda}\biggr\rceil,
\end{align}
which confirms \eref{eq:NumberBoundAdvHomo2}. If $\log_\lambda\delta$ is an integer, then both inequalities are saturated, so the bound in \eref{eq:NumberBoundAdvHomo2} is saturated. 
\end{proof}

\section{\label{sec:SingleCopyProof}Proof of \thref{thm:SingleCopy}}
\begin{proof}
	If the strategy  $\Omega$ is homogeneous, then we have
	$\zeta(N,\delta,\Omega)= \zeta(N,\delta,\beta)$, and 
	\thref{thm:SingleCopy} follows from \pref{pro:SingleCopyHomo}. In general 
	\thref{thm:SingleCopy} can be proved based on \eref{eq:zetaetaLP} and the observation that $\eta_{\bfk}(\bm{\lambda})-\zeta_{\bfk}(\bm{\lambda})=1/2$ for all $\bfk\in \scrS_1$ with $k_1=1$ given the assumption $N=1$. Here $\scrS_1$ is defined in the paragraph before \eref{eq:prhofrho} in the main text.  Geometrically, this fact means that all points $(\eta_{\bfk}(\bm{\lambda}),\zeta_{\bfk}(\bm{\lambda}))$ for  $\bfk\in \scrS_1$ with $k_1=1$ lie on a line segment.

	To be specific, recall that  $\zeta(N,\delta,\Omega)\leq \zeta(N,\delta,\beta)$. When $\beta\geq 1/2$,
	\eref{eq:SingleCopyAdvGen1} holds because the opposite inequality $\zeta(N,\delta,\Omega)\geq \zeta(N,\delta,\beta)$ also holds. In view of \lref{lem:VScompare}, to verify this claim, it suffices to prove that
	\begin{equation}\label{eq:SingleCopyProof1}
	\zeta_\bfk(\bm{\lambda})\geq \zeta(N,\delta=\eta_\bfk(\bm{\lambda}),\beta)\quad \forall \bfk\in \scrS_1. 
	\end{equation}
	The assumption $\bfk\in \scrS_1$ means  $k_j\geq 0$ and $\sum_j k_j=2$. When $k_1=2$, we have $\zeta_\bfk(\bm{\lambda})=\eta_\bfk(\bm{\lambda})=1$, so  \eref{eq:SingleCopyProof1} holds. When $k_1=0$, we have $\zeta_\bfk(\bm{\lambda})=0$, while $\eta_\bfk(\bm{\lambda})\leq \beta$ according to \lref{lem:etaf0}, so that $\zeta(N,\delta=\eta_\bfk(\bm{\lambda}),\beta)=0$ (cf.~\thref{thm:FidelityHomo}) and \eref{eq:SingleCopyProof1} also holds. When $k_1=1$, according to \eref{eq:etazeta}, we have
	\begin{equation}\label{eq:etazetaN1}
	\eta_\bfk(\bm{\lambda})=\frac{1+\lambda_j}{2},\quad \zeta_\bfk(\bm{\lambda})=\frac{\lambda_j}{2}
	\end{equation}
	for some $2\leq j\leq D$.  If $(1+\lambda_j)/2\leq \beta$, then we have $\zeta(N,\delta=\eta_\bfk(\bm{\lambda}),\beta)=0$ according to \eref{eq:SingleCopyAdvHomo}, so \eref{eq:SingleCopyProof1}  holds. If  $(1+\lambda_j)/2\geq \beta$
	(note that  $\lambda_j\leq \beta$), then 
	\begin{align}
	&\zeta_\bfk(\bm{\lambda})- \zeta(N,\delta=\eta_\bfk(\bm{\lambda}),\beta)= \frac{\lambda_j}{2} -\frac{\beta(1+\lambda_j-2\beta)}{2(1-\beta)}\nonumber\\
	&=\frac{(2\beta-1)(\beta-\lambda_j)}{2(1-\beta)}\geq 0. 
	\end{align}
Therefore,  \eref{eq:SingleCopyProof1} holds for all $\bfk\in \scrS_1$, which implies that  $\zeta(N,\delta,\Omega)\geq \zeta(N,\delta,\beta)$. In conjunction with  the opposite inequality, we can deduce the desired equality
	$\zeta(N,\delta,\Omega)= \zeta(N,\delta,\beta)$, which confirms \eref{eq:SingleCopyAdvGen1}. 
	
	Next, consider the  case  $\beta< 1/2$.  If $\tau=\beta$, then  \eref{eq:SingleCopyAdvGen2} follows from \pref{pro:SingleCopyHomo}. If $\tau<\beta$,  let $\tilde{\Omega}$ be a verification operator with three distinct eigenvalues, $1, \beta, \tau$ (the eigenvalue 1 is nondegenerate); then we have $\zeta(N,\delta,\Omega)\leq \zeta(N,\delta,\tilde{\Omega})$. In addition, it is straightforward to verify \eref{eq:SingleCopyAdvGen2} if $\Omega$ is replaced by $\tilde{\Omega}$. To prove \eref{eq:SingleCopyAdvGen2}, it suffices to prove that $\zeta(N,\delta,\Omega)\geq \zeta(N,\delta,\tilde{\Omega})$. 
	Thanks to  \lref{lem:VScompare}, this condition can be simplified to
	\begin{equation}\label{eq:SingleCopyProof2}
	\zeta_\bfk(\bm{\lambda})\geq \zeta(N,\delta=\eta_\bfk(\bm{\lambda}),\tilde{\Omega})\quad \forall \bfk\in \scrS_1. 
	\end{equation}
	When $k_1=2$, we have $\zeta_\bfk(\bm{\lambda})=\eta_\bfk(\bm{\lambda})=1$, so  \eref{eq:SingleCopyProof2} holds. When $k_1=0$, we have $\zeta_\bfk(\bm{\lambda})=0$ and  $\eta_\bfk(\bm{\lambda})\leq \beta$ according to \eref{eq:etazeta}, so 
	\begin{equation}
	\zeta(N,\delta=\eta_\bfk(\bm{\lambda}),\tilde{\Omega})\leq\zeta(N,\delta=\eta_\bfk(\bm{\lambda}),\beta)=0,
	\end{equation}
	 and \eref{eq:SingleCopyProof2} also holds. When $k_1=1$, \eref{eq:etazetaN1} and the inequality $\tau\leq \lambda_j\leq \beta$ imply that 
\begin{equation}
\zeta(N,\delta=\eta_\bfk(\bm{\lambda}),\tilde{\Omega})=\frac{\lambda_j}{2}=\zeta_\bfk(\bm{\lambda});
\end{equation}
recall that \eref{eq:SingleCopyAdvGen2} holds if $\Omega$ is replaced by $\tilde{\Omega}$. This observation confirms \eref{eq:SingleCopyProof2} and implies the inequality $\zeta(N,\delta,\Omega)\geq \zeta(N,\delta,\tilde{\Omega})$. In conjunction with the opposite inequality, we conclude that 
	$\zeta(N,\delta,\Omega)= \zeta(N,\delta,\tilde{\Omega})$, which implies \eref{eq:SingleCopyAdvGen2}. 
\end{proof}

\section{\label{sec:NIIDproof}Proofs of \lref{lem:MinFidelityUB} and \thref{thm:NIID}}
\subsection{Main body of the proofs}
\begin{proof}[Proof of \lref{lem:MinFidelityUB}]
By the definition of $F(N,\delta,\Omega)$ in \eref{eq:MinFdelta},	to prove the inequality in \eref{eq:MinFub} in the lemma, it suffices to find a permutation-invariant quantum  state $\rho$ on $\caH^{\otimes (N+1)}$ such  that	$p_\rho=\delta$ and
	\begin{equation}\label{eq:MinFidelityUBproof}
	f_\rho=p_\rho-\frac{1}{N+1}
	\end{equation}	
for each $\delta$ in the interval $1/(N+1)\leq \delta\leq \delta^*$.	
	Since $p_\rho$ and $f_\rho$ are linear in $\rho$, it suffices to find such a state in the two  cases $\delta=1/(N+1)$ and  $\delta=\delta^*$, respectively. 	When $\delta=1/(N+1)$, we can choose the state $\rho=\rho_\bfk$ with $\bfk=(N,0,\ldots, 0,1)$, in which case $p_\rho=1/(N+1)$ and $f_\rho=0$ by \eref{eq:etazeta}, so \eref{eq:MinFidelityUBproof} holds as desired; note that $\Omega$ is singular by assumption, which means  $\tau=\lambda_D=0$.

In the case $\delta=\delta^*$, we can choose the state $\rho=\rho_{\bfk_1}$ with $\bfk_1:=(N,1,0,\ldots,0)$. Then \eref{eq:etazeta} [cf.~\eref{eq:etazetaHomo}] yields
\begin{equation}\label{eq:etazetak1}
\begin{aligned}
p_\rho=\eta_{\bfk_1}(\bm{\lambda})&=\frac{1+N\beta}{N+1}=\frac{1+N(1-\nu)}{N+1}=\delta^*,\\  f_\rho=\zeta_{\bfk_1}(\bm{\lambda})&=\frac{N\beta}{N+1}=\frac{N(1-\nu)}{N+1}. 
\end{aligned}
\end{equation}
Therefore,
\begin{equation}
p_\rho-f_\rho=\eta_{\bfk_1}(\bm{\lambda})-\zeta_{\bfk_1}(\bm{\lambda})=\frac{1}{N+1},
\end{equation}
and \eref{eq:MinFidelityUBproof} holds again. 
This observation completes the proof of \lref{lem:MinFidelityUB}.
\end{proof}

\begin{proof}[Proof of \thref{thm:NIID}]
To prove the inequality in  \eref{eq:MinFidelityBound} in the theorem,	
let  $\rho=\sum_{\bfk\in \scrS_N} c_{\bfk} \rho_{\bfk}$ as in \eref{eq:prhofrho}, where $c_{\bfk}$ form a probability distribution on $\scrS_N$.  If $p_\rho=1$, then  $c_{\bfk}=\delta_{\bfk,\bfk_0}$ with  $\bfk_0:=(N+1,0,\ldots, 0)$, in which case we have   $F_\rho=f_\rho=1$ and $F(N,\delta=1,\Omega)=1$, so  \eref{eq:MinFidelityBound} holds.
If $0<p_\rho<1$, then  $c_{\bfk_0}<1$ and
\begin{align}
&\frac{1-p_\rho}{1-f_\rho}=\frac{1-\sum_{\bfk\in \scrS_N} c_\bfk\eta_\bfk(\bm{\lambda})}{1-\sum_{\bfk\in \scrS_N} c_\bfk\zeta_\bfk(\bm{\lambda})}\nonumber\\
&=\frac{1-c_{\bfk_0}-\sum_{\bfk\in \scrS_N^*} c_\bfk\eta_\bfk(\bm{\lambda})}{1-c_{\bfk_0}-\sum_{\bfk\in\scrS_N^*} c_\bfk\zeta_\bfk(\bm{\lambda})}=\frac{1-\sum_{\bfk\in \scrS_N^*} c_\bfk'\eta_\bfk(\bm{\lambda})}{1-\sum_{\bfk\in\scrS_N^*} c_\bfk'\zeta_\bfk(\bm{\lambda})}\nonumber\\
&=\frac{\sum_{\bfk\in \scrS_N^*} c_\bfk'[1-\eta_\bfk(\bm{\lambda})]}{\sum_{\bfk\in\scrS_N^*} c_\bfk'[1-\zeta_\bfk(\bm{\lambda})]},
\end{align}
where  $\scrS_N^*:=\scrS_N\setminus \{\bfk_0\}$ is  the subset of $\scrS_N$  without the vector  $\bfk_0:=(N+1,0,\ldots, 0)$,
and $c_\bfk':=c_\bfk/(1-c_{\bfk_0})$ form a probability distribution on $\scrS_N^*$. By virtue of \lref{lem:slope} below, we can deduce that
\begin{equation}
\frac{1-p_\rho}{1-f_\rho}\geq \min_{\bfk\in \scrS_N^*}\frac{1-\eta_\bfk(\bm{\lambda})}{1-\zeta_\bfk(\bm{\lambda})}=
 \frac{N\nu}{N\nu+1},
\end{equation}
so that
\begin{equation}\label{eq:frhoBound}
f_\rho\geq p_\rho-\frac{1-p_\rho}{N \nu}
\end{equation}
and 
\begin{equation}\label{eq:FrhoBound} F_\rho=\frac{f_\rho}{p_\rho}\geq 1-\frac{1-p_\rho}{N\nu p_\rho}.
\end{equation}
Here \esref{eq:frhoBound} and \eqref{eq:FrhoBound} also hold when $p_\rho=1$.
By the definition of $F(N,\delta,\Omega)$ in   \eref{eq:MinFdelta}, we conclude that
\begin{equation}\label{seq:MinFidelityBound}
F(N,\delta,\Omega)\geq 1-\frac{1-\delta}{N\nu \delta}.
\end{equation}
Incidentally,  this bound is negative and thus trivial  when $\delta< 1/(N\nu+1)$; in particular, it is negative  when $\delta\leq \beta^N$ since $\beta^N< 1/(N\nu+1)$ according to \eref{eq:lambdaN}.

Now we show that the inequality in
\eref{eq:MinFidelityBound} [same as \eref{seq:MinFidelityBound}] is saturated when $ \delta \ge \delta^*$.
Since $\delta^*=\eta_{\bfk_1}(\bm{\lambda})$ with $\bfk_1=(N,1,0,\ldots,0)$,  it suffices to show that the inequality in \eref{eq:frhoBound} can be saturated when $p_\rho \ge \eta_{\bfk_1}(\bm{\lambda})$.
When  $c_{\bfk}=\delta_{\bfk,\bfk_0}$, that is, $\rho=\rho_{\bfk_0}=(|\Psi\>\<\Psi|)^{\otimes (N+1)}$, we have $p_\rho=1$ and $f_\rho=1$, so \eref{eq:frhoBound}  is saturated.
When  $c_{\bfk}=\delta_{\bfk,\bfk_1}$, that is,  $\rho=\rho_{\bfk_1}$, we have $ p_\rho = \eta_{\bfk_1}(\bm{\lambda})=\delta^* $ and  $f_\rho = \zeta_{\bfk_1}(\bm{\lambda})$ [cf.~\eref{eq:etazetak1}], so \eref{eq:frhoBound} is also saturated.
Since both $p_\rho$ and $f_\rho$ are linear in $\rho$, it follows that  the inequality in \eref{eq:frhoBound} can be saturated by a convex combination of   $\rho_{\bfk_0}$ and  $\rho_{\bfk_1}$ whenever  $p_\rho\geq  \eta_{\bfk_1}(\bm{\lambda})$.

Next, we prove  \eref{eq:MinFidelityBound2} in the case $\nu \ge 1/2 $, that is, $\beta\leq 1/2$. To this end, note that
\begin{align}
p_\rho-f_\rho&=\sum_{\bfk\in \scrS_N} c_\bfk \eta_\bfk(\bm{\lambda})
-
\sum_{\bfk\in \scrS_N} c_\bfk \zeta_\bfk(\bm{\lambda})\nonumber\\
&=\sum_{\bfk\in \scrS_N} c_\bfk[ \eta_\bfk(\bm{\lambda})
-\zeta_\bfk(\bm{\lambda})]
\le \frac{1}{N+1},
\end{align}
where the last inequality follows from \lref{lem:etazetadif} below. Therefore,
\begin{equation}
F_\rho \geq 1-\frac{1}{(N+1)p_\rho}
\end{equation}
whenever $p_\rho>0$,
which implies that
\begin{equation}
F(N,\delta,\Omega) \geq 1-\frac{1}{(N+1)\delta}
\end{equation}
and confirms
\eref{eq:MinFidelityBound2}.
If in addition $\Omega$ is singular and $\delta$ satisfies $1/(N+1)\leq \delta\leq \delta^*$, then this bound is saturated according to \lref{lem:MinFidelityUB}.
\end{proof}

\subsection{Auxiliary lemmas}
Here we assume that $\lambda_j$ are the eigenvalues of a verification operator $\Omega$ that are arranged in decreasing order $1=\lambda_1> \lambda_2\geq\cdots \geq \lambda_D\geq 0$. In addition, $\beta=\lambda_2$ and $\tau=\lambda_D$ are the second largest and the smallest  eigenvalues; meanwhile, $\nu=1-\beta$. 
\begin{lemma}\label{lem:etazetadif}
 $\eta_\bfk(\bm{\lambda})
	-\zeta_\bfk(\bm{\lambda})\leq 1/(N+1)$ for all $\bfk\in \scrS_N$ if $\beta\leq 1/2$.
\end{lemma}
\begin{proof}
	If $\bfk=\bfk_0$, then $\eta_\bfk(\bm{\lambda})
	=\zeta_\bfk(\bm{\lambda})=1$, so we have $\eta_\bfk(\bm{\lambda})
	-\zeta_\bfk(\bm{\lambda})=0\leq 1/(N+1)$.
	If $\bfk\neq \bfk_0$, then \eref{eq:etazeta} implies that
	\begin{align}
&\eta_\bfk(\bm{\lambda})
	-\zeta_\bfk(\bm{\lambda})
=\sum_{i\geq 2|k_i\geq1}\frac{k_i}{(N+1)}\lambda_i^{k_i-1}\prod_{j\neq i | k_j\geq 1} \lambda_j^{k_j}\nonumber\\
&\le\frac{N+1-k_1}{N+1}\beta^{N-k_1}\leq \frac{N+1-k_1}{N+1}\Bigl(\frac{1}{2}\Bigr)^{N-k_1}\le \frac{1}{N+1}.
	\end{align}
The first inequality follows  from the facts that $\lambda_j\leq \beta$ for $j\geq 2$ and	
	 that $0\leq N-k_1\leq N$;   the second inequality follows from the assumption  $\beta\leq 1/2$.
\end{proof}

Define
\begin{equation}\label{eq:xik}
\xi_\bfk(\bm{\lambda}):=\frac{1-\eta_\bfk(\bm{\lambda})}{1-\zeta_\bfk(\bm{\lambda})},\quad 	\bfk\in \scrS_N^*, 
\end{equation}
where  $\scrS_N^*=\scrS_N\setminus \{\bfk_0\}$ is  the subset of $\scrS_N$  without the vector  $\bfk_0=(N+1,0,\ldots, 0)$. 
\begin{lemma}\label{lem:slope}
	For each $\bfk\in \scrS_N^*$, we have	
\begin{equation}\label{eq:xikBounds}
\frac{N\nu}{N\nu+1}\leq \xi_\bfk(\bm{\lambda})\leq 1-\tau^N,
\end{equation}
where  $\nu=1-\beta$ with $\beta=\lambda_2$ and $\tau=\lambda_D$, assuming that $\lambda_1=1$ and $\lambda_j$ are arranged in decreasing order. 
\end{lemma}
The lower bound in \eref{eq:xikBounds} can be expressed as
\begin{equation}
\frac{N\nu}{N\nu+1}=\frac{1-\eta_{\bfk_1}(\bm{\lambda})}{1-\zeta_{\bfk_1}(\bm{\lambda})},
\end{equation}
where $\bfk_1:=(N,1,0,\ldots,0)$. According to \eref{eq:lambdaN}, we have
\begin{equation}
\frac{N\nu}{N\nu+1}<1-\beta^N\leq 1-\tau^N. 
\end{equation}
 \Lref{lem:slope} implies that the region $R_{N,\Omega}$ is contained in the triangle determined by the following three lines
\begin{equation}
\begin{aligned}
f&=0,\\
1-p&=(1-\tau^N)(1-f),\\
1-p&=\frac{N\nu}{N\nu+1}(1-f). 
\end{aligned}
\end{equation}
The three vertices of the triangle are  $(1,1)$,  $(\tau^N,0)$, and $(1/(N\nu+1),0)$; the first two vertices are the extremal points of $R_{N,\Omega}$.

\begin{proof}[Proof of \lref{lem:slope}]
The assumption $\bfk\in \scrS_N^*$ implies  that  $\sum_j k_j=N+1$ and $k_1\leq N$.
Thanks to \lref{lem:hkmono} below, we have
\begin{align}
\xi_\bfk(\bm{\lambda})&\geq \xi_\bfk(1,\beta,\ldots,\beta)=\xi_{(k_1,N-k_1+1)}(1,\beta)\nonumber\\
&=\frac{1-\eta_{N-k_1+1}(\beta)}{1-\zeta_{N-k_1+1}(\beta)}\geq 
\frac{N\nu}{N\nu+1},
\end{align}
where the second inequality follows from \lref{lem:11slope} in \aref{sec:HomoApp}. 
Note that the definition of $\xi_\bfk(\bm{\lambda})$ [as well as that of $\eta_\bfk(\bm{\lambda})$ and $\zeta_\bfk(\bm{\lambda})$] can be extended as long as $\bfk$ and $\bm{\lambda}$ have the same number of components.

By the same token, we have
\begin{align}
\xi_\bfk(\bm{\lambda})&\leq \xi_\bfk(1,\tau,\ldots,\tau)=\xi_{(k_1,N-k_1+1)}(1,\tau)\nonumber\\
&=\frac{1-\eta_{N-k_1+1}(\tau)}{1-\zeta_{N-k_1+1}(\tau)}\leq 
1-\tau^N,
\end{align}
where the two inequalities follow from \lref{lem:hkmono} and \lref{lem:11slope}, respectively. 
\end{proof}

Here it is instructive to take a look at the special scenario in which $\zeta_\bfk(\bm{\lambda})=0$ (cf. the proof of \lref{lem:etaf0} in \aref{sec:zetaEtaProofs}), which means  $k_1=0$, or $\lambda_i =0$ and $k_i\geq 1$ for some $2\leq i\leq D$.  In the first case, we have $\tau^N\leq \eta_\bfk(\bm{\lambda})\leq \beta^N$ by \eref{eq:etazeta}, so that
\begin{align}
&\xi_\bfk(\bm{\lambda})=1-\eta_\bfk(\bm{\lambda})\leq 1-\tau^N,\\
&\xi_\bfk(\bm{\lambda})\geq 1- \beta^N\geq  \frac{N\nu}{N\nu+1},
\end{align}	
where the last inequality follows from \eref{eq:lambdaN}.
In the second case, we have $\tau=0$ and
\begin{equation}
\eta_\bfk(\bm{\lambda})=\frac{k_i\lambda_i^{k_i-1}}{N+1}\prod_{j\neq i, k_j>0}\lambda_j^{k_j}\leq \frac{1}{N+1},
\end{equation}
which implies that
\begin{align}
&\xi_\bfk(\bm{\lambda})=1-\eta_\bfk(\bm{\lambda})\leq 1=1-\tau^N,\\
&\xi_\bfk(\bm{\lambda})\geq \frac{N}{N+1}\geq \frac{N\nu}{N\nu+1}.
\end{align}	 
These results are compatible with \lref{lem:slope} as expected.

\begin{lemma}\label{lem:hkmono}
Suppose  $\bfk=(k_1,k_2, \ldots,k_m)$ is a sequence of $m\geq 2$ nonnegative integers that satisfies   $k_1\leq N$ and $\sum_j k_j=N+1$, where $N$ is a positive integer.
Let  $\bm{u},\bm{v}$ be two $m$-component vectors that satisfy $0\leq \bm{u}\leq \bm{v}\leq 1$ and $u_1=v_1=1$. Then we have $\xi_\bfk(\bm{u})\geq \xi_\bfk(\bm{v})$.
\end{lemma}
The inequality $0\leq \bm{u}\leq \bm{v}\leq 1$ in the above lemma means  $0\leq u_j\leq v_j\leq 1$ for each $j=1,2,\ldots, m$.

\begin{proof}	
The assumption  $0\leq \bm{u}\leq \bm{v}\leq 1$ and \eref{eq:etazeta}, imply that $\zeta_\bfk(\bm{u})\leq \zeta_\bfk(\bm{v})\leq k_1/(N+1)<1$, so  
 $\xi_\bfk(\bm{u})$ is continuous in  $\bm{u}$ for $0\leq \bm{u}\leq 1$ by the definition in \eref{eq:xik}. Therefore, it suffices to prove the lemma when $0<\bm{u}\leq \bm{v}\leq 1$,  in which case $\eta_\bfk(\bm{u})$ and $\zeta_\bfk(\bm{u})$ can be expressed as follows,
\begin{align}
\eta_\bfk(\bm{u})=\theta \sum_j \frac{k_j}{u_j}, \quad \zeta_\bfk(\bm{u})=\theta k_1,
\end{align}
where $\theta:=\bigl(\prod_iu_i^{k_i}\bigr)/(N+1)$.

For $j\geq 2$, calculation shows that
\begin{equation}
\begin{aligned}
\frac{\partial \eta_\bfk(\bm{u})}{\partial u_j }&=\theta\biggl(\frac{k_j}{u_j}\sum_i\frac{k_i}{u_i}-\frac{k_j}{u_j^2}\biggr),\\
\frac{\partial \zeta_\bfk(\bm{u})}{\partial u_j }&=\theta\frac{k_1k_j}{u_j}.
\end{aligned}
\end{equation}
These derivatives have well-defined limits even when some components $u_i$ approach zero; this fact would be clearer if we insert the expression of $\theta$. In addition,
	\begin{align}\label{eq:xiderivative}
&\frac{\partial \xi_\bfk(\bm{u})}{\partial u_j }=-\frac{\theta k_j u_j\sum_{i>1}\frac{k_i}{u_i}-\theta k_j+\theta^2k_1k_j}{(1-\theta k_1)^2u_j^2}\nonumber\\
&=-\frac{\theta k_j \bigl[u_j\sum_{i>1,i\neq j}\frac{k_i}{u_i}+ (k_j-1)+\theta k_1\bigr]}{(1-\theta k_1)^2u_j^2}\leq 0,
\end{align}
note that $1-\theta k_1  \geq 1/(N+1)>0$.
The inequality in \eref{eq:xiderivative} is strict except when $k_j=0$, in which case $\xi_\bfk(\bm{u})$ is independent of $u_j$, and so are $\eta_\bfk(\bm{u})$ and $\zeta_\bfk(\bm{u})$ [cf.~\eref{eq:etazeta} in the main text]. 
Therefore,  $\xi_\bfk(\bm{u})$ is nonincreasing in $u_j$ for $j\geq 2$, which means $\xi_\bfk(\bm{u})\geq \xi_\bfk(\bm{v})$ whenever $0< \bm{u}\leq \bm{v}\leq 1$ and $u_1=v_1=1$. The condition $0< \bm{u}\leq \bm{v}\leq 1$ can be relaxed to $0\leq \bm{u}\leq \bm{v}\leq 1$ by continuity.
\end{proof}

\section{\label{sec:NumberTesOptProof}Proofs of \lref{lem:Fboundf} and \thref{thm:NumTestBounds}}

\subsection{Auxiliary lemmas}

Before proving \lref{lem:Fboundf} and \thref{thm:NumTestBounds}, wee need to introduce a few auxiliary notations and results. 

Denote by $\bar{\scrS_N}$ the convex hull of $\scrS_N$, then $\bar{\scrS_N}$  is composed of real vectors  $\bfk=(k_1,k_2,\ldots,k_D)$ that satisfy $\sum_{j=1}^D  k_j=N+1$ and  $k_j\geq 0$ for $j=1,2,\ldots D$. When $\Omega$ is positive definite, that is,  $\tau(\Omega)>0$, we can extend the definition of $\eta_\bfk(\bm{\lambda})$ and 
$\zeta_\bfk(\bm{\lambda})$ over $\bfk$ to $\bar{\scrS_N}$ [cf.~\eref{eq:etazeta}]. Since all eigenvalues $\lambda_j$ of $\Omega$  for $j=1,2,\ldots,D$  are positive,  we have $\eta_\bfk(\bm{\lambda})>0$ for all $\bfk\in \bar{\scrS_N}$.
The following analogs of $\zeta(N,\delta,\Omega)$ and $\eta(N,f,\Omega)$ [cf.~\eref{eq:zetaetaLP}] will play  key roles in proving \lref{lem:Fboundf} and \thref{thm:NumTestBounds}. Define
\begin{align}
\bar{\zeta}(N,\delta,\Omega)\!&:=\begin{cases}
\displaystyle{\min_{\bfk\in \bar{\scrS_N}}}
\bigl\{
\zeta_\bfk(\bm{\lambda})
\big| \eta_\bfk(\bm{\lambda})=\delta
\bigr\}, & \beta^N\leq \delta\leq 1,\\
0, & 0\leq \delta\leq \beta^N;
\end{cases}\\
\bar{\eta}(N,f,\Omega)\!&:=
\max_{\bfk\in \bar{\scrS_N}}
\bigl\{ \eta_\bfk(\bm{\lambda})
\big| \zeta_\bfk(\bm{\lambda})=f
\bigr\},\, 0\leq f\leq 1,
\end{align}
where $\beta$ is the second largest eigenvalue of $\Omega$. Incidentally,  
$\eta(N,f=0,\Omega)=\delta_\rmc=\beta^N$ since $\tau>0$; see \lref{lem:etaf0}. 
\begin{lemma}\label{lem:ZetaEtabound}
	Suppose $0\leq \delta,f\leq1$ and $\Omega$ is a positive-definite verification operator; then
	\begin{align}
	\zeta(N,\delta,\Omega)	\geq \bar{\zeta}(N,\delta,\Omega), \label{eq:zetaBound}\\
	\eta(N,f,\Omega)	\leq \bar{\eta}(N,f,\Omega). \label{eq:etaBound}
	\end{align} 
\end{lemma}

\begin{proof}
		When $\delta$ satisfies $0\leq \delta\leq \beta^N$,  by definition we have $\bar{\zeta}(N,\delta,\Omega)=0\leq \zeta(N,\delta,\Omega)$, so \eref{eq:zetaBound} holds.

When $\delta>\beta^N$, by \lref{lem:ExtremalPoints} in \aref{sec:zetaEtaProofs}, we can find vectors $\bfq_0, \bfq_1\in \scrS_N$ such that $\beta^N\leq \eta_0<\eta_1\leq1$, 
	$\eta_0< \delta\leq \eta_1$, $0\leq \zeta_0<\zeta_1\leq 1$, and   $0\leq F_0<F_1\leq 1$, where $\eta_j=\eta_{\bfq_j}(\bm{\lambda})$,  $\zeta_j=\zeta_{\bfq_j}(\bm{\lambda})$, and  $F_j=\zeta_j/\eta_j$ for $j=0,1$. In addition,  $\zeta(N,\delta,\Omega)=c_0 \zeta_0+c_1\zeta_1$, where $c_0$ and $c_1$ are nonnegative coefficients determined by the requirements $c_0+c_1=1$ and $c_0\eta_0+c_1\eta_1=\delta$, that is,
	\begin{align}
	c_0&=\frac{\eta_1-\delta}{\eta_1-\eta_0},\quad c_1=\frac{\delta-\eta_0}{\eta_1-\eta_0}.
	\end{align}
If  $\delta=\eta_1$, then $\zeta(N,\delta,\Omega)=\zeta_1$ and  \eref{eq:zetaBound} holds because  $\scrS_N\subset \bar{\scrS_N}$. So it remains to consider the scenario $\eta_0< \delta< \eta_1$, in which case we have $0<c_0,c_1<1$.
	Geometrically, the point $(\delta,\zeta(N,\delta,\Omega))$ lies on the line segment that connects the two end points $(\eta_0,\zeta_0)$ and $(\eta_1,\zeta_1)$, which has slope $(\zeta_1-\zeta_0)/(\eta_1-\eta_0)$.

	For $0\leq t\leq 1$, let 
	\begin{align}
	\bfk(t)&=\bfq_0(1-t)+\bfq_1t=\bfq_0+(\bfq_1-\bfq_0)t, \label{eq:kt}\\
	\eta(t)&=\eta_{\bfk(t)}(\bm{\lambda}),\quad 
	\zeta(t)=\zeta_{\bfk(t)}(\bm{\lambda}).\label{eq:etazetat}
	\end{align}
Note that $\bfk(t)
\in \bar{\scrS_N}$ for $0\leq t\leq 1$; in addition, $\eta(0)=\eta_0$ and $\zeta(0)=\zeta_0$, while $\eta(1)=\eta_1$ and $\zeta(1)=\zeta_1$. So  \eref{eq:etazetat} defines a parametric curve $(\eta(t),\zeta(t))$ that connects $(\eta_0,\zeta_0)$ and $(\eta_1,\zeta_1)$.	
The explicit expressions of $\eta(t)$ and $\zeta(t)$ can be derived by virtue of \eref{eq:etazeta}, with the result
\begin{equation}\label{eq:etazetat2}
\eta(t)=\theta(t)\sum_j \frac{k_j(t)}{\lambda_j},\quad \zeta(t)=\theta(t)k_1(t),
\end{equation}		
where
\begin{align}
\theta(t)&=\frac{1}{N+1}\prod_j \lambda_j^{k_j(t)}.
\end{align}	
Let 
\begin{align}
F(t)&=\frac{\zeta(t)}{\eta(t)}=\frac{k_1(t)}{\sum_j \frac{k_j(t)}{\lambda_j}};  \label{eq:Ft}
\end{align}	
then $F(0)=F_0$ and  $F(1)=F_1$.

Let $t_\delta$ be the smallest value of $t$ such that $\eta(t)=\delta$; then $\bar{\zeta}(N,\delta,\Omega)\leq \zeta(t_\delta)$. So \eref{eq:zetaBound} would follow if we can prove that $\zeta(t_\delta)\leq \zeta(N,\delta,\Omega)$.

	To achieve our goal, we shall  prove that the parametric curve $(\eta(t),\zeta(t))$ for $0\leq t\leq t_\delta$ lies below the line segment passing through the two points  $(\eta_0,\zeta_0)$ and $(\eta_1,\zeta_1)$. To this end, we need to analyze the convexity (or concavity) property of the curve, which depends on the second derivative
	\begin{equation}
	\frac{d^2\zeta(t)}{d\eta(t)^2}=\frac{\zeta''(t)\eta'(t)-\eta''(t)\zeta'(t)}{\eta'(t)^3}. 
	\end{equation}
	Here the derivatives with respect to $t$ can be computed explicitly by virtue of \eref{eq:etazetat2}, with the result
	\begin{widetext}
	\begin{align}
	\eta'(t)&=\frac{d \eta(t)}{d t}=\eta(t)\sum_j(q_{1j} -q_{0j})\ln \lambda_j +\theta(t)\sum_j\frac{q_{1j}-q_{0j}}{\lambda_j}=\theta(t)\Bigl[\frac{\eta(t)}{\theta(t)}\ln\frac{\theta_1}{\theta_0}+\Bigl(\frac{\eta_1}{\theta_1}-\frac{\eta_0}{\theta_0}\Bigr)\Bigr],\label{eq:etaPrime}\\
	\zeta'(t)&=\frac{d \zeta(t)}{d t}=\zeta(t)\sum_j(q_{1j} -q_{0j})\ln \lambda_j +\theta(t)(q_{11}-q_{01})=\theta(t)\Bigl[k_1(t)\ln\frac{\theta_1}{\theta_0}+(q_{11}-q_{01})\Bigr], \label{eq:zetaPrime}\\
	\eta''(t)&=\frac{d^2 \eta(t)}{d t^2}=\theta(t)\Bigl(\ln\frac{\theta_1}{\theta_0}\Bigr)\Bigl[\frac{\eta(t)}{\theta(t)}\ln\frac{\theta_1}{\theta_0}+2\Bigl(\frac{\eta_1}{\theta_1}-\frac{\eta_0}{\theta_0}\Bigr)\Bigr],\\
	\zeta''(t)&=\frac{d^2 \zeta(t)}{d t^2}=\theta(t)\Bigl(\ln\frac{\theta_1}{\theta_0}\Bigr)\Bigl[k_1(t)\ln\frac{\theta_1}{\theta_0}+2(q_{11}-q_{01})\Bigr],
	\end{align}		
where 	
\begin{equation}
 \theta_0=\theta(t=0)=\frac{1}{N+1}\prod_j \lambda_j^{q_{0j}},\quad  \theta_1=\theta(t=1)=\frac{1}{N+1}\prod_j \lambda_j^{q_{1j}}.
\end{equation}
Note that
\begin{equation}
\theta'(t)=\frac{d \theta(t)}{d t}=\theta(t)\sum_j(q_{1j} -q_{0j})\ln \lambda_j =\theta(t)\ln\frac{\theta_1}{\theta_0}. 
		\end{equation}
		Therefore,
		\begin{align}
		&\zeta''(t)\eta'(t)-\eta''(t)\zeta'(t)=\theta(t)^2\Bigl(\ln\frac{\theta_1}{\theta_0}\Bigr)^2\Bigl[(q_{11}-q_{01})\frac{\eta(t)}{\theta(t)} -\Bigl(\frac{\eta_1}{\theta_1}-\frac{\eta_0}{\theta_0}\Bigr) k_1(t) \Bigr]\nonumber\\
		&=\theta(t)^2\Bigl(\ln\frac{\theta_1}{\theta_0}\Bigr)^2\biggl\{(q_{11}-q_{01})	\Bigl	[\frac{\eta_0}{\theta_0}+\Bigl(\frac{\eta_1}{\theta_1}-\frac{\eta_0}{\theta_0}\Bigr)t\Bigr]
 -\Bigl(\frac{\eta_1}{\theta_1}-\frac{\eta_0}{\theta_0}\Bigr) [q_{01 } +(q_{11}-q_{01})t] \biggr\}\nonumber\\
	&=\theta(t)^2\Bigl(\ln\frac{\theta_1}{\theta_0}\Bigr)^2\Bigl(\frac{\eta_0q_{11}}{\theta_0}-\frac{\eta_1q_{01}}{\theta_1}\Bigr)=\theta(t)^2\Bigl(\ln\frac{\theta_1}{\theta_0}\Bigr)^2\frac{\eta_0\eta_1}{\theta_0\theta_1}\Bigl(\frac{\theta_1q_{11}}{\eta_1}-\frac{\theta_0q_{01}}{\theta_0}\Bigr)\nonumber\\
		&=\theta(t)^2\Bigl(\ln\frac{\theta_1}{\theta_0}\Bigr)^2\frac{\eta_0\eta_1}{\theta_0\theta_1}(F_1-F_0)\geq 0.\label{eq:etazeta2121}
		\end{align}
	\end{widetext}
Here the inequality is strict except when $\theta_1=\theta_0$, in which case $\theta(t)$ is independent of $t$, while both $\eta(t)$ and $\zeta(t)$ are linear in  $t$.
So the derivative $\frac{d^2\zeta(t)}{d\eta(t)^2}$  has the same  sign  as $\eta'(t)$ unless it is identically zero.

	Note that $\eta(t)/\theta(t)$ is a linear function of $t$. So $\eta'(t)/\theta(t)$ is linear and thus  monotonic in $t$  according to \eref{eq:etaPrime}; actually,  $\eta'(t)/\theta(t)$ is strictly monotonic in $t$ unless it is a positive constant. When $t=0$, we have
	\begin{align}
	\eta'(0)&=\eta_0\Bigl[\ln\frac{\theta_1}{\theta_0}+\Bigl(\frac{\eta_1 \theta_0}{\theta_1 \eta_0}-1\Bigr)\Bigr]\nonumber\\
	&> \eta_0\Bigl[\ln\frac{\theta_1}{\theta_0}+\Bigl(\frac{ \theta_0}{\theta_1 }-1\Bigr)\Bigr]\geq 0\label{eq:etaPrime0}
	\end{align}
given that $\eta_1>\eta_0>0$.
Since $\theta(t)>0$, 	
it follows that $\eta'(t)$ has at most one zero  point in the interval $0\leq t\leq 1$. If $\eta'(t)>0$ in this interval, then $\frac{d^2\zeta(t)}{d\eta(t)^2}\geq 0$ and $\zeta(t)$  is a convex function of $\eta(t)$ for $0\leq t\leq 1$,
so  the parametric curve $(\eta(t),\zeta(t))$  lies below the line segment that connects the two points  $(\eta_0,\zeta_0)$ and $(\eta_1,\zeta_1)$, which implies the inequality $\zeta(t_\delta)\leq \zeta(N,\delta,\Omega)$ and  \eref{eq:zetaBound}. Here   $t_\delta$ is the smallest value of $t$ such that $\eta(t)=\delta$.
	Otherwise, $\eta'(t)$ has a unique zero point $0<t_2\leq 1$. If $t_2=1$, then the same conclusion holds. If $t_2<1$, then $\eta'(t)>0$ for $0\leq t<t_2$ and $\eta'(t)<0$ for $t_2< t\leq 1$, which implies that $\eta(t_2)>\eta_1$.
	So
	there exists a unique real number $t_3$ that satisfies the conditions $0<t_3<t_2$ and $\eta(t_3)=\eta_1$. Note that $\zeta(t)$  is  convex in $\eta(t)$
	for $0\leq t\leq t_3$ and that $t_\delta<t_3$. To prove \eref{eq:zetaBound}, it suffices to prove the inequality $\zeta(t_3)\leq \zeta_1$, that is, $F(t_3)\leq F_1$, given that $\eta(t_3)=\eta_1$. 
	
	To proceed, we compute the 
	derivative of $F(t)$ over $t$, with the result
	\begin{equation}\label{eq:Ftderivative}
	\frac{d F(t)}{d t}=\frac{\theta(t)^2 }{\eta(t)^2}  \frac{\eta_0\eta_1}{\theta_0\theta_1}(F_1-F_0)>0.
	\end{equation}
	This derivative can be derived either from \eref{eq:Ft} or from \esref{eq:etaPrime} and \eqref{eq:zetaPrime} given that $F(t)=\zeta(t)/\eta(t)$. 
	So $F(t)$ increases monotonically with $t$ for $0\leq t\leq 1$, which implies that $F(t_3)\leq F(1)=F_1$ and that
	$\zeta(t_3)\leq \zeta(1)=\zeta_1$. Therefore, the parametric curve $(\eta(t),\zeta(t))$ for $0\leq t\leq t_3$ lies below the line segment that connects the two points  $(\eta_0,\zeta_0)$ and $(\eta_1,\zeta_1)$, which implies that  $\zeta(t_\delta)\leq \zeta(N,\delta,\Omega)$ and
confirms \eref{eq:zetaBound}.

\Eref{eq:etaBound} can be proved using a similar reasoning used for proving \eref{eq:zetaBound}. When $f=0$, we have
\begin{align}
	&\bar{\eta}(N,f,\Omega)=
	\max_{\bfk\in \bar{\scrS_N}}
	\bigl\{ \eta_\bfk(\bm{\lambda})
	\big| \zeta_\bfk(\bm{\lambda})=0\}\nonumber\\
	&\geq \max_{\bfk\in \scrS_N}
	\bigl\{ \eta_\bfk(\bm{\lambda})
	\big| \zeta_\bfk(\bm{\lambda})=0\}
	=\eta(N,f,\Omega),
	\end{align}
which confirms \eref{eq:etaBound};
here the inequality follows from the fact that $\scrS_N$ is contained in $\bar{\scrS_N}$. 	When $f>0$, we can choose   $\bfq_0, \bfq_1\in \scrS_N$ and define $\eta_0, \zeta_0, \eta_1, \zeta_1, \eta(t),\zeta(t)$ in a similar way to the proof of \eref{eq:zetaBound}, but with the requirement  $\eta_0< \delta\leq \eta_1$ replaced by  $\zeta_0< f\leq \zeta_1$. Since the case $ f=\zeta_1$ is trivial, we can assume $\zeta_0< f< \zeta_1$.  Then \esref{eq:kt}-\eqref{eq:Ftderivative} still apply.  
According to  \eref{eq:etazeta2121} and the following equation
	\begin{equation}
	\frac{d^2\eta(t)}{d\zeta(t)^2}=-\frac{\zeta''(t)\eta'(t)-\eta''(t)\zeta'(t)}{\zeta'(t)^3},
	\end{equation}
the derivative $\frac{d^2\eta(t)}{d\zeta(t)^2}$  has the opposite  sign  to $\zeta'(t)$ unless it is identically zero .

 When $t=0$, we have
	\begin{align}
\zeta'(0)&=\theta_0\Bigl[q_{01}\ln\frac{\theta_1}{\theta_0}+(q_{11}-q_{01})\Bigr]\nonumber\\
&\geq \theta_0\Bigl[q_{01}\Bigl(1-\frac{\theta_0}{\theta_1}\Bigr)+(q_{11}-q_{01})\Bigr]\nonumber\\
&=\theta_0\frac{q_{11}\theta_1-q_{01}\theta_0}{\theta_1}=\theta_0\frac{\zeta_1-\zeta_0}{\theta_1}>0.
\end{align}	
 In addition,  $\theta(t)>0$, and  $\zeta'(t)/\theta(t)$ is a linear and thus monotonic function of $t$  according to \eref{eq:zetaPrime}. Therefore,  $\zeta'(t)$ has at most one zero  point in the interval $0\leq t\leq 1$ as is the case for $\eta'(t)$. Now \eref{eq:etaBound} can be proved using a similar reasoning as presented after
	\eref{eq:etaPrime0}, though "convex" is replaced by "concave".
\end{proof}

\begin{lemma}\label{lem:ConstrainedMax}
	Suppose $1>x_1\geq x_2\geq \cdots, x_m>0$ and $c\leq0$. Then 
	\begin{equation}\label{eq:ConstrainedMax}
	\max_{a_1,a_2,\ldots, a_m \geq0 }\Biggl\{\sum_j \frac{a_j}{x_j}\Bigg| \sum_j  a_j \ln x_j=c\Biggr\}=\frac{c}{y\ln y},
	\end{equation}
	where $y=x_1$ if $x_1\ln x_1^{-1}\leq x_m\ln x_m^{-1}$ and $y=x_m$ otherwise. 
\end{lemma}

\begin{proof}
	The maximization in \eref{eq:ConstrainedMax} is a linear programming in which the feasible region is defined by the inequalities $a_1,a_2,\ldots, a_m \geq0$ and the equality $\sum_j  a_j \ln x_j=c$. If $c=0$, then $a_1=a_2\cdots=a_m =0$, so \eref{eq:ConstrainedMax} holds. 
	
	If $c<0$, then the maximum in \eref{eq:ConstrainedMax} can be attained at one of the extremal points of the feasible region, which have the form
	\begin{equation}
	a_j=\frac{c}{\ln x_j}, \quad a_i=0\quad \forall i\neq j,\quad j=1,2,\ldots, m.
	\end{equation}
	Therefore,
	\begin{align}
	&\max_{a_1,a_2,\ldots, a_m \geq0 }\Biggl\{\sum_j \frac{a_j}{x_j}\Bigg|\sum_j  a_j \ln x_j=c\Biggr\}
	=\max_j\frac{c}{x_j\ln x_j}\nonumber\\
	&=\max\biggl\{\frac{c}{x_1\ln x_1},\frac{c}{x_m\ln x_m}\biggr\}
	=\frac{c}{y\ln y}. 
	\end{align}
	Here the second equality follows from the assumption $1>x_1\geq x_2\geq \cdots x_m>0$ and the fact that the function $c/(x\ln x)$ is convex in $x$ for $0<x<1$, given that $c$ is negative.
\end{proof}

\subsection{Main body of the proofs}

Now we are ready to prove \lref{lem:Fboundf}. 
\begin{proof}[Proof of \lref{lem:Fboundf}] We shall  first prove \eref{eq:Fboundf}. 
	According to \lref{lem:ZetaEtabound},
	\begin{align}
&\caF(N,f,\Omega)=\frac{f}{\eta(N,f,\Omega)}\geq \frac{f}{\bar{\eta}(N,f,\Omega)}
	\nonumber\\
	&=\min_{\bfk\in \bar{\scrS_N}|\zeta_\bfk(\bm{\lambda})=f}\frac{\zeta_\bfk(\bm{\lambda})}{\eta_\bfk(\bm{\lambda})}=\min_{\bfk\in \bar{\scrS_N}|\zeta_\bfk(\bm{\lambda})=f}\frac{k_1}{\sum_j (k_j/\lambda_j)}\nonumber\\
	&=\min_{\bfk\in \bar{\scrS_N}|\zeta_\bfk(\bm{\lambda})=f}\frac{k_1}{k_1 +\sum_{j=2}^D(k_j/\lambda_j)}.
	\end{align}
	The condition $\zeta_\bfk(\bm{\lambda})=f$ entails the following inequality,
	\begin{equation}
	\!\!f=\zeta_\bfk(\bm{\lambda})=\frac{k_1}{N+1}\prod_j \lambda_j^{k_j}\leq \prod_{j=2}^D \lambda_j^{k_j}\leq \beta^{N+1-k_1},
	\end{equation}
	which implies that $N+1-k_1\leq \ln f/\ln \beta=\log_{\beta}f$, that is, $k_1\geq N+1-(\ln f/\ln \beta)$. In addition, the above equation implies that 
	$0\geq\sum_{j=2}^D k_j\ln \lambda_j \geq \ln f$, which in turn implies that 
	$\sum_{j=2}^D(k_j/\lambda_j)\leq \ln f/(\tilde{\beta}\ln \tilde{\beta})$ in view of \lref{lem:ConstrainedMax}. Therefore,
	\begin{align}
&\caF(N,f,\Omega)\geq \min_{\bfk\in \bar{\scrS_N}|\zeta_\bfk(\bm{\lambda})=f} \frac{k_1}{k_1+(\tilde{\beta}\ln \tilde{\beta})^{-1}\ln f} \nonumber \\
&\geq \frac{N+1-(\ln \beta)^{-1}\ln f}{N+1-(\ln \beta)^{-1}\ln f-h\ln f},\label{eq:FboundfProof}
\end{align}	
which confirms \eref{eq:Fboundf}.

Next, let us  prove \eref{eq:Fbounddel}. If $\delta\leq \beta^N$, then we have $\tau\delta\leq \beta^{N+1}$ and	$N+1-(\ln \beta)^{-1}\ln (\tau\delta)\leq0 $, so the bound in \eref{eq:Fbounddel} is either zero or negative and is thus trivial. If $\delta>\beta^N$, then  \lref{lem:ZetaEtabound} implies that 
\begin{align}
F(N,\delta,\Omega)&=\frac{\zeta(N,\delta,\Omega)}{\delta}\geq \frac{\bar{\zeta}(N,\delta,\Omega)}{\delta}
\nonumber\\
&=\min_{\bfk\in \bar{\scrS_N}|\eta_\bfk(\bm{\lambda})=\delta}\frac{k_1}{k_1 +\sum_{j=2}^D(k_j/\lambda_j)}.
\end{align}
The condition $\eta_\bfk(\bm{\lambda})=\delta$ entails the following inequality,
\begin{align}
\tau\delta&=\tau\eta_\bfk(\bm{\lambda})=\frac{\tau}{N+1}\biggl(\prod_j \lambda_j^{k_j}\biggr) \biggl(\sum_j \frac{k_j}{\lambda_j}\biggr)
\leq \prod_{j=2}^D \lambda_j^{k_j}\nonumber\\
&\leq \beta^{N+1-k_1}.
\end{align}
Now, \eref{eq:Fbounddel} can be proved using a similar reasoning that leads to \eref{eq:FboundfProof}, but with $f$ replaced by $\tau\delta$.
\end{proof}

\begin{proof}[Proof of \thref{thm:NumTestBounds}]
\Eref{eq:NumTestLB} follows from
\eref{eq:StrategyOrder} and 
 \thref{thm:NumTestHomoBounds}   in the main text. 
The lower bound in \eref{eq:NumTestLUB} follows from \eref{eq:NumTestLB}
given that  $\tilde{\beta}=\beta=\lambda_2$ or $\tilde{\beta}=\tau=\lambda_D$. 
	
To prove the  upper bounds  in \eref{eq:NumTestLUB}, let $f=F\delta$ with $F=1-\epsilon$ and 
\begin{equation}
N=\left\lceil \frac{hF\ln f^{-1}}{\epsilon} +\frac{\ln f}{\ln\beta}-1\right\rceil;
\end{equation}
then $N\geq 1$ since
\begin{align}
 \frac{hF\ln f^{-1}}{\epsilon} +\frac{\ln f}{\ln\beta}> \frac{F\ln F}{\epsilon \beta \ln \beta} +\frac{\ln F}{\ln\beta}> 1. 
\end{align}
Here the second inequality is equivalent to 
\begin{align}
F\ln F+\epsilon \beta\ln F-\epsilon \beta \ln \beta< 0.
\end{align}
To prove this inequality, note that for a given $0<F<1$, the left-hand side is maximized when $\beta=F/\rme$. So 
\begin{align}
F\ln F+\epsilon \beta\ln F-\epsilon \beta \ln \beta\leq \frac{F(1-F+\rme\ln F)}{\rme}<  0.
\end{align}

In addition, \lref{lem:Fboundf} implies that 
	\begin{align}
&\caF(N,f,\Omega)\geq \frac{N+1-(\ln \beta)^{-1}\ln f}{N+1-(\ln \beta)^{-1}\ln f-h\ln f}\nonumber\\
&\geq 
\frac{hF\epsilon^{-1}\ln f^{-1} }{ hF\epsilon^{-1}\ln f^{-1}  -h\ln f}=1-\epsilon.
\end{align}	
In conjunction with \lref{lem:MinNumTestDef} this equation implies that  $N(\epsilon,\delta,\Omega)\leq N$,  which confirms the first upper bound in \eref{eq:NumTestLUB}. Furthermore, we have $h>|1/\ln\beta|$ since $0<\beta<1$ and $|\tilde{\beta}\ln \tilde{\beta}|\leq |\beta\ln \beta|< |\ln\beta|$. Therefore,
\begin{align}
N&=\left\lceil \frac{h(1-\epsilon)\ln f^{-1}}{\epsilon} +\frac{\ln f}{\ln\beta}-1\right\rceil\nonumber\\
&<
\frac{h\ln f^{-1}}{\epsilon}-h\ln f^{-1} +\frac{\ln f}{\ln\beta}\nonumber\\
&<\frac{h\ln f^{-1}}{\epsilon}=
\frac{h\ln(F\delta)^{-1}}{\epsilon},
\end{align}
which confirms the second upper bound in \eref{eq:NumTestLUB}.

\Eref{eq:NumTestUB2}	 
can be proved using a similar reasoning used to prove the upper bounds in \eref{eq:NumTestLUB}, but with $F\delta$ replaced by $\tau\delta$ and $\caF(N,f,\Omega)$ replaced by $F(N,\delta,\Omega)$.	 	 
\end{proof}

\section{\label{sec:PTTproofs}Proofs of \lsref{lem:hpnmin}-\ref{lem:OverheadPG}}
\begin{proof}[Proof of \lref{lem:hpnmin}]By \esref{eq:pstartaumax} and \eqref{eq:hstartaumax} in the main text, it is clear that  $p_*(\nu,1-\nu)$ is nondecreasing in $\nu$, and 	$h_*(\nu,1-\nu)$ is nonincreasing in  $\nu$. 
If   $1-\rme^{-1}\leq \nu\leq 1$, then 
\begin{equation}\label{eq:nuhnuproof}
\nu h_*(\nu,1-\nu)=\rme\nu\geq  \rme(1-\rme^{-1})=\rme-1>1, 
\end{equation}
and $\nu h_*(\nu,1-\nu)$  is strictly increasing in $\nu$. On the other hand,
if  $0<\nu\leq 1-\rme^{-1}$, then 
\begin{align}
\nu	h_*(\nu,1-\nu)= \nu\bigl[(1-\nu)\ln (1-\nu)^{-1}\bigr]^{-1}, 
\end{align}
so  that 
\begin{equation}
\lim_{\nu\rightarrow 0} \nu h_*(\nu,1-\nu)=
\lim_{\nu\rightarrow 0} \nu \bigl[(1-\nu)\ln (1-\nu)^{-1}\bigr]^{-1}=1.
\end{equation}
By computing the derivative of $\nu	h_*(\nu,1-\nu)$ over $\nu$ [cf.~\eref{eq:hpnutauDer} below with $p=0$] it is straightforward to verify that  $\nu h_*(\nu,1-\nu)$  is strictly increasing in $\nu$ for  $0<\nu\leq 1-\rme^{-1}$.
In conjunction with \eref{eq:nuhnuproof}, we conclude that  $\nu	h_*(\nu,1-\nu)>1$ and it is strictly increasing in $\nu$ for 
 $0<\nu\leq 1$.

In addition,
\begin{equation}
\nu h(p,\nu,1-\nu)=\nu\bigl(\beta_p\ln \beta_p^{-1}\bigr)^{-1}, 
\end{equation}
where $\beta_p=1-\nu+p\nu$ satisfies $0<\beta_p<1$. The derivative of $\nu h(p,\nu,1-\nu)$ over $\nu$ reads
\begin{equation}\label{eq:hpnutauDer}
\frac{d}{d\nu}\biggl(\frac{\nu}{\beta_p\ln \beta_p^{-1}}\biggr)=
-\frac{(1-p)\nu+\ln(1-\nu+p\nu
	)}{[(1-\nu+p\nu)\ln(1-\nu+p\nu)]^2}>0,
\end{equation}
where the last inequality follows from the simple fact that  $\ln (1+x)< x$ when $x>-1$ and $x\neq 0$. Therefore, $\nu h(p,\nu,1-\nu)$ increases strictly  monotonically with $\nu$.
Incidentally,  the derivative in \eref{eq:hpnutauDer}    approaches $1/2$ in the limit $\nu\rightarrow 0$.  
\end{proof}

\begin{proof}[Proof of \lref{lem:hpnt}]
We shall prove the seven statements of \lref{lem:hpnt}	in the order 1, 6, 2;  3, 4; 7, 5. 

Recall that  $p_*(\nu,\tau)$ is the smallest value of $p\geq0$  that satisfies $\beta_p\geq 1/\rme$ and
$\tau_p \ln \tau_p^{-1}\geq \beta_p \ln \beta_p^{-1}$; see \eref{eq:pstardef}. Let $q=p_*(\nu,\tau)$; then $0\leq q< 1$. 
Suppose $0< \nu'< \nu$ and  let  $\beta'=1-\nu'$. Then we have $1>\beta'> \beta\geq 0$ and $1>\beta'_q>\beta_q\geq 1/\rme$, so that 
\begin{equation}
\beta'_q \ln {\beta'_q}^{-1}< \beta_q \ln \beta_q^{-1}\leq \tau_q \ln \tau_q^{-1},
\end{equation}
which implies that $p_*(\nu',\tau)\leq q= p_*(\nu,\tau)$, that is, $p_*(\nu,\tau)$ is nondecreasing in $\nu$. If $q>0$, actually we can deduce a stronger conclusion, namely, $p_*(\nu',\tau)< p_*(\nu,\tau)$.

In addition, the inequalities  $\tau_p\leq \beta_p\leq \beta'_p$ imply that
\begin{equation}
\beta_p \ln {\beta_p}^{-1}\geq \min\bigl\{\beta'_p \ln {\beta'_p}^{-1}, \tau_p \ln \tau_p^{-1}\bigr\}
\end{equation}
and that
\begin{align} &h(p,\nu',\tau)=\bigl[\min\bigl\{\beta'_p \ln {\beta'_p}^{-1}, \tau_p \ln \tau_p^{-1}\bigr\}\bigr]^{-1}\nonumber\\
&\geq\bigl[\min\bigl\{\beta_p \ln \beta_p^{-1}, \tau_p \ln \tau_p^{-1}\bigr\}\bigr]^{-1}=h(p,\nu,\tau). 
\end{align}	
So  $h(p,\nu,\tau)$ is nonincreasing in $\nu$.  When $p=p_*(\nu',\tau)$, the above equation implies that
\begin{align}
h_*(\nu',\tau)=h(p,\nu',\tau)\geq h(p,\nu,\tau)\geq h_*(\nu,\tau).
\end{align}
So $h_*(\nu,\tau)$ is also  nonincreasing in  $\nu$.

Next, suppose $ \tau\leq \tau'\leq \beta$. Then we have $\tau_q\leq \tau'_q\leq \beta_q$, $ \beta_q\geq 1/\rme$, and
\begin{equation} 
\tau'_q \ln {\tau'_q}^{-1}\geq \min\{\beta_q \ln \beta_q^{-1}, \tau_q \ln \tau_q^{-1}\}=\beta_q \ln \beta_q^{-1},
\end{equation}
which implies that $p_*(\nu,\tau')\leq q=  p_*(\nu,\tau)$. Therefore,
 $p_*(\nu,\tau)$ is nonincreasing in $\tau$, which confirms statement~1  of \lref{lem:hpnt} given that $p_*(\nu,\tau)$ is nondecreasing in $\nu$ as shown above.

In addition,  the inequalities  $\tau_p\leq \tau'_p\leq \beta_p$ imply that
\begin{equation}
\tau'_p \ln {\tau'_p}^{-1}\geq \min\bigl\{\beta_p \ln {\beta_p}^{-1}, \tau_p \ln \tau_p^{-1}\bigr\}
\end{equation}
and that
\begin{align} &h(p,\nu,\tau')=\bigl[\min\bigl\{\beta_p \ln {\beta_p}^{-1}, \tau'_p \ln {\tau'_p}^{-1}\bigr\}\bigr]^{-1}\nonumber\\
&\leq\bigl[\min\bigl\{\beta_p \ln \beta_p^{-1}, \tau_p \ln \tau_p^{-1}\bigr\}\bigr]^{-1}=h(p,\nu,\tau). \label{eq:hpvtOrderProof}
\end{align}	
So  $h(p,\nu,\tau)$ is nonincreasing in $\tau$, which confirms statement 6  of \lref{lem:hpnt} in view of the above conclusion.  When $p=p_*(\nu,\tau)$, \eref{eq:hpvtOrderProof} implies that
\begin{align}
h_*(\nu,\tau)=h(p,\nu,\tau)\geq h(p,\nu,\tau')\geq h_*(\nu,\tau').
\end{align}
So $h_*(\nu,\tau)$ is also  nonincreasing in  $\tau$,  which confirms statement 2  of \lref{lem:hpnt}. 

Next, consider statements 3 and 4 in \lref{lem:hpnt}. 
By  \lref{lem:hpnmin} and statement 2 in \lref{lem:hpnt} proved above, we have  $\nu h_*(\nu,\tau)\geq \nu h_*(\nu,1-\nu)>1$, which confirms statement 3 in \lref{lem:hpnt}. In addition, the following equations
\begin{align}
&\lim_{\nu\rightarrow 0} \nu h_*(\nu,\tau)\geq \lim_{\nu\rightarrow 0} \nu h_*(\nu,1-\nu)=1,\\
&\lim_{\nu\rightarrow 0} \nu h_*(\nu,\tau)\leq
\lim_{\nu\rightarrow 0} \nu h(\nu,\nu,\tau)
=1,
\end{align}
 imply the equality $\lim_{\nu\rightarrow 0} \nu h_*(\nu,\tau)=1$ and confirm statement~4 in \lref{lem:hpnt}.	

Finally, we can prove statements 7 and 5 in \lref{lem:hpnt}. By definition we have
\begin{equation}
\nu h(p,\nu,\tau)=\max\Bigl\{\nu\bigl(\beta_p\ln \beta_p^{-1}\bigr)^{-1}, \nu\bigl(\tau_p\ln \tau_p^{-1}\bigr)^{-1}\Bigr\},
\end{equation}
where $\beta_p=1-\nu+p\nu$. It is clear  that $\nu\bigl(\tau_p\ln \tau_p^{-1}\bigr)^{-1}$ increases strictly monotonically with $\nu$. The same conclusion holds for $\nu\bigl(\beta_p\ln \beta_p^{-1}\bigr)^{-1}$ according to the derivative in \eref{eq:hpnutauDer}. 
Therefore, $\nu h(p,\nu,\tau)$ increases strictly  monotonically with $\nu$, which confirms statement 7 in \lref{lem:hpnt}.

Suppose $0<\nu'<\nu\leq 1$. Then 
\begin{equation}
\nu' h_*(\nu',\tau)\leq \nu' h(q,\nu',\tau)< \nu h(q,\nu,\tau)=\nu h_*(\nu,\tau),
\end{equation}
where $q=p_*(\nu,\tau)$.
So $\nu h_*(\nu,\tau)$ increases strictly monotonically with $\nu$, which confirms statement 5 in \lref{lem:hpnt}.
\end{proof}

\begin{proof}[Proof of \lref{lem:OverheadPG}]
Recall that  $p_*(\nu)$ is the smallest value of $p>0$  that satisfies the conditions $\beta_p\geq 1/\rme$ and
$p \ln p= \beta_p \ln \beta_p$; see \eref{eq:pstarSingular}. Let $q=p_*(\nu)$; then $0< q\leq  1/\rme$. 
Suppose $0< \nu'< \nu$ and  let  $\beta'=1-\nu'$. Then  $1>\beta'> \beta\geq 0$ and $1>\beta'_q>\beta_q\geq 1/\rme$, so that 
\begin{equation}
\beta'_q \ln {\beta'_q}^{-1}< \beta_q \ln \beta_q^{-1}= q \ln q^{-1},
\end{equation}
which implies that $p_*(\nu')< q= p_*(\nu)$ and that $p_*(\nu)$ is strictly increasing in $\nu$. 
Consequently, 	$h_*(\nu)$ is strictly decreasing in  $\nu$ given that $h_*(\nu)=[p_*(\nu)\ln p_*(\nu)^{-1}]^{-1}$ and $0< p_*(\nu)\leq  1/\rme$. By contrast, $\nu h_*(\nu)$ is strictly  increasing in  $\nu$ according to \lref{lem:hpnt}.

Next, let us consider the monotonicity of $h(\rme^{-1}\nu,\nu)$ and $\nu h(\rme^{-1}\nu,\nu)$. 	By definition we have
\begin{align}
h(\rme^{-1}\nu,\nu)&\!=\!\Bigl[\min\Bigl\{\beta_{p_0} \ln \beta_{p_0}^{-1}, \frac{\nu}{\rme}\ln \frac{\rme}{\nu}\Bigr\}\Bigr]^{-1}, \\
\!\!\nu h(\rme^{-1}\nu,\nu)&\!=\!\max\Bigl\{ \nu\bigl(\beta_{p_0} \ln \beta_{p_0}^{-1}\bigr)^{-1}, \rme \Bigl(\ln \frac{\rme}{\nu}\Bigr)^{-1}\Bigr\}, 
\end{align}		
where $p_0=\nu/\rme$ and  $\beta_{p_0}=1-\nu+(\nu^2/\rme)$. As $\nu$ increases to 1, $\beta_{p_0}$ decreases strictly monotonically to $1/\rme$, while $\nu/\rme$ increases strictly monotonically to $1/\rme$. So $h(\rme^{-1}\nu,\nu)$ decreases strictly monotonically with $\nu$.

In addition, $\rme \bigl(\ln \frac{\rme}{\nu}\bigr)^{-1}$ is strictly increasing in  $\nu$ for $0<\nu\leq 1$. Meanwhile we have
	\begin{align}
	\frac{d [\nu(\beta_{p_0} \ln \beta_{p_0}^{-1})^{-1}]}{d\nu}=\frac{\rme \beta_{p_0}-(\rme-\nu^2)\ln(\rme \beta_{p_0})}{\rme\beta_{p_0}^2(\ln \beta_{p_0})^2 },
	\end{align}
where  the denominator is positive. The numerator is also positive according to the following equation. 
	\begin{align}
	&\rme \beta_{p_0}-(\rme-\nu^2)\ln (\rme \beta_{p_0})\nonumber\\
	&=\rme-\rme \nu+\nu^2-(\rme-\nu^2)\ln(\rme-\rme\nu+\nu^2)\nonumber\\
	&\geq \rme-\rme \nu+\nu^2-(\rme-\nu^2)(1-\nu)\nonumber\\
	&=(2-\nu)\nu^2> 0.
	\end{align}
Here the first inequality follows from the inequality below
\begin{equation}\label{eq:bpineq}
\ln (\rme-\rme\nu +\nu^2)\leq 1-\nu, 
\end{equation}
which can be proved  by inspecting the derivative. Therefore,  both $\nu (\beta_{p_0} \ln \beta_{p_0}^{-1})^{-1}$ and $\rme \bigl(\ln \frac{\rme}{\nu}\bigr)^{-1}$ are strictly increasing in $\nu$, which implies that $\nu h(\rme^{-1}\nu,\nu)$ is  strictly increasing in $\nu$.
	
Finally, we are ready to prove \eref{eq:OverheadPG2}. The first inequality there follows from the definition  of $h_*(\nu)$. To prove the rest inequalities, note that
\begin{align}
&\ln \beta_{p_0}^{-1}=-\ln (1-\nu +\rme^{-1}\nu^2)\geq\nu,\\
&\beta_{p_0} \ln \beta_{p_0}^{-1}\geq (1-\nu +\rme^{-1}\nu^2)\nu
\end{align}
by \eref{eq:bpineq}, where $p_0=\nu/\rme$. 
In addition, it is straightforward to verify the following inequality,
\begin{align}
p_0 \ln (p_0^{-1})=\frac{\nu}{\rme}\ln \frac{\rme}{\nu}\geq  (1-\nu +\rme^{-1}\nu^2)\nu.
\end{align}
Therefore,
\begin{equation}
\nu h(\rme^{-1}\nu,\nu)\leq (1-\nu +\rme^{-1}\nu^2)^{-1}\leq 1+(\rme-1)\nu\leq \rme,
\end{equation}
which confirms \eref{eq:OverheadPG2}  in \lref{lem:OverheadPG}. 
Here the second inequality follows from the inequality below
\begin{align}
& (1-\nu +\rme^{-1}\nu^2)[ 1+(\rme-1)\nu]\nonumber\\
&=1+\rme^{-1}\nu(1-\nu)(\rme^2-2\rme+\nu-\rme\nu)
\geq 1, 
\end{align}
given that $0< \nu\leq 1$. 	
\end{proof}

\section{\label{sec:MiniMeas}Proof of \pref{pro:MiniMeas}}
\begin{proof}
	First consider the bipartite case,	let $|\Psi\>$ be any bipartite entangled state shared between Alice and Bob.
	Suppose on the contrary that $|\Psi\>$ can be verified by a strategy $\Omega$ for which Alice  performs only one projective measurement. Without loss of generality, we may assume that this is a complete projective measurement associated with an orthonormal basis, say  $\{|\varphi_1\>, |\varphi_2\>, \ldots, |\varphi_d\>\}$,
	where $d$ is the dimension of the Hilbert space of Alice. Let $P_k=|\varphi_k\>\<\varphi_k|$ be the corresponding rank-1 projectors. Then any test operator necessarily has the form $E=\sum_{k=1}^d P_k\otimes Q_k$, where $Q_k$ are positive operators  on  the Hilbert space of Bob that satisfy $0\leq Q_k\leq 1$. To ensure that the target state can always pass the test, $E$ must satisfy the condition $\<\Psi|E|\Psi\>=1$. 
	
	Let  $|\tilde{\psi}_k\>=
	\<\varphi_k|\Psi\>$ be the unnormalized reduced state of Bob when Alice obtains outcome $k$ 
	and $p_k=\<\tilde{\psi}_k|\tilde{\psi}_k\>$  the corresponding probability. Let  $|\psi_k\>=|\tilde{\psi}_k\>/\sqrt{p_k}$ when   $p_k>0$. Then
	\begin{align}
	\<\Psi| E|\Psi\>=\sum_k \<\tilde{\psi}_k|Q_k|\tilde{\psi}_k\>\leq\sum_k \<\tilde{\psi}_k|\tilde{\psi}_k\>= \sum_k p_k=1.
	\end{align}
	By assumption, this inequality is saturated, which implies that $\<\psi_k|Q_k|\psi_k\>=1$ whenever $p_k>0$, in which case $|\psi_k\>$ is an eigenstate of $Q_k$ with eigenvalue 1. 
	So  all kets $|\varphi_k\>\otimes |\psi_k\>$ with $p_k>0$  belong to the pass eigenspace (corresponding to the eigenvalue 1) of each test operator $E$
	and thus the pass eigenspace of $\Omega$.
	Note that the number of outcomes with $p_k>0$ is at least equal to the Schmidt rank of $|\Psi\>$.  So
	the dimension of the pass eigenspace of $\Omega$ 
	is not smaller than the Schmidt rank of $|\Psi\>$; in particular, it is not smaller than 2 given that $|\Psi\>$ is entangled. Therefore, $|\Psi\>$ cannot be verified if Alice performs only one projective measurement; the same conclusion  holds if Bob performs only one projective measurement. 
	
In general, the proposition follows from the fact that a multipartite state can  also be considered as a bipartite state between one party and the other parties.
\end{proof}	

%%%%%%%%%%%%%%%%%%%%%%%%%%%%%%%%%%%%%%%%%%%%%%%%%%%%%%%%%%%%%%%%%%%%%%%%%%%%%%%%%%%%%%%%

%@CONTROL{REVTEX41Control}
%@CONTROL{apsrev41Control,author="48",editor="1",pages="1",title="0",year="0"}
%

\nocite{apsrev41Control}
\bibliographystyle{apsrev4-1}
\bibliography{all_references}

\begin{thebibliography}{63}%
\makeatletter
\providecommand \@ifxundefined [1]{%
 \@ifx{#1\undefined}
}%
\providecommand \@ifnum [1]{%
 \ifnum #1\expandafter \@firstoftwo
 \else \expandafter \@secondoftwo
 \fi
}%
\providecommand \@ifx [1]{%
 \ifx #1\expandafter \@firstoftwo
 \else \expandafter \@secondoftwo
 \fi
}%
\providecommand \natexlab [1]{#1}%
\providecommand \enquote  [1]{``#1''}%
\providecommand \bibnamefont  [1]{#1}%
\providecommand \bibfnamefont [1]{#1}%
\providecommand \citenamefont [1]{#1}%
\providecommand \href@noop [0]{\@secondoftwo}%
\providecommand \href [0]{\begingroup \@sanitize@url \@href}%
\providecommand \@href[1]{\@@startlink{#1}\@@href}%
\providecommand \@@href[1]{\endgroup#1\@@endlink}%
\providecommand \@sanitize@url [0]{\catcode `\\12\catcode `\$12\catcode
  `\&12\catcode `\#12\catcode `\^12\catcode `\_12\catcode `\%12\relax}%
\providecommand \@@startlink[1]{}%
\providecommand \@@endlink[0]{}%
\providecommand \url  [0]{\begingroup\@sanitize@url \@url }%
\providecommand \@url [1]{\endgroup\@href {#1}{\urlprefix }}%
\providecommand \urlprefix  [0]{URL }%
\providecommand \Eprint [0]{\href }%
\providecommand \doibase [0]{http://dx.doi.org/}%
\providecommand \selectlanguage [0]{\@gobble}%
\providecommand \bibinfo  [0]{\@secondoftwo}%
\providecommand \bibfield  [0]{\@secondoftwo}%
\providecommand \translation [1]{[#1]}%
\providecommand \BibitemOpen [0]{}%
\providecommand \bibitemStop [0]{}%
\providecommand \bibitemNoStop [0]{.\EOS\space}%
\providecommand \EOS [0]{\spacefactor3000\relax}%
\providecommand \BibitemShut  [1]{\csname bibitem#1\endcsname}%
\let\auto@bib@innerbib\@empty
%</preamble>
\bibitem [{\citenamefont {Horodecki}\ \emph {et~al.}(2009)\citenamefont
  {Horodecki}, \citenamefont {Horodecki}, \citenamefont {Horodecki},\ and\
  \citenamefont {Horodecki}}]{HoroHHH09}%
  \BibitemOpen
  \bibfield  {author} {\bibinfo {author} {\bibfnamefont {R.}~\bibnamefont
  {Horodecki}}, \bibinfo {author} {\bibfnamefont {P.}~\bibnamefont
  {Horodecki}}, \bibinfo {author} {\bibfnamefont {M.}~\bibnamefont
  {Horodecki}}, \ and\ \bibinfo {author} {\bibfnamefont {K.}~\bibnamefont
  {Horodecki}},\ }\bibfield  {title} {\enquote {\bibinfo {title} {Quantum
  entanglement},}\ }\href@noop {} {\bibfield  {journal} {\bibinfo  {journal}
  {Rev. Mod. Phys.}\ }\textbf {\bibinfo {volume} {81}},\ \bibinfo {pages} {865}
  (\bibinfo {year} {2009})}\BibitemShut {NoStop}%
\bibitem [{\citenamefont {G\"uhne}\ and\ \citenamefont
  {T\'oth}(2009)}]{GuhnT09}%
  \BibitemOpen
  \bibfield  {author} {\bibinfo {author} {\bibfnamefont {O.}~\bibnamefont
  {G\"uhne}}\ and\ \bibinfo {author} {\bibfnamefont {G.}~\bibnamefont
  {T\'oth}},\ }\bibfield  {title} {\enquote {\bibinfo {title} {Entanglement
  detection},}\ }\href@noop {} {\bibfield  {journal} {\bibinfo  {journal}
  {Phys. Rep.}\ }\textbf {\bibinfo {volume} {474}},\ \bibinfo {pages} {1}
  (\bibinfo {year} {2009})}\BibitemShut {NoStop}%
\bibitem [{\citenamefont {Hein}\ \emph {et~al.}(2004)\citenamefont {Hein},
  \citenamefont {Eisert},\ and\ \citenamefont {Briegel}}]{HeinEB04}%
  \BibitemOpen
  \bibfield  {author} {\bibinfo {author} {\bibfnamefont {M.}~\bibnamefont
  {Hein}}, \bibinfo {author} {\bibfnamefont {J.}~\bibnamefont {Eisert}}, \ and\
  \bibinfo {author} {\bibfnamefont {H.~J.}\ \bibnamefont {Briegel}},\
  }\bibfield  {title} {\enquote {\bibinfo {title} {Multiparty entanglement in
  graph states},}\ }\href@noop {} {\bibfield  {journal} {\bibinfo  {journal}
  {Phys. Rev. A}\ }\textbf {\bibinfo {volume} {69}},\ \bibinfo {pages} {062311}
  (\bibinfo {year} {2004})}\BibitemShut {NoStop}%
\bibitem [{\citenamefont {Kruszynska}\ and\ \citenamefont
  {Kraus}(2009)}]{KrusK09}%
  \BibitemOpen
  \bibfield  {author} {\bibinfo {author} {\bibfnamefont {C.}~\bibnamefont
  {Kruszynska}}\ and\ \bibinfo {author} {\bibfnamefont {B.}~\bibnamefont
  {Kraus}},\ }\bibfield  {title} {\enquote {\bibinfo {title} {Local
  entanglability and multipartite entanglement},}\ }\href@noop {} {\bibfield
  {journal} {\bibinfo  {journal} {Phys. Rev. A}\ }\textbf {\bibinfo {volume}
  {79}},\ \bibinfo {pages} {052304} (\bibinfo {year} {2009})}\BibitemShut
  {NoStop}%
\bibitem [{\citenamefont {Qu}\ \emph {et~al.}(2013)\citenamefont {Qu},
  \citenamefont {Wang}, \citenamefont {Li},\ and\ \citenamefont
  {Bao}}]{QuWLB13}%
  \BibitemOpen
  \bibfield  {author} {\bibinfo {author} {\bibfnamefont {R.}~\bibnamefont
  {Qu}}, \bibinfo {author} {\bibfnamefont {J.}~\bibnamefont {Wang}}, \bibinfo
  {author} {\bibfnamefont {Z.-s.}\ \bibnamefont {Li}}, \ and\ \bibinfo {author}
  {\bibfnamefont {Y.-r.}\ \bibnamefont {Bao}},\ }\bibfield  {title} {\enquote
  {\bibinfo {title} {Encoding hypergraphs into quantum states},}\ }\href@noop
  {} {\bibfield  {journal} {\bibinfo  {journal} {Phys. Rev. A}\ }\textbf
  {\bibinfo {volume} {87}},\ \bibinfo {pages} {022311} (\bibinfo {year}
  {2013})}\BibitemShut {NoStop}%
\bibitem [{\citenamefont {Rossi}\ \emph {et~al.}(2013)\citenamefont {Rossi},
  \citenamefont {Huber}, \citenamefont {Bru\ss},\ and\ \citenamefont
  {Macchiavello}}]{RossHBM13}%
  \BibitemOpen
  \bibfield  {author} {\bibinfo {author} {\bibfnamefont {M.}~\bibnamefont
  {Rossi}}, \bibinfo {author} {\bibfnamefont {M.}~\bibnamefont {Huber}},
  \bibinfo {author} {\bibfnamefont {D.}~\bibnamefont {Bru\ss}}, \ and\ \bibinfo
  {author} {\bibfnamefont {C.}~\bibnamefont {Macchiavello}},\ }\bibfield
  {title} {\enquote {\bibinfo {title} {Quantum hypergraph states},}\
  }\href@noop {} {\bibfield  {journal} {\bibinfo  {journal} {New J. Phys.}\
  }\textbf {\bibinfo {volume} {15}},\ \bibinfo {pages} {113022} (\bibinfo
  {year} {2013})}\BibitemShut {NoStop}%
\bibitem [{\citenamefont {Steinhoff}\ \emph {et~al.}(2017)\citenamefont
  {Steinhoff}, \citenamefont {Ritz}, \citenamefont {Miklin},\ and\
  \citenamefont {G\"uhne}}]{SteiRMG17}%
  \BibitemOpen
  \bibfield  {author} {\bibinfo {author} {\bibfnamefont {F.~E.~S.}\
  \bibnamefont {Steinhoff}}, \bibinfo {author} {\bibfnamefont {C.}~\bibnamefont
  {Ritz}}, \bibinfo {author} {\bibfnamefont {N.~I.}\ \bibnamefont {Miklin}}, \
  and\ \bibinfo {author} {\bibfnamefont {O.}~\bibnamefont {G\"uhne}},\
  }\bibfield  {title} {\enquote {\bibinfo {title} {Qudit hypergraph states},}\
  }\href@noop {} {\bibfield  {journal} {\bibinfo  {journal} {Phys. Rev. A}\
  }\textbf {\bibinfo {volume} {95}},\ \bibinfo {pages} {052340} (\bibinfo
  {year} {2017})}\BibitemShut {NoStop}%
\bibitem [{\citenamefont {Xiong}\ \emph {et~al.}(2018)\citenamefont {Xiong},
  \citenamefont {Zhen}, \citenamefont {Cao}, \citenamefont {Chen},\ and\
  \citenamefont {Chen}}]{XionZCC18}%
  \BibitemOpen
  \bibfield  {author} {\bibinfo {author} {\bibfnamefont {F.-L.}\ \bibnamefont
  {Xiong}}, \bibinfo {author} {\bibfnamefont {Y.-Z.}\ \bibnamefont {Zhen}},
  \bibinfo {author} {\bibfnamefont {W.-F.}\ \bibnamefont {Cao}}, \bibinfo
  {author} {\bibfnamefont {K.}~\bibnamefont {Chen}}, \ and\ \bibinfo {author}
  {\bibfnamefont {Z.-B.}\ \bibnamefont {Chen}},\ }\bibfield  {title} {\enquote
  {\bibinfo {title} {Qudit hypergraph states and their properties},}\
  }\href@noop {} {\bibfield  {journal} {\bibinfo  {journal} {Phys. Rev. A}\
  }\textbf {\bibinfo {volume} {97}},\ \bibinfo {pages} {012323} (\bibinfo
  {year} {2018})}\BibitemShut {NoStop}%
\bibitem [{\citenamefont {Raussendorf}\ and\ \citenamefont
  {Briegel}(2001)}]{RausB01}%
  \BibitemOpen
  \bibfield  {author} {\bibinfo {author} {\bibfnamefont {R.}~\bibnamefont
  {Raussendorf}}\ and\ \bibinfo {author} {\bibfnamefont {H.~J.}\ \bibnamefont
  {Briegel}},\ }\bibfield  {title} {\enquote {\bibinfo {title} {A one-way
  quantum computer},}\ }\href@noop {} {\bibfield  {journal} {\bibinfo
  {journal} {Phys. Rev. Lett.}\ }\textbf {\bibinfo {volume} {86}},\ \bibinfo
  {pages} {5188--5191} (\bibinfo {year} {2001})}\BibitemShut {NoStop}%
\bibitem [{\citenamefont {Raussendorf}\ \emph {et~al.}(2003)\citenamefont
  {Raussendorf}, \citenamefont {Browne},\ and\ \citenamefont
  {Briegel}}]{RausBB03}%
  \BibitemOpen
  \bibfield  {author} {\bibinfo {author} {\bibfnamefont {R.}~\bibnamefont
  {Raussendorf}}, \bibinfo {author} {\bibfnamefont {D.~E.}\ \bibnamefont
  {Browne}}, \ and\ \bibinfo {author} {\bibfnamefont {H.~J.}\ \bibnamefont
  {Briegel}},\ }\bibfield  {title} {\enquote {\bibinfo {title}
  {Measurement-based quantum computation on cluster states},}\ }\href@noop {}
  {\bibfield  {journal} {\bibinfo  {journal} {Phys. Rev. A}\ }\textbf {\bibinfo
  {volume} {68}},\ \bibinfo {pages} {022312} (\bibinfo {year}
  {2003})}\BibitemShut {NoStop}%
\bibitem [{\citenamefont {Broadbent}\ \emph {et~al.}(2009)\citenamefont
  {Broadbent}, \citenamefont {Fitzsimons},\ and\ \citenamefont
  {Kashefi}}]{BroaFK09}%
  \BibitemOpen
  \bibfield  {author} {\bibinfo {author} {\bibfnamefont {A.}~\bibnamefont
  {Broadbent}}, \bibinfo {author} {\bibfnamefont {J.}~\bibnamefont
  {Fitzsimons}}, \ and\ \bibinfo {author} {\bibfnamefont {E.}~\bibnamefont
  {Kashefi}},\ }\bibfield  {title} {\enquote {\bibinfo {title} {Universal blind
  quantum computation},}\ }in\ \href@noop {} {\emph {\bibinfo {booktitle}
  {Proceedings of the 50th Annual IEEE Symposium on Foundations of Computer
  Science}}}\ (\bibinfo  {publisher} {IEEE Computer Society},\ \bibinfo
  {address} {Washington, DC, USA},\ \bibinfo {year} {2009})\ pp.\ \bibinfo
  {pages} {517--526}\BibitemShut {NoStop}%
\bibitem [{\citenamefont {Morimae}\ and\ \citenamefont
  {Fujii}(2013)}]{MoriF13}%
  \BibitemOpen
  \bibfield  {author} {\bibinfo {author} {\bibfnamefont {T.}~\bibnamefont
  {Morimae}}\ and\ \bibinfo {author} {\bibfnamefont {K.}~\bibnamefont
  {Fujii}},\ }\bibfield  {title} {\enquote {\bibinfo {title} {Blind quantum
  computation protocol in which {A}lice only makes measurements},}\ }\href@noop
  {} {\bibfield  {journal} {\bibinfo  {journal} {Phys. Rev. A}\ }\textbf
  {\bibinfo {volume} {87}},\ \bibinfo {pages} {050301(R)} (\bibinfo {year}
  {2013})}\BibitemShut {NoStop}%
\bibitem [{\citenamefont {Hayashi}\ and\ \citenamefont
  {Morimae}(2015)}]{HayaM15}%
  \BibitemOpen
  \bibfield  {author} {\bibinfo {author} {\bibfnamefont {M.}~\bibnamefont
  {Hayashi}}\ and\ \bibinfo {author} {\bibfnamefont {T.}~\bibnamefont
  {Morimae}},\ }\bibfield  {title} {\enquote {\bibinfo {title} {Verifiable
  measurement-only blind quantum computing with stabilizer testing},}\
  }\href@noop {} {\bibfield  {journal} {\bibinfo  {journal} {Phys. Rev. Lett.}\
  }\textbf {\bibinfo {volume} {115}},\ \bibinfo {pages} {220502} (\bibinfo
  {year} {2015})}\BibitemShut {NoStop}%
\bibitem [{\citenamefont {Fujii}\ and\ \citenamefont
  {Hayashi}(2017)}]{FujiH17}%
  \BibitemOpen
  \bibfield  {author} {\bibinfo {author} {\bibfnamefont {K.}~\bibnamefont
  {Fujii}}\ and\ \bibinfo {author} {\bibfnamefont {M.}~\bibnamefont
  {Hayashi}},\ }\bibfield  {title} {\enquote {\bibinfo {title} {Verifiable
  fault tolerance in measurement-based quantum computation},}\ }\href@noop {}
  {\bibfield  {journal} {\bibinfo  {journal} {Phys. Rev. A}\ }\textbf {\bibinfo
  {volume} {96}},\ \bibinfo {pages} {030301(R)} (\bibinfo {year}
  {2017})}\BibitemShut {NoStop}%
\bibitem [{\citenamefont {Hayashi}\ and\ \citenamefont
  {Hajdu\ifmmode~\check{s}\else \v{s}\fi{}ek}(2018)}]{HayaH18}%
  \BibitemOpen
  \bibfield  {author} {\bibinfo {author} {\bibfnamefont {M.}~\bibnamefont
  {Hayashi}}\ and\ \bibinfo {author} {\bibfnamefont {M.}~\bibnamefont
  {Hajdu\ifmmode~\check{s}\else \v{s}\fi{}ek}},\ }\bibfield  {title} {\enquote
  {\bibinfo {title} {Self-guaranteed measurement-based quantum computation},}\
  }\href@noop {} {\bibfield  {journal} {\bibinfo  {journal} {Phys. Rev. A}\
  }\textbf {\bibinfo {volume} {97}},\ \bibinfo {pages} {052308} (\bibinfo
  {year} {2018})}\BibitemShut {NoStop}%
\bibitem [{\citenamefont {Takeuchi}\ \emph
  {et~al.}(2019{\natexlab{a}})\citenamefont {Takeuchi}, \citenamefont
  {Morimae},\ and\ \citenamefont {Hayashi}}]{TakeMH19}%
  \BibitemOpen
  \bibfield  {author} {\bibinfo {author} {\bibfnamefont {Y.}~\bibnamefont
  {Takeuchi}}, \bibinfo {author} {\bibfnamefont {T.}~\bibnamefont {Morimae}}, \
  and\ \bibinfo {author} {\bibfnamefont {M.}~\bibnamefont {Hayashi}},\
  }\bibfield  {title} {\enquote {\bibinfo {title} {{Quantum computational
  universality of hypergraph states with Pauli-X and Z basis measurements}},}\
  }\href@noop {} {\bibfield  {journal} {\bibinfo  {journal} {Sci. Rep.}\
  }\textbf {\bibinfo {volume} {9}},\ \bibinfo {pages} {13585} (\bibinfo {year}
  {2019}{\natexlab{a}})}\BibitemShut {NoStop}%
\bibitem [{\citenamefont {Miller}\ and\ \citenamefont
  {Miyake}(2016)}]{MillM16}%
  \BibitemOpen
  \bibfield  {author} {\bibinfo {author} {\bibfnamefont {J.}~\bibnamefont
  {Miller}}\ and\ \bibinfo {author} {\bibfnamefont {A.}~\bibnamefont
  {Miyake}},\ }\bibfield  {title} {\enquote {\bibinfo {title} {Hierarchy of
  universal entanglement in {2D} measurement-based quantum computation},}\
  }\href@noop {} {\bibfield  {journal} {\bibinfo  {journal} {npj Quantum Inf.}\
  }\textbf {\bibinfo {volume} {2}},\ \bibinfo {pages} {16036} (\bibinfo {year}
  {2016})}\BibitemShut {NoStop}%
\bibitem [{\citenamefont {Morimae}\ \emph {et~al.}(2017)\citenamefont
  {Morimae}, \citenamefont {Takeuchi},\ and\ \citenamefont
  {Hayashi}}]{MoriTH17}%
  \BibitemOpen
  \bibfield  {author} {\bibinfo {author} {\bibfnamefont {T.}~\bibnamefont
  {Morimae}}, \bibinfo {author} {\bibfnamefont {Y.}~\bibnamefont {Takeuchi}}, \
  and\ \bibinfo {author} {\bibfnamefont {M.}~\bibnamefont {Hayashi}},\
  }\bibfield  {title} {\enquote {\bibinfo {title} {Verification of hypergraph
  states},}\ }\href@noop {} {\bibfield  {journal} {\bibinfo  {journal} {Phys.
  Rev. A}\ }\textbf {\bibinfo {volume} {96}},\ \bibinfo {pages} {062321}
  (\bibinfo {year} {2017})}\BibitemShut {NoStop}%
\bibitem [{\citenamefont {Gachechiladze}\ \emph {et~al.}(2019)\citenamefont
  {Gachechiladze}, \citenamefont {G\"uhne},\ and\ \citenamefont
  {Miyake}}]{GachGM19}%
  \BibitemOpen
  \bibfield  {author} {\bibinfo {author} {\bibfnamefont {M.}~\bibnamefont
  {Gachechiladze}}, \bibinfo {author} {\bibfnamefont {O.}~\bibnamefont
  {G\"uhne}}, \ and\ \bibinfo {author} {\bibfnamefont {A.}~\bibnamefont
  {Miyake}},\ }\bibfield  {title} {\enquote {\bibinfo {title} {Changing the
  circuit-depth complexity of measurement-based quantum computation with
  hypergraph states},}\ }\href@noop {} {\bibfield  {journal} {\bibinfo
  {journal} {Phys. Rev. A}\ }\textbf {\bibinfo {volume} {99}},\ \bibinfo
  {pages} {052304} (\bibinfo {year} {2019})}\BibitemShut {NoStop}%
\bibitem [{\citenamefont {Gottesman}(1997)}]{Gott97the}%
  \BibitemOpen
  \bibfield  {author} {\bibinfo {author} {\bibfnamefont {D.}~\bibnamefont
  {Gottesman}},\ }\emph {\bibinfo {title} {Stabilizer Codes and Quantum Error
  Correction}},\ \href@noop {} {Ph.D. thesis},\ \bibinfo  {school} {California
  Institute of Technology} (\bibinfo {year} {1997}),\ \bibinfo {note}
  {available at \url{http://arxiv.org/abs/quant-ph/9705052}}\BibitemShut
  {NoStop}%
\bibitem [{\citenamefont {Schlingemann}\ and\ \citenamefont
  {Werner}(2001)}]{SchlW01}%
  \BibitemOpen
  \bibfield  {author} {\bibinfo {author} {\bibfnamefont {D.}~\bibnamefont
  {Schlingemann}}\ and\ \bibinfo {author} {\bibfnamefont {R.~F.}\ \bibnamefont
  {Werner}},\ }\bibfield  {title} {\enquote {\bibinfo {title} {Quantum
  error-correcting codes associated with graphs},}\ }\href@noop {} {\bibfield
  {journal} {\bibinfo  {journal} {Phys. Rev. A}\ }\textbf {\bibinfo {volume}
  {65}},\ \bibinfo {pages} {012308} (\bibinfo {year} {2001})}\BibitemShut
  {NoStop}%
\bibitem [{\citenamefont {Perseguers}\ \emph {et~al.}(2013)\citenamefont
  {Perseguers}, \citenamefont {Lapeyre~Jr}, \citenamefont {Cavalcanti},
  \citenamefont {Lewenstein},\ and\ \citenamefont {Ac\'in}}]{PersLCL13}%
  \BibitemOpen
  \bibfield  {author} {\bibinfo {author} {\bibfnamefont {S.}~\bibnamefont
  {Perseguers}}, \bibinfo {author} {\bibfnamefont {G.~J.}\ \bibnamefont
  {Lapeyre~Jr}}, \bibinfo {author} {\bibfnamefont {D.}~\bibnamefont
  {Cavalcanti}}, \bibinfo {author} {\bibfnamefont {M.}~\bibnamefont
  {Lewenstein}}, \ and\ \bibinfo {author} {\bibfnamefont {A.}~\bibnamefont
  {Ac\'in}},\ }\bibfield  {title} {\enquote {\bibinfo {title} {Distribution of
  entanglement in large-scale quantum networks},}\ }\href@noop {} {\bibfield
  {journal} {\bibinfo  {journal} {Rep. Prog. Phys.}\ }\textbf {\bibinfo
  {volume} {76}},\ \bibinfo {pages} {096001} (\bibinfo {year}
  {2013})}\BibitemShut {NoStop}%
\bibitem [{\citenamefont {McCutcheon}\ \emph {et~al.}(2016)\citenamefont
  {McCutcheon}, \citenamefont {Pappa}, \citenamefont {Bell}, \citenamefont
  {McMillan}, \citenamefont {Chailloux}, \citenamefont {Lawson}, \citenamefont
  {Mafu}, \citenamefont {Markham}, \citenamefont {Diamanti}, \citenamefont
  {Kerenidis}, \citenamefont {Rarity},\ and\ \citenamefont {Tame}}]{MccuPBM16}%
  \BibitemOpen
  \bibfield  {author} {\bibinfo {author} {\bibfnamefont {W.}~\bibnamefont
  {McCutcheon}}, \bibinfo {author} {\bibfnamefont {A.}~\bibnamefont {Pappa}},
  \bibinfo {author} {\bibfnamefont {B.~A.}\ \bibnamefont {Bell}}, \bibinfo
  {author} {\bibfnamefont {A.}~\bibnamefont {McMillan}}, \bibinfo {author}
  {\bibfnamefont {A.}~\bibnamefont {Chailloux}}, \bibinfo {author}
  {\bibfnamefont {T.}~\bibnamefont {Lawson}}, \bibinfo {author} {\bibfnamefont
  {M.}~\bibnamefont {Mafu}}, \bibinfo {author} {\bibfnamefont {D.}~\bibnamefont
  {Markham}}, \bibinfo {author} {\bibfnamefont {E.}~\bibnamefont {Diamanti}},
  \bibinfo {author} {\bibfnamefont {I.}~\bibnamefont {Kerenidis}}, \bibinfo
  {author} {\bibfnamefont {J.~G.}\ \bibnamefont {Rarity}}, \ and\ \bibinfo
  {author} {\bibfnamefont {M.~S.}\ \bibnamefont {Tame}},\ }\bibfield  {title}
  {\enquote {\bibinfo {title} {Experimental verification of multipartite
  entanglement in quantum networks},}\ }\href@noop {} {\bibfield  {journal}
  {\bibinfo  {journal} {Nat. Commun.}\ }\textbf {\bibinfo {volume} {7}},\
  \bibinfo {pages} {13251} (\bibinfo {year} {2016})}\BibitemShut {NoStop}%
\bibitem [{\citenamefont {Markham}\ and\ \citenamefont
  {Krause}(2018)}]{MarkK18}%
  \BibitemOpen
  \bibfield  {author} {\bibinfo {author} {\bibfnamefont {D.}~\bibnamefont
  {Markham}}\ and\ \bibinfo {author} {\bibfnamefont {A.}~\bibnamefont
  {Krause}},\ }\href@noop {} {\enquote {\bibinfo {title} {{A simple protocol
  for certifying graph states and applications in quantum networks}},}\ }
  (\bibinfo {year} {2018}),\ \Eprint {http://arxiv.org/abs/1801.05057}
  {arXiv:1801.05057} \BibitemShut {NoStop}%
\bibitem [{\citenamefont {Greenberger}\ \emph {et~al.}(1990)\citenamefont
  {Greenberger}, \citenamefont {Horne}, \citenamefont {Shimony},\ and\
  \citenamefont {Zeilinger}}]{GreeHSZ90}%
  \BibitemOpen
  \bibfield  {author} {\bibinfo {author} {\bibfnamefont {D.~M.}\ \bibnamefont
  {Greenberger}}, \bibinfo {author} {\bibfnamefont {M.~A.}\ \bibnamefont
  {Horne}}, \bibinfo {author} {\bibfnamefont {A.}~\bibnamefont {Shimony}}, \
  and\ \bibinfo {author} {\bibfnamefont {A.}~\bibnamefont {Zeilinger}},\
  }\bibfield  {title} {\enquote {\bibinfo {title} {Bell's theorem without
  inequalities},}\ }\href@noop {} {\bibfield  {journal} {\bibinfo  {journal}
  {American J. Phys.}\ }\textbf {\bibinfo {volume} {58}},\ \bibinfo {pages}
  {1131--1143} (\bibinfo {year} {1990})}\BibitemShut {NoStop}%
\bibitem [{\citenamefont {Scarani}\ \emph {et~al.}(2005)\citenamefont
  {Scarani}, \citenamefont {Ac\'{\i}n}, \citenamefont {Schenck},\ and\
  \citenamefont {Aspelmeyer}}]{ScarASA05}%
  \BibitemOpen
  \bibfield  {author} {\bibinfo {author} {\bibfnamefont {V.}~\bibnamefont
  {Scarani}}, \bibinfo {author} {\bibfnamefont {A.}~\bibnamefont {Ac\'{\i}n}},
  \bibinfo {author} {\bibfnamefont {E.}~\bibnamefont {Schenck}}, \ and\
  \bibinfo {author} {\bibfnamefont {M.}~\bibnamefont {Aspelmeyer}},\ }\bibfield
   {title} {\enquote {\bibinfo {title} {Nonlocality of cluster states of
  qubits},}\ }\href@noop {} {\bibfield  {journal} {\bibinfo  {journal} {Phys.
  Rev. A}\ }\textbf {\bibinfo {volume} {71}},\ \bibinfo {pages} {042325}
  (\bibinfo {year} {2005})}\BibitemShut {NoStop}%
\bibitem [{\citenamefont {G\"uhne}\ \emph {et~al.}(2005)\citenamefont
  {G\"uhne}, \citenamefont {T\'oth}, \citenamefont {Hyllus},\ and\
  \citenamefont {Briegel}}]{GuhnTHB05}%
  \BibitemOpen
  \bibfield  {author} {\bibinfo {author} {\bibfnamefont {O.}~\bibnamefont
  {G\"uhne}}, \bibinfo {author} {\bibfnamefont {G.}~\bibnamefont {T\'oth}},
  \bibinfo {author} {\bibfnamefont {P.}~\bibnamefont {Hyllus}}, \ and\ \bibinfo
  {author} {\bibfnamefont {H.~J.}\ \bibnamefont {Briegel}},\ }\bibfield
  {title} {\enquote {\bibinfo {title} {Bell inequalities for graph states},}\
  }\href@noop {} {\bibfield  {journal} {\bibinfo  {journal} {Phys. Rev. Lett.}\
  }\textbf {\bibinfo {volume} {95}},\ \bibinfo {pages} {120405} (\bibinfo
  {year} {2005})}\BibitemShut {NoStop}%
\bibitem [{\citenamefont {Gachechiladze}\ \emph {et~al.}(2016)\citenamefont
  {Gachechiladze}, \citenamefont {Budroni},\ and\ \citenamefont
  {G\"uhne}}]{GachBG16}%
  \BibitemOpen
  \bibfield  {author} {\bibinfo {author} {\bibfnamefont {M.}~\bibnamefont
  {Gachechiladze}}, \bibinfo {author} {\bibfnamefont {C.}~\bibnamefont
  {Budroni}}, \ and\ \bibinfo {author} {\bibfnamefont {O.}~\bibnamefont
  {G\"uhne}},\ }\bibfield  {title} {\enquote {\bibinfo {title} {Extreme
  violation of local realism in quantum hypergraph states},}\ }\href@noop {}
  {\bibfield  {journal} {\bibinfo  {journal} {Phys. Rev. Lett.}\ }\textbf
  {\bibinfo {volume} {116}},\ \bibinfo {pages} {070401} (\bibinfo {year}
  {2016})}\BibitemShut {NoStop}%
\bibitem [{\citenamefont {Dicke}(1954)}]{Dick54}%
  \BibitemOpen
  \bibfield  {author} {\bibinfo {author} {\bibfnamefont {R.~H.}\ \bibnamefont
  {Dicke}},\ }\bibfield  {title} {\enquote {\bibinfo {title} {Coherence in
  spontaneous radiation processes},}\ }\href@noop {} {\bibfield  {journal}
  {\bibinfo  {journal} {Phys. Rev.}\ }\textbf {\bibinfo {volume} {93}},\
  \bibinfo {pages} {99} (\bibinfo {year} {1954})}\BibitemShut {NoStop}%
\bibitem [{\citenamefont {H\"affner}\ \emph {et~al.}(2005)\citenamefont
  {H\"affner}, \citenamefont {H\"ansel}, \citenamefont {Roos}, \citenamefont
  {Benhelm}, \citenamefont {{Chek-al-kar}}, \citenamefont {Chwalla},
  \citenamefont {K\"orber}, \citenamefont {Rapol}, \citenamefont {Riebe},
  \citenamefont {Schmidt}, \citenamefont {Becher}, \citenamefont {G\"uhne},
  \citenamefont {D\"ur},\ and\ \citenamefont {Blatt}}]{HaffHRB05}%
  \BibitemOpen
  \bibfield  {author} {\bibinfo {author} {\bibfnamefont {H.}~\bibnamefont
  {H\"affner}}, \bibinfo {author} {\bibfnamefont {W.}~\bibnamefont {H\"ansel}},
  \bibinfo {author} {\bibfnamefont {C.~F.}\ \bibnamefont {Roos}}, \bibinfo
  {author} {\bibfnamefont {J.}~\bibnamefont {Benhelm}}, \bibinfo {author}
  {\bibfnamefont {D.}~\bibnamefont {{Chek-al-kar}}}, \bibinfo {author}
  {\bibfnamefont {M.}~\bibnamefont {Chwalla}}, \bibinfo {author} {\bibfnamefont
  {T.}~\bibnamefont {K\"orber}}, \bibinfo {author} {\bibfnamefont {U.~D.}\
  \bibnamefont {Rapol}}, \bibinfo {author} {\bibfnamefont {M.}~\bibnamefont
  {Riebe}}, \bibinfo {author} {\bibfnamefont {P.~O.}\ \bibnamefont {Schmidt}},
  \bibinfo {author} {\bibfnamefont {C.}~\bibnamefont {Becher}}, \bibinfo
  {author} {\bibfnamefont {O.}~\bibnamefont {G\"uhne}}, \bibinfo {author}
  {\bibfnamefont {W.}~\bibnamefont {D\"ur}}, \ and\ \bibinfo {author}
  {\bibfnamefont {R.}~\bibnamefont {Blatt}},\ }\bibfield  {title} {\enquote
  {\bibinfo {title} {Scalable multiparticle entanglement of trapped ions},}\
  }\href@noop {} {\bibfield  {journal} {\bibinfo  {journal} {Nature}\ }\textbf
  {\bibinfo {volume} {438}},\ \bibinfo {pages} {643--646} (\bibinfo {year}
  {2005})}\BibitemShut {NoStop}%
\bibitem [{\citenamefont {Pezz\`e}\ \emph {et~al.}(2018)\citenamefont
  {Pezz\`e}, \citenamefont {Smerzi}, \citenamefont {Oberthaler}, \citenamefont
  {Schmied},\ and\ \citenamefont {Treutlein}}]{PezzSOS18}%
  \BibitemOpen
  \bibfield  {author} {\bibinfo {author} {\bibfnamefont {L.}~\bibnamefont
  {Pezz\`e}}, \bibinfo {author} {\bibfnamefont {A.}~\bibnamefont {Smerzi}},
  \bibinfo {author} {\bibfnamefont {M.~K.}\ \bibnamefont {Oberthaler}},
  \bibinfo {author} {\bibfnamefont {R.}~\bibnamefont {Schmied}}, \ and\
  \bibinfo {author} {\bibfnamefont {P.}~\bibnamefont {Treutlein}},\ }\bibfield
  {title} {\enquote {\bibinfo {title} {Quantum metrology with nonclassical
  states of atomic ensembles},}\ }\href {\doibase 10.1103/RevModPhys.90.035005}
  {\bibfield  {journal} {\bibinfo  {journal} {Rev. Mod. Phys.}\ }\textbf
  {\bibinfo {volume} {90}},\ \bibinfo {pages} {035005} (\bibinfo {year}
  {2018})}\BibitemShut {NoStop}%
\bibitem [{\citenamefont {Verstraete}\ \emph {et~al.}(2008)\citenamefont
  {Verstraete}, \citenamefont {Murg},\ and\ \citenamefont {Cirac}}]{VersMC08}%
  \BibitemOpen
  \bibfield  {author} {\bibinfo {author} {\bibfnamefont {F.}~\bibnamefont
  {Verstraete}}, \bibinfo {author} {\bibfnamefont {V.}~\bibnamefont {Murg}}, \
  and\ \bibinfo {author} {\bibfnamefont {J.~I.}\ \bibnamefont {Cirac}},\
  }\bibfield  {title} {\enquote {\bibinfo {title} {Matrix product states,
  projected entangled pair states, and variational renormalization group
  methods for quantum spin systems},}\ }\href@noop {} {\bibfield  {journal}
  {\bibinfo  {journal} {Adv. Phys.}\ }\textbf {\bibinfo {volume} {57}},\
  \bibinfo {pages} {143--224} (\bibinfo {year} {2008})}\BibitemShut {NoStop}%
\bibitem [{\citenamefont {Or\'us}(2014)}]{Orus14}%
  \BibitemOpen
  \bibfield  {author} {\bibinfo {author} {\bibfnamefont {R.}~\bibnamefont
  {Or\'us}},\ }\bibfield  {title} {\enquote {\bibinfo {title} {A practical
  introduction to tensor networks: Matrix product states and projected
  entangled pair states},}\ }\href {\doibase
  https://doi.org/10.1016/j.aop.2014.06.013} {\bibfield  {journal} {\bibinfo
  {journal} {Ann. Phys.}\ }\textbf {\bibinfo {volume} {349}},\ \bibinfo {pages}
  {117--158} (\bibinfo {year} {2014})}\BibitemShut {NoStop}%
\bibitem [{\citenamefont {Paris}\ and\ \citenamefont
  {{\v{R}}eh\'{a}\v{c}ek}(2004)}]{PariR04Book}%
  \BibitemOpen
  \bibinfo {editor} {\bibfnamefont {M.~G.~A.}\ \bibnamefont {Paris}}\ and\
  \bibinfo {editor} {\bibfnamefont {J.}~\bibnamefont {{\v{R}}eh\'{a}\v{c}ek}},\
  eds.,\ \href@noop {} {\emph {\bibinfo {title} {Quantum State Estimation}}},\
  \bibinfo {series} {Lecture Notes in Physics}, Vol.\ \bibinfo {volume} {649}\
  (\bibinfo  {publisher} {Springer},\ \bibinfo {address} {Berlin},\ \bibinfo
  {year} {2004})\BibitemShut {NoStop}%
\bibitem [{\citenamefont {Gross}\ \emph {et~al.}(2010)\citenamefont {Gross},
  \citenamefont {Liu}, \citenamefont {Flammia}, \citenamefont {Becker},\ and\
  \citenamefont {Eisert}}]{GrosLFB10}%
  \BibitemOpen
  \bibfield  {author} {\bibinfo {author} {\bibfnamefont {D.}~\bibnamefont
  {Gross}}, \bibinfo {author} {\bibfnamefont {Y.-K.}\ \bibnamefont {Liu}},
  \bibinfo {author} {\bibfnamefont {S.~T.}\ \bibnamefont {Flammia}}, \bibinfo
  {author} {\bibfnamefont {S.}~\bibnamefont {Becker}}, \ and\ \bibinfo {author}
  {\bibfnamefont {J.}~\bibnamefont {Eisert}},\ }\bibfield  {title} {\enquote
  {\bibinfo {title} {Quantum state tomography via compressed sensing},}\
  }\href@noop {} {\bibfield  {journal} {\bibinfo  {journal} {Phys. Rev. Lett.}\
  }\textbf {\bibinfo {volume} {105}},\ \bibinfo {pages} {150401} (\bibinfo
  {year} {2010})}\BibitemShut {NoStop}%
\bibitem [{\citenamefont {Flammia}\ and\ \citenamefont {Liu}(2011)}]{FlamL11}%
  \BibitemOpen
  \bibfield  {author} {\bibinfo {author} {\bibfnamefont {S.~T.}\ \bibnamefont
  {Flammia}}\ and\ \bibinfo {author} {\bibfnamefont {Y.-K.}\ \bibnamefont
  {Liu}},\ }\bibfield  {title} {\enquote {\bibinfo {title} {Direct fidelity
  estimation from few \uppercase{P}auli measurements},}\ }\href@noop {}
  {\bibfield  {journal} {\bibinfo  {journal} {Phys. Rev. Lett.}\ }\textbf
  {\bibinfo {volume} {106}},\ \bibinfo {pages} {230501} (\bibinfo {year}
  {2011})}\BibitemShut {NoStop}%
\bibitem [{\citenamefont {Mayers}\ and\ \citenamefont {Yao}(2004)}]{MayeY04}%
  \BibitemOpen
  \bibfield  {author} {\bibinfo {author} {\bibfnamefont {D.}~\bibnamefont
  {Mayers}}\ and\ \bibinfo {author} {\bibfnamefont {A.}~\bibnamefont {Yao}},\
  }\bibfield  {title} {\enquote {\bibinfo {title} {Self testing quantum
  apparatus},}\ }\href@noop {} {\bibfield  {journal} {\bibinfo  {journal}
  {Quantum Info. Comput.}\ }\textbf {\bibinfo {volume} {4}},\ \bibinfo {pages}
  {273--286} (\bibinfo {year} {2004})}\BibitemShut {NoStop}%
\bibitem [{\citenamefont {\v{S}upi\'{c}}\ and\ \citenamefont
  {Bowles}(2019)}]{SupiB19}%
  \BibitemOpen
  \bibfield  {author} {\bibinfo {author} {\bibfnamefont {I.}~\bibnamefont
  {\v{S}upi\'{c}}}\ and\ \bibinfo {author} {\bibfnamefont {J.}~\bibnamefont
  {Bowles}},\ }\href@noop {} {\enquote {\bibinfo {title} {{Self-testing of
  quantum systems: a review}},}\ } (\bibinfo {year} {2019}),\ \Eprint
  {http://arxiv.org/abs/1904.10042} {arXiv:1904.10042} \BibitemShut {NoStop}%
\bibitem [{\citenamefont {Hayashi}\ \emph {et~al.}(2006)\citenamefont
  {Hayashi}, \citenamefont {Matsumoto},\ and\ \citenamefont
  {Tsuda}}]{HayaMT06}%
  \BibitemOpen
  \bibfield  {author} {\bibinfo {author} {\bibfnamefont {M.}~\bibnamefont
  {Hayashi}}, \bibinfo {author} {\bibfnamefont {K.}~\bibnamefont {Matsumoto}},
  \ and\ \bibinfo {author} {\bibfnamefont {Y.}~\bibnamefont {Tsuda}},\
  }\bibfield  {title} {\enquote {\bibinfo {title} {A study of {LOCC}-detection
  of a maximally entangled state using hypothesis testing},}\ }\href@noop {}
  {\bibfield  {journal} {\bibinfo  {journal} {J. Phys. A: Math. Gen.}\ }\textbf
  {\bibinfo {volume} {39}},\ \bibinfo {pages} {14427} (\bibinfo {year}
  {2006})}\BibitemShut {NoStop}%
\bibitem [{\citenamefont {Hayashi}(2009)}]{Haya09G}%
  \BibitemOpen
  \bibfield  {author} {\bibinfo {author} {\bibfnamefont {M.}~\bibnamefont
  {Hayashi}},\ }\bibfield  {title} {\enquote {\bibinfo {title} {Group
  theoretical study of {LOCC}-detection of maximally entangled states using
  hypothesis testing},}\ }\href@noop {} {\bibfield  {journal} {\bibinfo
  {journal} {New J. Phys.}\ }\textbf {\bibinfo {volume} {11}},\ \bibinfo
  {pages} {043028} (\bibinfo {year} {2009})}\BibitemShut {NoStop}%
\bibitem [{\citenamefont {Aolita}\ \emph {et~al.}(2015)\citenamefont {Aolita},
  \citenamefont {Gogolin}, \citenamefont {Kliesch},\ and\ \citenamefont
  {Eisert}}]{AoliGKE15}%
  \BibitemOpen
  \bibfield  {author} {\bibinfo {author} {\bibfnamefont {L.}~\bibnamefont
  {Aolita}}, \bibinfo {author} {\bibfnamefont {C.}~\bibnamefont {Gogolin}},
  \bibinfo {author} {\bibfnamefont {M.}~\bibnamefont {Kliesch}}, \ and\
  \bibinfo {author} {\bibfnamefont {J.}~\bibnamefont {Eisert}},\ }\bibfield
  {title} {\enquote {\bibinfo {title} {Reliable quantum certification of
  photonic state preparations},}\ }\href@noop {} {\bibfield  {journal}
  {\bibinfo  {journal} {Nat. Commun.}\ }\textbf {\bibinfo {volume} {6}},\
  \bibinfo {pages} {8498} (\bibinfo {year} {2015})}\BibitemShut {NoStop}%
\bibitem [{\citenamefont {Takeuchi}\ and\ \citenamefont
  {Morimae}(2018)}]{TakeM18}%
  \BibitemOpen
  \bibfield  {author} {\bibinfo {author} {\bibfnamefont {Y.}~\bibnamefont
  {Takeuchi}}\ and\ \bibinfo {author} {\bibfnamefont {T.}~\bibnamefont
  {Morimae}},\ }\bibfield  {title} {\enquote {\bibinfo {title} {Verification of
  many-qubit states},}\ }\href@noop {} {\bibfield  {journal} {\bibinfo
  {journal} {Phys. Rev. X}\ }\textbf {\bibinfo {volume} {8}},\ \bibinfo {pages}
  {021060} (\bibinfo {year} {2018})}\BibitemShut {NoStop}%
\bibitem [{\citenamefont {Pallister}\ \emph {et~al.}(2018)\citenamefont
  {Pallister}, \citenamefont {Linden},\ and\ \citenamefont
  {Montanaro}}]{PallLM18}%
  \BibitemOpen
  \bibfield  {author} {\bibinfo {author} {\bibfnamefont {S.}~\bibnamefont
  {Pallister}}, \bibinfo {author} {\bibfnamefont {N.}~\bibnamefont {Linden}}, \
  and\ \bibinfo {author} {\bibfnamefont {A.}~\bibnamefont {Montanaro}},\
  }\bibfield  {title} {\enquote {\bibinfo {title} {Optimal verification of
  entangled states with local measurements},}\ }\href@noop {} {\bibfield
  {journal} {\bibinfo  {journal} {Phys. Rev. Lett.}\ }\textbf {\bibinfo
  {volume} {120}},\ \bibinfo {pages} {170502} (\bibinfo {year}
  {2018})}\BibitemShut {NoStop}%
\bibitem [{\citenamefont {Zhu}\ and\ \citenamefont
  {Hayashi}(2019{\natexlab{a}})}]{ZhuH19O}%
  \BibitemOpen
  \bibfield  {author} {\bibinfo {author} {\bibfnamefont {H.}~\bibnamefont
  {Zhu}}\ and\ \bibinfo {author} {\bibfnamefont {M.}~\bibnamefont {Hayashi}},\
  }\bibfield  {title} {\enquote {\bibinfo {title} {Optimal verification and
  fidelity estimation of maximally entangled states},}\ }\href@noop {}
  {\bibfield  {journal} {\bibinfo  {journal} {Phys. Rev. A}\ }\textbf {\bibinfo
  {volume} {99}},\ \bibinfo {pages} {052346} (\bibinfo {year}
  {2019}{\natexlab{a}})}\BibitemShut {NoStop}%
\bibitem [{\citenamefont {Li}\ \emph {et~al.}(2019{\natexlab{a}})\citenamefont
  {Li}, \citenamefont {Han},\ and\ \citenamefont {Zhu}}]{LiHZ19}%
  \BibitemOpen
  \bibfield  {author} {\bibinfo {author} {\bibfnamefont {Z.}~\bibnamefont
  {Li}}, \bibinfo {author} {\bibfnamefont {Y.-G.}\ \bibnamefont {Han}}, \ and\
  \bibinfo {author} {\bibfnamefont {H.}~\bibnamefont {Zhu}},\ }\bibfield
  {title} {\enquote {\bibinfo {title} {Efficient verification of bipartite pure
  states},}\ }\href@noop {} {\bibfield  {journal} {\bibinfo  {journal} {Phys.
  Rev. A}\ }\textbf {\bibinfo {volume} {100}},\ \bibinfo {pages} {032316}
  (\bibinfo {year} {2019}{\natexlab{a}})}\BibitemShut {NoStop}%
\bibitem [{\citenamefont {Wang}\ and\ \citenamefont {Hayashi}(2019)}]{WangH19}%
  \BibitemOpen
  \bibfield  {author} {\bibinfo {author} {\bibfnamefont {K.}~\bibnamefont
  {Wang}}\ and\ \bibinfo {author} {\bibfnamefont {M.}~\bibnamefont {Hayashi}},\
  }\bibfield  {title} {\enquote {\bibinfo {title} {Optimal verification of
  two-qubit pure states},}\ }\href@noop {} {\bibfield  {journal} {\bibinfo
  {journal} {Phys. Rev. A}\ }\textbf {\bibinfo {volume} {100}},\ \bibinfo
  {pages} {032315} (\bibinfo {year} {2019})}\BibitemShut {NoStop}%
\bibitem [{\citenamefont {Yu}\ \emph {et~al.}(2019)\citenamefont {Yu},
  \citenamefont {Shang},\ and\ \citenamefont {G\"uhne}}]{YuSG19}%
  \BibitemOpen
  \bibfield  {author} {\bibinfo {author} {\bibfnamefont {X.-D.}\ \bibnamefont
  {Yu}}, \bibinfo {author} {\bibfnamefont {J.}~\bibnamefont {Shang}}, \ and\
  \bibinfo {author} {\bibfnamefont {O.}~\bibnamefont {G\"uhne}},\ }\bibfield
  {title} {\enquote {\bibinfo {title} {{Optimal verification of general
  bipartite pure states}},}\ }\href@noop {} {\bibfield  {journal} {\bibinfo
  {journal} {npj Quantum Inf.}\ }\textbf {\bibinfo {volume} {5}},\ \bibinfo
  {pages} {112} (\bibinfo {year} {2019})}\BibitemShut {NoStop}%
\bibitem [{\citenamefont {Li}\ \emph {et~al.}(2019{\natexlab{b}})\citenamefont
  {Li}, \citenamefont {Han},\ and\ \citenamefont {Zhu}}]{LiHZ19O}%
  \BibitemOpen
  \bibfield  {author} {\bibinfo {author} {\bibfnamefont {Z.}~\bibnamefont
  {Li}}, \bibinfo {author} {\bibfnamefont {Y.-G.}\ \bibnamefont {Han}}, \ and\
  \bibinfo {author} {\bibfnamefont {H.}~\bibnamefont {Zhu}},\ }\href@noop {}
  {\enquote {\bibinfo {title} {{Optimal Verification of
  Greenberger-Horne-Zeilinger States}},}\ } (\bibinfo {year}
  {2019}{\natexlab{b}}),\ \bibinfo {note} {arXiv:1909.08979}\BibitemShut
  {NoStop}%
\bibitem [{\citenamefont {Zhu}\ and\ \citenamefont
  {Hayashi}(2019{\natexlab{b}})}]{ZhuH19E}%
  \BibitemOpen
  \bibfield  {author} {\bibinfo {author} {\bibfnamefont {H.}~\bibnamefont
  {Zhu}}\ and\ \bibinfo {author} {\bibfnamefont {M.}~\bibnamefont {Hayashi}},\
  }\bibfield  {title} {\enquote {\bibinfo {title} {Efficient verification of
  hypergraph states},}\ }\href@noop {} {\bibfield  {journal} {\bibinfo
  {journal} {Phys. Rev. Applied}\ }\textbf {\bibinfo {volume} {12}},\ \bibinfo
  {pages} {054047} (\bibinfo {year} {2019}{\natexlab{b}})}\BibitemShut
  {NoStop}%
\bibitem [{\citenamefont {Hayashi}\ and\ \citenamefont
  {Takeuchi}(2019)}]{HayaT19}%
  \BibitemOpen
  \bibfield  {author} {\bibinfo {author} {\bibfnamefont {M.}~\bibnamefont
  {Hayashi}}\ and\ \bibinfo {author} {\bibfnamefont {Y.}~\bibnamefont
  {Takeuchi}},\ }\bibfield  {title} {\enquote {\bibinfo {title} {Verifying
  commuting quantum computations via fidelity estimation of weighted graph
  states},}\ }\href@noop {} {\bibfield  {journal} {\bibinfo  {journal} {New J.
  Phys.}\ }\textbf {\bibinfo {volume} {21}},\ \bibinfo {pages} {093060}
  (\bibinfo {year} {2019})}\BibitemShut {NoStop}%
\bibitem [{\citenamefont {Liu}\ \emph {et~al.}(2019{\natexlab{a}})\citenamefont
  {Liu}, \citenamefont {Yu}, \citenamefont {Shang}, \citenamefont {Zhu},\ and\
  \citenamefont {Zhang}}]{LiuYSZ19}%
  \BibitemOpen
  \bibfield  {author} {\bibinfo {author} {\bibfnamefont {Y.-C.}\ \bibnamefont
  {Liu}}, \bibinfo {author} {\bibfnamefont {X.-D.}\ \bibnamefont {Yu}},
  \bibinfo {author} {\bibfnamefont {J.}~\bibnamefont {Shang}}, \bibinfo
  {author} {\bibfnamefont {H.}~\bibnamefont {Zhu}}, \ and\ \bibinfo {author}
  {\bibfnamefont {X.}~\bibnamefont {Zhang}},\ }\bibfield  {title} {\enquote
  {\bibinfo {title} {Efficient verification of {Dicke} states},}\ }\href@noop
  {} {\bibfield  {journal} {\bibinfo  {journal} {Phys. Rev. Applied}\ }\textbf
  {\bibinfo {volume} {12}},\ \bibinfo {pages} {044020} (\bibinfo {year}
  {2019}{\natexlab{a}})}\BibitemShut {NoStop}%
\bibitem [{\citenamefont {Zhang}\ \emph {et~al.}(2019)\citenamefont {Zhang},
  \citenamefont {Chen}, \citenamefont {Peng}, \citenamefont {Xu}, \citenamefont
  {Yin}, \citenamefont {Ye}, \citenamefont {Xu}, \citenamefont {Chen},
  \citenamefont {Li},\ and\ \citenamefont {Guo}}]{ZhanCPX19}%
  \BibitemOpen
  \bibfield  {author} {\bibinfo {author} {\bibfnamefont {W.-H.}\ \bibnamefont
  {Zhang}}, \bibinfo {author} {\bibfnamefont {Z.}~\bibnamefont {Chen}},
  \bibinfo {author} {\bibfnamefont {X.-X.}\ \bibnamefont {Peng}}, \bibinfo
  {author} {\bibfnamefont {X.-Y.}\ \bibnamefont {Xu}}, \bibinfo {author}
  {\bibfnamefont {P.}~\bibnamefont {Yin}}, \bibinfo {author} {\bibfnamefont
  {X.-J.}\ \bibnamefont {Ye}}, \bibinfo {author} {\bibfnamefont {J.-S.}\
  \bibnamefont {Xu}}, \bibinfo {author} {\bibfnamefont {G.}~\bibnamefont
  {Chen}}, \bibinfo {author} {\bibfnamefont {C.-F.}\ \bibnamefont {Li}}, \ and\
  \bibinfo {author} {\bibfnamefont {G.-C.}\ \bibnamefont {Guo}},\ }\href@noop
  {} {\enquote {\bibinfo {title} {{Experimental Optimal Verification of
  Entangled States using Local Measurements}},}\ } (\bibinfo {year} {2019}),\
  \bibinfo {note} {arXiv:1905.12175}\BibitemShut {NoStop}%
\bibitem [{\citenamefont {Takeuchi}\ \emph
  {et~al.}(2019{\natexlab{b}})\citenamefont {Takeuchi}, \citenamefont {Mantri},
  \citenamefont {Morimae}, \citenamefont {Mizutani},\ and\ \citenamefont
  {Fitzsimons}}]{TakeMMM19}%
  \BibitemOpen
  \bibfield  {author} {\bibinfo {author} {\bibfnamefont {Y.}~\bibnamefont
  {Takeuchi}}, \bibinfo {author} {\bibfnamefont {A.}~\bibnamefont {Mantri}},
  \bibinfo {author} {\bibfnamefont {T.}~\bibnamefont {Morimae}}, \bibinfo
  {author} {\bibfnamefont {A.}~\bibnamefont {Mizutani}}, \ and\ \bibinfo
  {author} {\bibfnamefont {J.~F.}\ \bibnamefont {Fitzsimons}},\ }\bibfield
  {title} {\enquote {\bibinfo {title} {{Resource-efficient verification of
  quantum computing using Serfling's bound}},}\ }\href@noop {} {\bibfield
  {journal} {\bibinfo  {journal} {npj Quantum Inf.}\ }\textbf {\bibinfo
  {volume} {5}},\ \bibinfo {pages} {27} (\bibinfo {year}
  {2019}{\natexlab{b}})}\BibitemShut {NoStop}%
\bibitem [{\citenamefont {Zhu}\ and\ \citenamefont
  {Hayashi}(2019{\natexlab{c}})}]{ZhuH19AdS}%
  \BibitemOpen
  \bibfield  {author} {\bibinfo {author} {\bibfnamefont {H.}~\bibnamefont
  {Zhu}}\ and\ \bibinfo {author} {\bibfnamefont {M.}~\bibnamefont {Hayashi}},\
  }\bibfield  {title} {\enquote {\bibinfo {title} {Efficient verification of
  pure quantum states in the adversarial scenario},}\ }\href@noop {} {\bibfield
   {journal} {\bibinfo  {journal} {Phys. Rev. Lett.}\ }\textbf {\bibinfo
  {volume} {123}},\ \bibinfo {pages} {260504} (\bibinfo {year}
  {2019}{\natexlab{c}})}\BibitemShut {NoStop}%
\bibitem [{\citenamefont {Dimi\'c}\ and\ \citenamefont
  {Daki\'c}(2018)}]{DimiD18}%
  \BibitemOpen
  \bibfield  {author} {\bibinfo {author} {\bibfnamefont {A.}~\bibnamefont
  {Dimi\'c}}\ and\ \bibinfo {author} {\bibfnamefont {B.}~\bibnamefont
  {Daki\'c}},\ }\bibfield  {title} {\enquote {\bibinfo {title} {Single-copy
  entanglement detection},}\ }\href@noop {} {\bibfield  {journal} {\bibinfo
  {journal} {npj Quantum Inf.}\ }\textbf {\bibinfo {volume} {4}},\ \bibinfo
  {pages} {11} (\bibinfo {year} {2018})}\BibitemShut {NoStop}%
\bibitem [{\citenamefont {Zhu}\ and\ \citenamefont {Zhang}(2019)}]{ZhuZ19}%
  \BibitemOpen
  \bibfield  {author} {\bibinfo {author} {\bibfnamefont {H.}~\bibnamefont
  {Zhu}}\ and\ \bibinfo {author} {\bibfnamefont {H.}~\bibnamefont {Zhang}},\
  }\href@noop {} {\enquote {\bibinfo {title} {Efficient verification of quantum
  gates with local operations},}\ } (\bibinfo {year} {2019}),\ \bibinfo {note}
  {arXiv:1910.14032}\BibitemShut {NoStop}%
\bibitem [{\citenamefont {Liu}\ \emph {et~al.}(2019{\natexlab{b}})\citenamefont
  {Liu}, \citenamefont {Shang}, \citenamefont {Yu},\ and\ \citenamefont
  {Zhang}}]{LiuSYZ19}%
  \BibitemOpen
  \bibfield  {author} {\bibinfo {author} {\bibfnamefont {Y.-C.}\ \bibnamefont
  {Liu}}, \bibinfo {author} {\bibfnamefont {J.}~\bibnamefont {Shang}}, \bibinfo
  {author} {\bibfnamefont {X.-D.}\ \bibnamefont {Yu}}, \ and\ \bibinfo {author}
  {\bibfnamefont {X.}~\bibnamefont {Zhang}},\ }\href@noop {} {\enquote
  {\bibinfo {title} {Efficient and practical verification of quantum
  processes},}\ } (\bibinfo {year} {2019}{\natexlab{b}}),\ \bibinfo {note}
  {arXiv:1910.13730}\BibitemShut {NoStop}%
\bibitem [{\citenamefont {Coladangelo}\ \emph {et~al.}(2017)\citenamefont
  {Coladangelo}, \citenamefont {Goh},\ and\ \citenamefont
  {Scarani}}]{ColaTS17}%
  \BibitemOpen
  \bibfield  {author} {\bibinfo {author} {\bibfnamefont {A.}~\bibnamefont
  {Coladangelo}}, \bibinfo {author} {\bibfnamefont {K.~T.}\ \bibnamefont
  {Goh}}, \ and\ \bibinfo {author} {\bibfnamefont {V.}~\bibnamefont
  {Scarani}},\ }\bibfield  {title} {\enquote {\bibinfo {title} {All pure
  bipartite entangled states can be self-tested},}\ }\href@noop {} {\bibfield
  {journal} {\bibinfo  {journal} {Nat. Commun.}\ }\textbf {\bibinfo {volume}
  {8}},\ \bibinfo {pages} {15485} (\bibinfo {year} {2017})}\BibitemShut
  {NoStop}%
\bibitem [{\citenamefont {Schlingemann}(2002)}]{Schl02}%
  \BibitemOpen
  \bibfield  {author} {\bibinfo {author} {\bibfnamefont {D.}~\bibnamefont
  {Schlingemann}},\ }\bibfield  {title} {\enquote {\bibinfo {title} {Stabilizer
  codes can be realized as graph codes},}\ }\href@noop {} {\bibfield  {journal}
  {\bibinfo  {journal} {Quantum Info. Comput.}\ }\textbf {\bibinfo {volume}
  {2}},\ \bibinfo {pages} {307--323} (\bibinfo {year} {2002})}\BibitemShut
  {NoStop}%
\bibitem [{\citenamefont {Grassl}\ \emph {et~al.}(2002)\citenamefont {Grassl},
  \citenamefont {Klappenecker},\ and\ \citenamefont {R\"otteler}}]{GrasKR02}%
  \BibitemOpen
  \bibfield  {author} {\bibinfo {author} {\bibfnamefont {M.}~\bibnamefont
  {Grassl}}, \bibinfo {author} {\bibfnamefont {A.}~\bibnamefont
  {Klappenecker}}, \ and\ \bibinfo {author} {\bibfnamefont {M.}~\bibnamefont
  {R\"otteler}},\ }\bibfield  {title} {\enquote {\bibinfo {title} {Graphs,
  quadratic forms, and quantum codes},}\ }in\ \href@noop {} {\emph {\bibinfo
  {booktitle} {Proceedings of the 2002 IEEE International Symposium on
  Information Theory}}}\ (\bibinfo {address} {IEEE Information Theory Society,
  Lausanne, Switzerland},\ \bibinfo {year} {2002})\ \bibinfo {note} {available
  at arXiv:quant-ph/0703112}\BibitemShut {NoStop}%
\bibitem [{\citenamefont {Hartmann}\ \emph {et~al.}(2007)\citenamefont
  {Hartmann}, \citenamefont {Calsamiglia}, \citenamefont {D\"ur},\ and\
  \citenamefont {Briegel}}]{HartCDB07}%
  \BibitemOpen
  \bibfield  {author} {\bibinfo {author} {\bibfnamefont {L.}~\bibnamefont
  {Hartmann}}, \bibinfo {author} {\bibfnamefont {J.}~\bibnamefont
  {Calsamiglia}}, \bibinfo {author} {\bibfnamefont {W.}~\bibnamefont {D\"ur}},
  \ and\ \bibinfo {author} {\bibfnamefont {H.~J.}\ \bibnamefont {Briegel}},\
  }\bibfield  {title} {\enquote {\bibinfo {title} {Weighted graph states and
  applications to spin chains, lattices and gases},}\ }\href@noop {} {\bibfield
   {journal} {\bibinfo  {journal} {J. Phys. B: At. Mol. Opt. Phys.}\ }\textbf
  {\bibinfo {volume} {40}},\ \bibinfo {pages} {S1--S44} (\bibinfo {year}
  {2007})}\BibitemShut {NoStop}%
\bibitem [{\citenamefont {{\v{S}}upi{\'{c}}}\ \emph {et~al.}(2018)\citenamefont
  {{\v{S}}upi{\'{c}}}, \citenamefont {Coladangelo}, \citenamefont {Augusiak},\
  and\ \citenamefont {Ac{\'{\i}}n}}]{SupiCAA18}%
  \BibitemOpen
  \bibfield  {author} {\bibinfo {author} {\bibfnamefont {I.}~\bibnamefont
  {{\v{S}}upi{\'{c}}}}, \bibinfo {author} {\bibfnamefont {A.}~\bibnamefont
  {Coladangelo}}, \bibinfo {author} {\bibfnamefont {R.}~\bibnamefont
  {Augusiak}}, \ and\ \bibinfo {author} {\bibfnamefont {A.}~\bibnamefont
  {Ac{\'{\i}}n}},\ }\bibfield  {title} {\enquote {\bibinfo {title}
  {Self-testing multipartite entangled states through projections onto two
  systems},}\ }\href@noop {} {\bibfield  {journal} {\bibinfo  {journal} {New J.
  Phys.}\ }\textbf {\bibinfo {volume} {20}},\ \bibinfo {pages} {083041}
  (\bibinfo {year} {2018})}\BibitemShut {NoStop}%
\bibitem [{\citenamefont {Fadel}(2017)}]{Fade17}%
  \BibitemOpen
  \bibfield  {author} {\bibinfo {author} {\bibfnamefont {M.}~\bibnamefont
  {Fadel}},\ }\href@noop {} {\enquote {\bibinfo {title} {Self-testing {Dicke}
  states},}\ } (\bibinfo {year} {2017}),\ \bibinfo {note}
  {arXiv:1707.01215}\BibitemShut {NoStop}%
\end{thebibliography}%

\end{document}